\documentclass[11pt]{amsart}
\usepackage[left=1.5in, right=1.5in]{geometry}
\usepackage[latin1]{inputenc}
\usepackage{amsfonts}
\usepackage[toc,page,title,titletoc,header]{appendix}
\usepackage{graphicx,psfrag,epsfig,multirow,caption,subcaption}
\usepackage{amssymb,amsmath,amscd,amsthm,amssymb,verbatim,setspace}
\numberwithin{equation}{section}
\usepackage{mathrsfs,wrapfig}
\usepackage{indentfirst}
\usepackage{extarrows}
\usepackage[colorlinks=true]{hyperref}

\newtheorem{Thm}{Theorem}[section]
\newtheorem{Def}[Thm]{Definition}

\newtheorem{Rk}[Thm]{Remark}

\newtheorem{Con}[Thm]{Conjecture}
\newtheorem{Not}[Thm]{Notation}
\def\bmt{\left[\begin{array}}
\def\emt{\end{array}\right]}

\def\Z{\mathbb Z}
\def\T{\mathbb T}
\def\C{\mathbb C}

\def\R{\mathbb R}
\def\eps{\varepsilon}
\def\al{\alpha}

\def\dt{\delta}

\def\lb{\lambda}

\title[Arnold diffusion and blackholes]{Arnold diffusion and geodesic dynamics of blackholes}
\author{Jinxin Xue}
\date\today
\begin{document}
\address{Jingzhai 310, Department of mathematics and Yau Mathematical Sciences Center, Tsinghua University, Beijing,100084 }
\email{jxue@tsinghua.edu.cn}%
\maketitle
\begin{abstract}
In this paper, we study the chaotic motion of a massive particle moving in a perturbed Schwarzschild or Kerr background.  We discover three novel orbits that do not exist in the unperturbed cases. First, we find zoom-whirl orbits moving around the photon shell which simultaneously exhibits Arnold diffusion: large oscillations of particle's angular momentum and energy.  Next, we show the existence of oscillating orbits between a bounded region and infinity, analogous to Newtonian three-body problem. Thirdly, we find that in perturbed Kerr, there exists chaotic orbits around the event horizon that escapes the event horizon after approaching it.

%This article is an attempt of applying the theory of advanced Hamiltonian dynamics to the study of the blackhole dynamics. The main goal is to explore the instability behaviors of dynamics of a particle moving in perturbed Schwarzschild and Kerr space-times.

%We reveal a link between Arnold diffusion and Penrose process.
% quasi-periodic oscillation (QPO) regime and
\end{abstract}

\tableofcontents
%\end{spacing}
\renewcommand\contentsname{Index}

\section{Introduction}

The main theme of this paper is to study the dynamics of a particle moving in a blackhole background. The most important blackhole models are Schwarzschild and Kerr metrics. They are integrable in the sense that we can find four constants of motions for either of them and the equations of motion can be integrated. In reality, the spacetime is always perturbed, so it is reasonable to perform a perturbative analysis for the dynamics of particle moving in a perturbed Schwarzschild and Kerr background. Ignoring the radiation effects, we model such systems as the nearly integrable Hamiltonian systems, which has a very well developed theory primarily for the purpose of studying the Newtonian $N$-body problem. There are two facets of this theory, the stability part and the instability part. On the one hand, the Kolmogorov-Arnold-Moser theory says that most bounded motions remain quasiperiodic when slightly perturbed. On the other hand, Poincar\'e discovered that in general a small perturbation will create ``chaos" hence lead to complicated instability behaviors. This mechanism was first discovered by Poincar\'e when studying Newtonian three-body problem and we will explain it in Section \ref{SSArnold} and Appendix \ref{AppPoincare}. The chaotic dynamics in general relativity has been studied by many authors (c.f. \cite{Bo,GAC,KC,Mo,SS} etc). We study the stability part of geodesic dynamics of blackholes in a companion paper \cite{X}. In the current paper, we focus on the instability part.

The general principle is first to find a hyperbolic fixed point and its associated homoclinic orbit in the unperturbed system, then a generic small time periodic perturbation will create chaos following Poincar\'e.

The first prominent place to look for the hyperbolic fixed point is the bound photon orbits. For massless particles, these are orbits moving around the blackhole with constant radius and unstable under radial perturbations. A supermassive blackhole in the galaxy M87 has been observed recently \cite{M87,N} by the  Event Horizon Telescope that exhibits a bright ring. If we model the blackhole by Kerr spacetime, then the ring is a neighborhood of the set of  bound photon orbits \cite{J}. In this paper we study only the dynamics of massive particles, instead of massless ones, since for the former there are other important bound orbits such as homoclinic orbits other than the bound spherical orbits. We use the terminology {\it photon shell} to call the set of bound spherical orbits unstable under radial perturbations, for both massive and massless cases and both Schwarzschild and Kerr cases. Moreover, the photon shell will be only slightly perturbed under metric perturbations, so for perturbed Schwarzschild or Kerr, we also use the terminology {\it photon shell} to call the set of bound orbits unstable under radial perturbations. There are two other places to look for the hyperbolic fixed point that are the infinity as well as the event horizon. Utilizing these hyperbolicity, we discover some remarkable new phenomena in the perturbed Schwarzchild and Kerr that do not exist in the unperturbed systems. These include
\begin{enumerate}
\item Arnold diffusion type orbit with large oscillation of angular momentum $L_z$ as well as other variables;
\item Oscillatory orbits between a bounded domain and infinity;
\item Chaotic orbits near the event horizon without falling into the event horizon. 
\end{enumerate}

In this paper, we first prove the existence of Arnold diffusion in perturbed Schwarzschild and Kerr metrics.  Arnold diffusion is a prominent manifestation of such instability behaviors. It says that if an integrable system is generically perturbed, then one can find an orbit along which the conserved quantities undergoes an oscillation of order 1, regardless the size of the perturbation. Moreover, the orbit can be almost dense on the part of phase space of bounded motions. Hence Arnold diffusion has global nature, and to find such a diffusing orbit, one has to organize local chaotic behaviors to form large scale instability. In the construction, we utilize the photon shell and homoclinic orbits to it (orbits approaching the photon shell in both future and past), so in physical space the diffusing orbit has a zoom-whirl behavior (see Figure 1 and \cite{GK,HLS} for the term), i.e. the orbit zooms out into quasi-elliptical laves to apastron and then performs multiple whirls along a photon shell orbit before zooming out again. In both cases, we find orbits whose $z$-component of the angular momentum, denoted by $L_z$, undergoes a big oscillation during the process of zoom-whirl.  For Schwarzschild, recalling that all orbits of the unperturbed system necessarily lie on a fixed plane passing through the origin, while our diffusing orbit will slowly change its orbital plane when zoom-whirling, and may even visit any $\dt$-ball of the photon sphere. For Kerr, the diffusing orbit will shadow orbits on the photon shell with different radii and simultaneously the latitudinal range of the orbit will also change significantly. If perturbations depending  periodically on the coordinate time $\tau$ are considered, then Arnold diffusion orbit with big oscillations of the particle's energy $E$ can also be found.  Our study does not take into account of the radiation effect. If the radiation were considered, then the moving particle tends to lose energy, but we expect that the mechanism of Arnold diffusion may be exploited for the particle to gain energy and to avoid falling into the blackhole. %The speed of diffusion is very slow and it takes time of order $O(|\log\eps|)$ to achieve noticeable change where $\eps$ is the size of the perturbation (c.f. \cite{Tr1}), so the diffusion mechanism only works when the radiation is slower than this rate.

Similar strategy also allows us to find oscillatory orbits in perturbed Schwarzschild and Kerr metrics. These are orbits that oscillates between infinity and a bounded domain, which are known to exist in Newtonian three-body problem.

Furthermore, the event horizon is a coordinate singularity where the azimuthal dynamics ($\dot\varphi$) becomes singular (infinite spin) while the radial dynamics $(\dot r)$ is regular. When we examine the dynamics of $\frac{dr}{d\varphi}$, we see that the outer event horizon for Kerr becomes a hyperbolic fixed point with an associated homoclinic orbit. When generically perturbed, Poincar\'e's theorem applies to yield chaotic motions around the event horizon, in particular orbits visiting a neighborhood of the event horizon infinitely often without falling into it.

\begin{figure}
\includegraphics[width=3.5in]{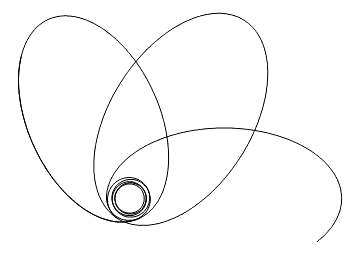}
\label{fig:zoomwhirl}
\caption{Zoom-whirl orbit}
\end{figure}

In the field of dynamical systems, there are various notions of stability, among which there are two major ones: orbital stability and the structural stability. The orbital stability refers to the fact that a small perturbation of the initial condition does not lead to significant divergent of the future orbits, while the structural stability means that the dynamical system after a small perturbation can be conjugate to the unperturbed one. This two notions usually differ drastically. For example, Arnold cat map $\left(\begin{array}{cc}2&1\\
1&1\end{array}\right):\ \T^2\to \T^2$ has exponential divergent behaviors for typical pair of initial conditions but it is structurally stable under $C^1$ small perturbations.  The study of the stability of the classical metrics such as Minkowski/Schwarzschild/Kerr is a fascinating field. We refer readers to \cite{RW,C,CK,KS,GKS} for some classical works, recent advances and more references therein.  The notion of stability in these researches is closer to the structural stability. In this paper, we complement them by examining the ``orbital stability problem" for the perturbed Schwarzschild/Kerr spacetime in general relativity.

%In the remaining sections of the introduction, we elaborate these points and give more precise statements.
%We study the dynamics of a particle moving in perturbed Schwarzschild and Kerr spacetimes. This can also model the dynamics of binary blackholes when one is significantly less massive than the other.

 %The main point that we want to highlight is that the dynamics of the geodesic flow in the perturbed Schwarzschild spacetime differs drastically from the unperturbed one.

 %The periodic timelike Kerr geodesics are classified in [Grossman-Levin-Perez-Giz, harmonic] and nonperiodic ones can be approximated by periodic ones.

 \subsection{Arnold diffusion around the photon shell}
 \subsubsection{Arnold's conjecture}
By Liouville-Arnold Theorem (see Theorem \ref{ThmLA}), an integrable system in the part of phase space consisting of bounded orbits can be reduced to a Hamiltonian system $H(x,y)=h(y), \ (x,y)\in \T^n\times D, \ D\subset \R^n,\ \T=\R/\Z$ whose Hamiltonian equation is $\begin{cases}\dot x=Dh(y)\\
 \dot y=0 \end{cases}$. The dynamics of this system is simple: each $y$ is conserved and on each $y$-torus $\T^n\times\{y\}$, the motion is linear with frequency $Dh(y) $ ($x\mapsto x+Dh(y) t,$  mod $\Z^n$,\ $t\in \R$). However, when we perturb the system to the following form called nearly integrable system
\begin{equation}\label{EqHameq} H(x,y)=h(y)+\eps P(x,y), \end{equation}
the dynamics becomes highly nontrivial. Note that nearly integrable systems model many important physical phenomena such as the Newtonian $N$-body problem.
The Kolmogorov-Arnold-Moser theory (see Theorem \ref{ThmKAM}) implies that  when the frequency $Dh(y)$ is a Diophantine vector, the corresponding $\T^n\times\{y\}$-torus is robust under perturbations. For each $\eps$ small enough, the set of invariant tori form a Cantor set (closed nowhere dense subset) of large measure (the complement has measure of order $\sqrt\eps$). Each invariant torus is an $n$-dimensional Lagrangian torus. When $n=2$, one a fixed energy level (three dimensional) each orbit is either lying on a torus or trapped between two neighboring tori so the action variable $y$ oscillates at most $O(\sqrt\eps)$. However, when $n\geq 3$, the complement to the set of invariant tori is connected, therefore there is a chance for orbits to wander in the phase space and produce almost dense orbits.

Arnold in \cite{A64} constructed an example in which the $y$-variable can wander arbitrarily (see Section \ref{SSArnold}), then made the following conjecture.
\begin{Con}[\cite{A66,A94}] \label{ConArnold} For any two points $y'$ and $y''$ on the connected level hypersurface of $h$ in the action space there exist orbits of \eqref{EqHameq} connecting an arbitrary small neighborhood of the torus $y=y'$ with an arbitrary small neighborhood of the torus $y=y''$, provided that $\eps\neq 0$ is sufficiently small and $P$ is generic.
\end{Con}
The statement can be found in \cite{A66,A94} etc, as well as in the book \cite{A05} in Problem 1963-1, 1966-3, 1994-33 etc. We refer readers to the work \cite{CX} by the author and Cheng for the up-to-date advances on this conjecture (its solution in the smooth category for convex systems), as well as further literature review therein.
It is well-known that a generic statement in mathematics says nothing for a given system, thus it is also desirable to have checkable conditions to find diffusing orbits. In this direction, we mention the geometric method developed by Delshams-Gidea-de la Llave-Seara \cite{DLS06,GLS} and Treschev \cite{Tr1,Tr2}, which will be the main tools that we apply in this paper.

\subsubsection{Arnold diffusion in perturbed Schwarzschild/Kerr spacetime}
%Small amplitude diffusion (s.a.d.), almost dense diffusion (a.d.d.).
In both Schwarzschild and Kerr metrics we use coordinates $(\tau,r,\theta,\varphi)$ where $\tau$ is the coordinate time, $r$ is the polar radius and $(\theta,\varphi)$ is the standard spherical coordinates with $\theta$ the latitude angle and $\varphi$ the azimuth angle. Dual to the variables $\tau$ and $\varphi$, the particles energy $E$ and $z$-component of the angular momentum $L_z$ are conserved respectively, in both Schwarzschild and Kerr spacetimes.
  We have the following definition.
\begin{Def}
A metric is called stationary if it is independent of the coordinate time $\tau$ and axisymmetric if it is independent of the azimuth angle $\varphi$.
\end{Def}
We first give the \emph{informal} statements of our results on Arnold diffusion in perturbed Schwarzschild/Kerr spacetime, in order to avoid technicalities. We will use the phrase ``under certain nondegeneracy conditions", which mainly means the nondegeneracy of a Melnikov function. The full statements can be found in Section \ref{SSADStat} and \ref{SSADNonStat} for Schwarzschild and Section \ref{SSSADK} for Kerr.
\begin{Thm}\label{ThmADSK1}
Let $\eps h_{\mu\nu}dx^\mu dx^\nu$  be a stationary perturbation to the Schwarzschild/Kerr metric   satisfying certain nondegeneracy conditions. Then  there exist $\rho>0$ and $\eps_0>0$, such that for all $|\eps|<\eps_0$, there exist  $T>0$ and an orbit of the perturbed system such that
$$|L_z(T)-L_z(0)|> \rho. $$
\end{Thm}
In physical space, the orbit in the previous theorem has zoom-whirl pattern. Moreover, the orbit plane will gradually tilt, which differs drastically from the Schwarzschild case, where all orbit moves on a plane.

It is also natural to consider perturbations depending on $\tau$ periodically (c.f. \cite{RW}), in which case the particles energy $E$ can also diffuse.
\begin{Thm} Let $\eps h_{\mu\nu}dx^\mu dx^\nu$  be a $1$-periodic in $\tau$ perturbation to Schwarzschild or Kerr spacetime  satisfying certain nondegeneracy conditions. Then there exist $\rho>0,\eps_0>0$, such that for all $|\eps|<\eps_0,$ there exists an orbit and a time $T$ with
$$|(E,L^2_z)(0)-(E,L^2_z)(T)|>\rho.$$
\end{Thm}

The nondegeneracy conditions are explicitly given and in general very easy to satisfy. In Appendix \ref{AppMelnikovAD}, we give an explicit perturbation to Schwarzschild metric that solves the linearized Einstein Field Equation and verify the condition. However, we do not do it for Kerr due to the complexity of the perturbation problem. When a perturbation is given, the nondegeneracy condition can be checked in a similar manner to the Schwarzschild case. %The physical meaning of this kind of diffusion orbit was discussed in the beginning of the introduction.

Here the lower bound $\rho$ depends on the $h_{\mu\nu}$ but does not depend on $\eps$. In fact, there are also checkable conditions in literature to guarantee that $\rho$ is independent of $h_{\mu\nu}$ and can be as large as possible. We will discuss this in Remark \ref{RkADLong} and \ref{RkConS} and Conjecture \ref{ConS1} and \ref{ConK1}.

The above mentioned results on Arnold diffusion utilize the photon shell as well as the homoclinic orbits to it. However, in the spirit of Arnold's conjecture \ref{ConArnold}, we expect that for generic perturbations, there exist orbits that are $\dt$-dense on each energy level of bounded motions. Here the genericity is posed for perturbations satisfying Einstein Field Equation or at least the linearized equation, so it introduces extra difficulties. We propose several conjectures in the spirit of Arnold as Conjecture \ref{ConS1}, \ref{ConS2}, \ref{ConK1} and \ref{ConK2} where we locate the part of phase space to look for Arnold diffusion and give expected phenomena. Despite of lacking a rigorous proof, the phenomena may still be observable in practice.

We finally speculate that the Arnold diffusion along zoom-whirl orbits can be a complement to the Penrose process. Recall that Penrose process is a process of extracting energy and angular momentum from a Kerr blackhole where a moving particle falls into the ergosphere but still outside the event horizon. Exploiting the Arnold diffusion mechanism, each time when the particle's energy and angular momentum are increased after the Penrose process, we may arrange the particle to lower its energy and angular momentum before falling into the ergosphere to perform another Penrose process again, thus hopefully we can continue to extract more energy and angular momentum from the blackhole. We shall elaborate on this in Section \ref{RkADLong}.

\subsection{Oscillatory Motions between a bounded domain and infinity}
In both Schwarzschild and Kerr spacetime, when the particle's energy is 1, there is a timelike geodesic escaping to infinity that is similar to the parabolic Kepler orbit in Newtonian two-body problem. In this case, we may treat the infinity as a degenerate hyperbolic fixed point and the parabolic orbit as a homoclinic orbit to infinity. Then a generic small perturbation will create chaotic motion that is similar to the oscillatory motion in Newtonian three-body problem. This shows that the post Newtonian nature of the two metrics in the region far from the blackhole center.

The next result is the \emph{informal} statement of our results on the existence of oscillatory motions. Precise statements can be found in Theorem \ref{ThmOsc} for Schwarzschild and Theorem \ref{ThmOscK} for Kerr.
\begin{Thm}
Let $\eps h_{\mu\nu}dx^\mu dx^\nu$ be a stationary perturbation of the Schwarzschild/Kerr metric satisfying
\begin{enumerate}
\item The perturbation $H_1= \frac12  g^{\al\mu}h_{\mu\nu}g^{\beta\nu}p_\al p_\beta+O(\eps)$ to the Schwarzschild or Kerr Hamiltonian converges to 0 asymptotically like $O(\frac{1}{r^2})$, and $\lim r^2 H_1$ is constant as $r\to\infty;$
\item The subspace $\{\theta=\pi/2,\ p_\theta=0\}$ is invariant under the perturbed Hamiltonian system;
\item $|L_z|$ is sufficiently large and the Melnikov function has a nondegenerate zero for $\varphi_0=0$,
\end{enumerate}
then  for $\eps>0$ sufficiently small in the perturbed system there exists an orbit of oscillatory type, i.e. there exists a constant $C$ such that along the orbit we have
$$\liminf_{t\to\infty} r(t)<C,\quad \limsup_{t\to\infty} r(t)=\infty.$$
\end{Thm}
Again we give an explicit perturbation of Schwarzschild that verifies the assumptions in Appendix \ref{AppOsc}.

\subsection{Chaotic Carrousel Orbits near the Event Horizon }
We also apply the general strategy to study the dynamics near the event horizon for a perturbed Kerr metric. Near the event horizon $r=r_+$, the radial motion $(r,\frac{dr}{dt})$ is nonsingular while the motion of $\frac{d\tau}{dt}$ and $\frac{d\varphi}{dt}$ becomes singular like $\frac{1}{r-r_+}$. It is known that the singularity is only a coordinate singularity which can be resolved by a coordinate change. We introduce a new coordinate system where we replace the proper time by the azimuthal angle $\varphi$ (or $\tau$). Then we see in the phase space $(r,\frac{dr}{d\varphi})$, the event horizon becomes a hyperbolic fixed point and has an associated homoclinic orbit. If the metric is perturbed generically by a non-axisymmetric perturbation, then we can apply Poincar\'e's theorem to show that chaos occurs. Such a chaotic behavior can be described by Smale's horseshoe, which is a compact invariant set that can be conjugate to a shift of finitely many alphabet (c.f. Appendix \ref{AppPoincare} and \cite{BS,KH}). In particular we have orbits jumping back and forth between neighborhoods of $r_+$ and $r_O(>r_+)$, where the latter is a zero of the velocity field $\frac{dr}{dt}$ as a function of $r$. Note that this differs drastically from the unperturbed case where all orbit falls into the event horizon when getting close to it.

We again give the following informal statement by removing the technical assumptions and refer readers to Theorem \ref{ThmHorizon} for more details of the following theorem.
\begin{Thm} Let $\eps H_1=\eps h^{\mu\nu}p_{x^\mu} p_{x^\nu}$ be a stationary and nonaxisymmetric perturbation of the Kerr spacetime with $L_z=aE$ satisfies the following \begin{enumerate}
\item The perturbed system preserves the equator plane $\{\theta=\pi/2,p_\theta=0\}$.
\item As $r\to r_+$, $H_1$ satisfies certain boundary condition;
\item The Melnikov function has a nondegenerate zero.
\end{enumerate}
 Then for $\eps$ sufficiently small, there exists a Smale horseshoe in neighborhoods of $r_+$ and $r_O$. In particular, there exists an orbit which visits the two neighborhoods  repeatedly as $t\to\pm\infty$.
\end{Thm}

We will further discuss dynamics inside the event horizon using similar methods in Section \ref{SCarrousel}.

\noindent {\bf Organization of the paper} The paper is organized as follows. The main body of the paper consists mainly of introductory and conceptual arguments and statements, and we put all the technical ingredients to the appendix.
 In Section \ref{SIntroHam}, we give a brief introduction to nearly integrable Hamiltonian dynamics including Liouville-Arnold theorem, KAM theorem and Arnold diffusion. In Section \ref{SS}, we study Arnold diffusion in the perturbed Schwarzschild spacetime. In Section \ref{SOsc}, we study oscillatory orbits in perturbed Schwarzschild spacetime.  In Section \ref{SK}, we study Arnold diffusion in perturbed Kerr spacetime.  In Section \ref{SCarrousel}, we show the existence of chaotic orbits near the event horizon in perturbed Kerr spacetime.
 Finally, we have three appendices. In Appendix \ref{AppPoincare} we give Poincar\'e's theorem on separatrix splitting and the derivation of Melnikov function and scattering map. In Appendix \ref{SNHIM}, we introduce the theorem of normally hyperbolic invariant manifolds. In Appendix \ref{SReggeWheeler}, we find concrete perturbations to Schwarzschild metric and confirming the existence of Arnold diffusion and oscillatory orbits.

 \section{Introduction to nearly integrable Hamiltonian dynamics}\label{SIntroHam}
 In the Hamiltonian formalism of the classical mechanics, a smooth Hamiltonian function $H$ on a symplectic manifold $(M,\omega)$ is given, and defines a vector field $X_H$ through $\omega(\cdot, X_H)=dH$ which defines a dynamical system by solving the ODE $\dot x=X_H(x)$.
 \subsection{Liouville-Arnold theorem}
The classical Liouville-Arnold theorem gives important characterization of the integrable systems.

\begin{Thm}[Liouville-Arnold] \label{ThmLA}Let $H_1=H:\ M^{2n}\to \R$ be a Hamiltonian and suppose there are $H_2,\ldots,H_n:\ M\to \R$ satisfying
\begin{enumerate}
\item[(a)] $\{H_i,H_j\}\equiv0$, for all $i,j=1,\ldots,n$, where $\{H_i,H_j\}:=\omega(X_{H_i},X_{H_j})$ is the Poisson bracket;
\item[(b)] the level set $M_{\mathbf a}:=\{(q,p)\in M\ |\ H_i(q,p)=a_i,\ i=1,\ldots,n\}$ is compact;
\item[(c)] At each point of $M_{\mathbf a}$, the $n$ vectors $DH_i,\ i=1,\ldots,n$ are linearly independent.
\end{enumerate}
Then
\begin{enumerate}
\item $M_{\mathbf a}$ is diffeomorphic to $\T^n=\R^n/\Z^n$ and is invariant under the Hamiltonian flow of each $H_i$.
%\item $M_{\mathbf a}$ is a Lagrangian submanifold, i.e. for any $u,v\in T_xM_{\mathbf a},\ \forall\ x\in M_{\mathbf a}$, we have $\omega(u,v)=0$.
\item in a neighborhood $U$ of $M_{\mathbf a}$, there is a symplectic transform $\Phi(q,p)=(\theta,I)$ such that $\Phi(U)=\T^n\times (-\dt,\dt)^n$ for some $\dt>0$.
\item In the new coordinates, each $K_i:=H_i\circ \Phi^{-1}$ is a function of $I$ only so the Hamiltonian equation is
$$\dot \theta=\omega_i(I):=\frac{\partial K_i}{\partial I},\quad \dot I=0. $$
\end{enumerate}
\end{Thm}
We refer readers to Section 50 of \cite{A89} for the proof and more details. The coordinates $(\theta, I)$ are called action-angle coordinates and can be constructed by standard recipe in \cite{A89}. %So it is natural to assume the phase space is a domain in $T^*\T^n$ when studying integrable systems or its small perturbations. We next consider nearly integrable system of the form
% $$H(x,y)=h(y)+\eps P(x,y),\quad (x,y)\in \T^n\times \R^n:=T^*\T^n.$$
% When $\eps=0$, the Hamiltonian equations are $\begin{cases}&\dot x=Dh(y)\\
% &\dot y=0
% \end{cases}.$ This implies that the phase space is foliated into tori of the form $\T^n\times \{y\}$, each of which is invariant under the Hamiltonian flow. Restricted to each torus, the dynamics is simply $x\mapsto x+Dh(y) t,$  mod $\Z^n$ for all $t\in \R$. So the dynamics of integrable systems is well understood. We are interested in nearly integrable systems with $\eps>0$. %In the following, we first state two stability results: the KAM theorem and Nekhroshev theorem, then introduce the instability mechanism: the Arnold diffusion.

 \subsection{The Kolmogorov-Arnold-Moser theorem}
 The Kolmogorov-Arnold-Moser theorem is an important stability result on the dynamics of nearly integrable systems.
\begin{Def}
A vector $\al\in \R^n$ is called \emph{Diophantine} if there exist some $C>0,\tau>n-1$ such that for all $k\in \Z^n\setminus\{0\}$, we have
$$|\langle\al,k\rangle|\geq \frac{C}{|k|^\tau}. $$
We denote $\al\in \mathrm{DC}(C,\tau)$.
\end{Def}
For each $\tau>n-1$, the set $\cup_{C>0}\mathrm{DC}(C,\tau)$ has full Lebesgue measure in $\R^n$. A version of the KAM theorem states as follows.
\begin{Thm}\label{ThmKAM}
Suppose $\det D^2 h(I)\neq 0$ for $I\in D$. Then for all $C>0,\tau>n-1$ and all $I\in D$ with $\omega(I)\in \mathrm{DC}(C,\tau)$, there exists $\eps_0>0$ and $\ell$ sufficiently large such that when $\eps\|f\|_{C^\ell}<\eps_0$, we have that there is a torus $\mathcal T_I$ that is invariant under the Hamiltonian flow generated by the Hamiltonian $H(\theta,I)=h(I)+\eps f(\theta,I)$. Moreover, there is a symplectic transform $\Phi$ such that $\Phi(\mathcal T_I)=\T^n$ and the conjugated flow $\Phi\circ\phi^t_H\circ\Phi^{-1}$ is the translation $\theta\mapsto \theta+\omega(I)t$ on $\T^n$.
\end{Thm}
We refer readers to \cite{P} for more details. %In this paper, we shall use a version of nondegeneracy condition on a fixed energy level in place of $\det D^2 h(I)\neq 0$, whose geometric meaning is well explained in Appendix 8 of \cite{A89}.  The condition will become clear in Section \ref{SAM} and \ref{SKAM}. %The KAM theorem implies, in particular, when $n=2$, each three dimensional energy level is divided by each two dimensional invariant torus, so each orbit either lies on a KAM torus or trapped between two nearby tori thus the $y$ variable undergoes only small oscillations.
%While the KAM theorem dictates perpetual stability for initial conditions on a large measure set, the following Nekhoroshev theorem gives exponentially long time stability for all initial conditions.
 %\begin{Thm} Suppose $h$ and $P$ are real analytic functions and $D^2h$ is positive definite in the bounded domain $D$. Then there exists $C_1,C_2>0$ such that for all orbits $\{(x(t),y(t))\}\subset D$ we have
% $$|y(t)-y(0)|\leq \eps ^{\frac{1}{2n}},\quad \mathrm{for}\ |t|\leq e^{C_1\eps^\frac{-C_2}{2n}}.$$
 %\end{Thm}
% The exponents $1/2n$ can be slightly improved. There are also results for less regular assumptions on $h$ and $P$ and weaker assumptions on $h$. We refer readers to ???
\subsection{Arnold's example}\label{SSArnold}
When $n\geq 3$, each invariant torus has codimension $n-1(>1)$ on each energy level so this leaves the possibility that there exists orbit wandering in the complement of the KAM tori.
Arnold in \cite{A64} constructed the following example% for which he first discovered the phenomena called now Arnold diffusion
\begin{equation}\label{EqArnold}
H(\theta,I,x,y,t)=\dfrac{I^2}{2}+\dfrac{y^2}{2}+(\cos(2\pi x)-1)(1+\eps (\cos(2\pi\theta)+\sin(2\pi t))),
\end{equation}
where $(\theta,I;x,y;t)\in T^*\T^1\times T^*\T^1\times \T$ and proved that
\begin{Thm}[Arnold]\label{ThmArnold64}
In the system \eqref{EqArnold}, for any given $A<B$, there is an orbit  $\{(\theta(t),I(t),x(t),y(t))\}$ of the system and times $t_1<t_2$ with $I(t_1)\leq A$ and $I(t_2)\geq B$, provided $\eps>0$ is small enough. \end{Thm}
To understand this result of Arnold, we start with the analysis of the mathematical pendulum.
\subsubsection{The pendulum}\label{SSSPendulum}
The mathematical pendulum is prominent in the study of Arnold diffusion whose Hamiltonian is
$$H_0(x,y)=\frac{1}{2}y^2+(\cos 2\pi x-1),\quad (x,y)\in T^*\T. $$
%First, as a system of one degree of freedom, the Liouville-Arnold theorem can be applied to the regular values of $H$, so each regular energy level is an entire periodic orbit. Thus the phase space dynamics is further determined by the critical values of $H$.

Near the fixed points $O=(0,0)$, the Hamiltonian can be linearized as $\frac{y^2}{2}-\frac{(2\pi x)^2}{2}$. The linearized Hamiltonian equation is $\begin{cases}
\dot x=y\\
\dot y=4\pi^2 x\end{cases},$ so the fixed point is hyperbolic.  Let $O$ be the hyperbolic fixed point and $\phi^t:\ T^*\T\to T^*\T,\ t\in\R,$ be the flow generated by the Hamiltonian vector field. We define the stable $(W^s)$ and unstable $(W^u)$ manifolds of the fixed point $O$ as
$$W^s(O)=\{z\in T^*\T\ |\ \phi^t(z)\to O,\ \mathrm{as}\ t\to+\infty\},$$$$W^u(O)=\{z\in T^*\T\ |\ \phi^t(z)\to O,\ \mathrm{as}\ t\to-\infty\}.$$
For the pendulum, due to the 1-periodicity in $x$, $(0,0)=(1,0)$, we see that $W^s(O)$ coincides with $W^u(O)$ consisting of an entire homoclinic orbit denoted by $\{(x_0(t),y_0(t)),\ t\in \R\}$ with $(x_0(t),y_0(t))\to (0,0)$ as $t\to\pm\infty$.

It was discovered by Poincar\'e that the stable and unstable manifold will split (i.e. intersect transversely) if a generic time-periodic perturbation is added. Let us consider the perturbed Hamiltonian
$$H_\eps(x,y,t)=\frac{1}{2}y^2+(\cos 2\pi x-1)+\eps H_1(x,y,t),\quad (x,y)\in T^*\T, $$
where $H_1(x,y,t)=H_1(x,y,t+1)$. In this case, the Hamiltonian equation is time-dependent, so its solution is no longer a diffeomorphism on $T^*\T$ for each $t\in \R$. Instead, we consider the time-1 map denoted by $\phi^1_\eps$, which is indeed a diffeomorphism on $T^*\T$, due to the 1-periodic dependence on $t$ of $H_1$. We redefine the stable and unstable manifolds as
\begin{equation}\label{EqWSU}
\begin{aligned}
W_\eps^s(O)&=\{z\in T^*\T\ |\ \phi_\eps^n(z)\to O,\ \mathrm{as}\ n\to+\infty\},\\
W_\eps^u(O)&=\{z\in T^*\T\ |\ \phi_\eps^n(z)\to O,\ \mathrm{as}\ n\to-\infty\}.
\end{aligned}
\end{equation}
The splitting of $W^s_\eps(O)$ and $W^u_\eps(O)$ is the main mechanism responsible for the nonintegrability of the perturbed system.
The general method of measuring the separatrix splitting to the leading order is the nonvanishingness of the Melnikov function
\begin{equation}\label{EqMelnikov2}\mathcal M(\nu)=\int_\R\{H_0,H_1\}(x_0(t),y_0(t),t+\nu)\,dt,\end{equation}
where the Poisson bracket is given explicitly as $ \{H_0,H_1\}=\partial_x H_0\cdot\partial_y H_1-\partial_y H_0\cdot\partial_x H_1.$ The derivation of the Melnikov function can be found in the Appendix \ref{AppPoincare}.

\subsubsection{Arnold diffusion in Arnold's example}
In this section, we explain the proof of Theorem \ref{ThmArnold64}. We first consider the case of $\mu=0$. The resulting system has two degrees of freedom. Away from the set $\{\frac{y^2}{2}+(\cos(2\pi x)-1)=0\}$, the system is integrable in the Liouville-Arnold sense.

In the phase space $T^*\T^2$ there exists the cylinder $$\mathcal C=\{(\theta,I,x,y)=(\theta,I,0,0),\ I\in \R,\ \theta\in \T^1\}.$$
When restricted to $\mathcal C$, the resulting Hamiltonian system is given by the integrable Hamiltonian $H=\frac{I^2}{2}$.
 Each circle $\mathcal C(I):=\{I=\mathrm{const},\ \theta\in \T,\ x=0,\ y=0\}$ in the cylinder is invariant under the Hamiltonian flow of $H_0$. %The frequency $\omega$ along $\mathcal C$ has the form $(I,0)$ (item (4) of Theorem \ref{ThmLA}), so the cylinder on which the Liouville-Arnold theorem does not apply has resonant frequency, i.e. for all integer vector $k\in \Z^2$ of the form $(0,*)$, we have $\omega\cdot k=0$. %dynamics since $\dot I=0$.
Each circle $\mathcal C(I)$ also has stable and unstable manifolds denoted by $W_I^{u,s}:=\{z\in T^*\T^2\ |\ \mathrm{dist}(\phi^t(z), \mathcal C(I))\to 0,\ t\to\pm\infty\}$ with $+$ for $s$ and $-$ for $u$.

When the time-dependent perturbation is turned on ($\eps>0$), using the Melnikov function \eqref{EqMelnikov2} in the previous subsection,  it can be verified that the stable manifold $W^s_I$ and the unstable manifold $W^u_I$ of $\mathcal C(I)$ intersect transversely for all $I$. %in the system $$H_3(x,y,t)=\frac{y^2}{2}+(\cos(2\pi x)-1)(1+\mu\sin(2\pi t)).$$
Therefore the transversality implies that $W^u_I$ intersects $W^s_{I'}$ transversely if $I$ and $I'$ is sufficiently close.  Then orbits can be found to shadow a sequence of $W_I^{u/s}$ chain to have large oscillation of $I$.

 \section{Arnold diffusion in the perturbed Schwarzschild spacetime}\label{SS}
 We study the dynamics of a particle moving in the Schwarzschild background and show that Arnold diffusion exists when the metric is properly perturbed.
 \subsection{Schwarzschild dynamics}\label{SSMetricS}
 The Schwarzschild metric is as follows
 \begin{equation}
 ds^2=-\left(1-\frac{2M}{r}\right)d\tau^2+\frac{1}{1-\frac{2M}{r}} dr^2+r^2(d\theta^2+\sin^2\theta d\varphi^2)=g_{\mu\nu}dx^\mu dx^\nu
\end{equation}
 with $x^0=\tau$ the time coordinate, $x^1=r$ is the polar radius, and $x^2=\theta\in [0,\pi], x^3=\varphi\in[0,2\pi)$ are the spherical coordinates, where $\theta$ is the latitudinal angle and $\varphi$ is the azimuthal angle. Here $M$ is the mass of the blackhole, and $r=2M$ is the event horizon. We are only interested in the dynamics of a particle moving outside of the event horizon. In the following, we denote $\al=\left(1-\frac{2M}{r}\right)$.

 We view the metric as twice of a Lagrangian $L(x,\dot x)=\frac12 g_{\mu\nu}\dot x^\mu \dot x^\nu$ (here $\dot x=\frac{dx}{dt}$ means the derivative with respect to the particle's proper time   $t$) and obtain its conjugate Hamiltonian after a formal Legendre transform as follows
 \begin{equation}\label{EqHamS}
\begin{aligned}
& \begin{cases}
 \frac{d\tau}{dt}&=-\al^{-1}p_\tau, \\
 \frac{dr}{dt}&=\al p_r,\\
 \frac{d\theta}{dt}&= \frac{p_\theta}{r^2},\\
  \frac{d\varphi}{dt}&=\frac{p_\varphi}{r^2\sin^2\theta},\\
\end{cases}\\
  2H(x,p)&=2(\dot x^\mu p_{x^\mu}-L(x,\dot x))=g^{\mu\nu}p_\mu p_\nu\\
  &=-\al^{-1}p_\tau^2+\al p_r^2+r^{-2}\left(p_\theta^2+\frac{1}{\sin^2\theta} p_\varphi^2\right).
\end{aligned}
\end{equation}
The Hamiltonian system has four conserved quantities as $H, p_\tau, p_\varphi, L^2=p_\theta^2+\frac{1}{\sin^2\theta} p_\varphi^2$ with the physical meanings: the total Hamiltonian, the particle energy, the third component of the angular momentum and the square of the total angular momentum respectively. In the following, we also use $L_z$ for $p_\varphi$ and $E$ for $p_\tau$.

We consider a particle moving in the Schwarzschild spacetime along geodesics with invariant mass $\mu$, so along a geodesic, we have
$g_{\mu\nu} \frac{dx^\mu}{dt}\frac{dx^\nu}{dt}=-\mu^2. $
We choose $\mu^2=0$ for massless particles and $\mu^2=1$ for massive particles, and the geodesics on the former case is called a lightlike geodesic and the latter timelike. In this paper, we mainly focus on the $\mu^2=1$ case.
\subsubsection{The radial dynamics, critical points and homoclinic orbits}\label{SSSABC}
We next analyze the radial dynamics. Setting $p_\tau=E$, $H=-\frac12\mu^2$ and $\frac{dr}{dt}=\al  p_r$, we obtain
 \begin{equation}\label{ThmRadialS}
 \begin{aligned}
%-\frac12\mu^2&=-\left(1-\frac{2m}{r}\right)^{-1}E^2+\left(1-\frac{2m}{r}\right) p_r^2+r^{-2}L^2\\
\frac{1}{2}E^2&=\frac12\left(\frac{dr}{dt}\right)^2+V(r),\\
2V(r)&=\left(\mu^2+\frac{L^2}{r^2}\right)\al =\mu^2-\frac{2\mu^2 M}{r}+\frac{L^2}{r^2}-\frac{2ML^2}{r^3}.
\end{aligned}
\end{equation}
Thus we may visualize the radial dynamics in the $(r,\frac{dr}{dt})$-plane as a particle moving in the potential well of $V$. The critical points of $V$ corresponds to orbits with $r=\mathrm{const}$. Setting
$\frac{dV}{dr}=0$ we obtain
$$\mu^2 M r^2-L^2 r+3ML^2=0$$
with solutions (for $\mu^2=1$)
\begin{equation}\label{Eqrpm}r_\pm=\frac{L^2\pm \sqrt{L^4-12M^2L^2}}{2M}. \end{equation}
For $L$ satisfying $L^2>12M^2$, there are two solutions. The one with $-$ is a local maximum of $V$ hence  $(r_-,0)$ is a hyperbolic fixed point and that with $+$ is a local minimum of $V$ hence $(r_+,0)$ is an elliptic fixed point. All orbits with initial value $r<r_-$ plunge into the blackhole.
When $L^2=12M^2$, the two critical points of $V$ coincide. The value $12M^2$ is called $L_{isco}^2$ where \emph{isco} stands for the \emph{innermost stable circular orbit}. For all values $L^2<L^2_{isco}$, all orbits plunge into the central blackhole.  When $L^2=L^2_{isco}=12M^2$ we have $r_+=r_-=6M:=r_{isco}$ and We also denote $\frac12E_{isco}^2=V(r_{isco})=\frac{4}{9}$. The set $\{r=r_-\}$ is called the \emph{photon sphere}, orbits on which remain on a circle with constant radius $r_-$ and are unstable under radial perturbations.

We next study the homoclinic orbits to the photon sphere. This involves comparing the value $V(r_-)$ with $V(\infty)=\frac{1}{2}$. We note that $$\frac{d V}{d L^2}\big|_{r=r_-}=\frac{\partial V}{\partial r_-}\frac{dr_-}{dL^2}+\frac{\partial V}{\partial L^2}=\frac{\partial V}{\partial L^2}=\frac{1}{r_-^2}(r_--2M)>0.$$
This implies that $V(r_-)$ as a function of $L^2$ is strictly increasing. When $L^2=L_{isco}^2$ we have $V(r_{isco})=4/9<V(\infty)=1/2$.

We denote by $L^2_{homo}$ the value of $L^2$ such that $V(r_-)=V(\infty)$.
There are three cases depending on the values of $L^2$:

\begin{figure}
\begin{minipage}[t]{0.5\linewidth}
\centering
\includegraphics[width=2.2in]{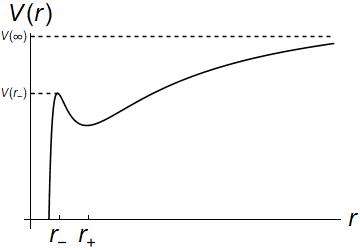}
\subcaption{Case (A)}
\label{fig:side:a}
\end{minipage}%
\begin{minipage}[t]{0.5\linewidth}
\centering
\includegraphics[width=2.2in]{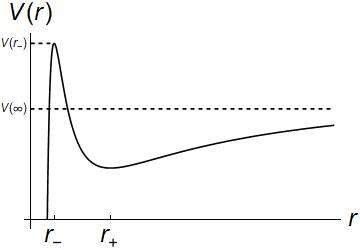}
\subcaption{Case (B)}
\label{fig:side:b}
\end{minipage}
\caption{The Schwarzschild effective potential}
\label{FigV}
\end{figure}

\begin{enumerate}
\item[(A)] $V(r_-)<V(\infty)$, i.e. $L^2_{isco}<L^2<L_{homo}^2,$ see Figure \ref{fig:side:a}. In this case, there is a homoclinic orbit to the hyperbolic fixed points $(r,\frac{dr}{dt})=(r_-,0)$. In the $(r,\frac{dr}{dt})$-plane, the orbit is given explicitly as
$$\mathrm{Graph}\left\{\frac{ dr}{dt}=\pm\sqrt{ E^2-2V(r)},\quad E^2=2V(r_-)\in \left(\frac49,\frac12\right)\right\}.$$
\item[(B)] $V(r_-)>V(\infty)$, i.e. $L_{homo}^2<L^2,$ see Figure \ref{fig:side:b}. In this case, we may view $\infty$ as a degenerate hyperbolic fixed point after a transformation (see below Section \ref{SOsc}).
Then it also has a homoclinic orbit
$$ \mathrm{Graph}\left\{\frac{dr}{dt}=\pm\sqrt{ E^2-2V(r)},\quad E^2=2V(\infty)=1\right\}.$$
\item[(C)] $V(r_-)=V(\infty)$, i.e. $L^2=L_{homo}^2.$ This is a critical case interpolating between the above two cases. There is an orbit converging to $(r_-,0)$ as $t\to\infty$ and to infinity as $t\to-\infty$, and \emph{vice versa}.
\end{enumerate}

\begin{Rk}
We remark that the dynamics of a massless particle differs drastically from a massive particle. For a massless particle, we take $\mu^2=0$ and consider lightlike geodesics. It turns out that there is only one critical point $r_c$ of $V$ that is a global max. If the initial position $r=r_c$, then the particle remains on the photon sphere $\{r=r_c\}$. If the initial position $r>r_c$, then it escapes to infinity and if $r<r_c$, it falls into the blackhole.
\end{Rk}
\subsubsection{The normally hyperbolic invariant manifold}
Let us consider case (A) first. In the Hamiltonian system \eqref{EqHamS}, there is no explicit dependence on $\tau$. If we fix the constant $p_\tau^2=E^2=2V(r_-)$, which is a function of $L^2$ and $M$, then we can treat the Hamiltonian $H$ in \eqref{EqHamS} as a system of three degrees of freedom depending on the variables $r,p_r,\theta,p_\theta,\varphi,p_\varphi$. The four dimensional submanifold $$N:=\left\{r=r_-,\ \frac{dr}{dt}=0\right\}$$ is a normally hyperbolic invariant manifold (NHIM) in the sense of Definition \ref{DefNHIM}. One remarkable property of the NHIM is its persistence under small perturbations, which is summarized in Theorem \ref{ThmNHIM} in Appendix \ref{SNHIM}. The normal Lyapunov exponents are $\pm \sqrt{-V''(r_-)}$ obtained as the eigenvalues of the linearized radial Hamiltonian equation $$\left(\begin{array}{c} \dt r\\
\dt\dot r\end{array}\right)'=\left(\begin{array}{cc}0&1\\
-V''(r_-)&0\end{array}\right)\left(\begin{array}{c} \dt r\\
\dt\dot r\end{array}\right)$$
 at $r=r_-$. The Hamiltonian system restricted to $N$ is a Hamiltonian system of two degrees of freedom of the form
 \begin{equation}\label{HamN}H_N(\theta,p_\theta,\varphi,p_\varphi)=r_-^{-2}\left(p_\theta^2+\frac{1}{\sin^2\theta} L_z^2\right)-\al^{-1}E^2.
\end{equation}
The system has two degrees of freedom and is integrable in the sense of Theorem \ref{ThmLA}. So we can introduce action-angle coordinates $(\al_\theta,J_\theta,\al_\varphi,J_\varphi)\in \T\times \R_+\times \T\times \R$ to write the Hamiltonian as a function of $J_\theta,J_\varphi$ only. We refer readers to \cite{X} for more details. %We will carry out details in Section \ref{SAM} and Appendix \ref{AppPhotonS}.

The NHIM $N$ has its stable and unstable manifolds $W^s(N)$ (defined as the set of points whose forward limit under the flow is $N$) and $W^u(N)$ (the set of points whose backward limit is $N$) coincide which consists of homoclinic orbits approaching $N$ in both the future and the past.
 \subsubsection{The homoclinic orbit}
 In this section, we solve for the homoclinic orbit in case (A) (see Figure \ref{fig:side:aHomo}).
 We perform a time reparametrization $d\lambda=\frac{L_z}{r^2} dt$ and solve for $r,p_r, \theta, p_\theta$ as functions of the new time $\lambda$. We will call $\lambda$ the Mino time.  Introducing $u=\frac1r$, we get that the equations of motion has the form
 \begin{equation}\label{EqHomo}
 \begin{cases}
\left(\frac{d u}{d\lambda}\right)^2&=\frac{E^2}{L^2}-(1-Mu)\left(u^2+\frac{1}{L^2}\right),\\
\left(\frac{d\theta}{d\lambda}\right)^2&=\frac{L^2}{L_z^2}-\frac{1}{\sin^2\theta},\\
 \frac{d\varphi}{d\lambda}&=\frac{1}{L_z^2\sin^2\theta},\\
  \frac{dL_z}{d\lambda}&=0.\\
 \end{cases}
  \end{equation}

 We first solve for $u$ as a function of $\lambda$. Since we are considering case $(A)$, the polynomial on the RHS of the $\left(\frac{d u}{d\lambda}\right)^2$ equation has a double root $u_1=1/r_-$,
 %Here the value $\frac{E^2}{L^2}$ is critical in the sense that the polynomial on the RHS of the $\left(\frac{d u}{d\lambda}\right)^2$ equation has no constant term
  so we can write $\left(\frac{d u}{d\lambda}\right)^2=(u-u_1)^2(u-u_2)$ where $u_2<u_1$. The equation can be solved explicitly as \begin{equation}\label{EqRadial}u(\lambda)=u_1-a^2\mathrm{sech}^2(a(\lambda-\lambda_0)/2), \ a^2=u_1-u_2.\end{equation}

\begin{figure}
\begin{minipage}[t]{0.5\linewidth}
\centering
\includegraphics[width=2.in]{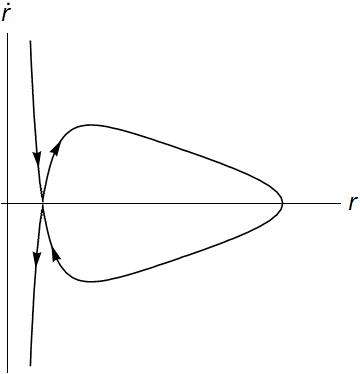}
\subcaption{Homoclinic orbit to the hyperbolic fixed point}
\label{fig:side:aHomo}
\end{minipage}%
\begin{minipage}[t]{0.5\linewidth}
\centering
\includegraphics[width=2.in]{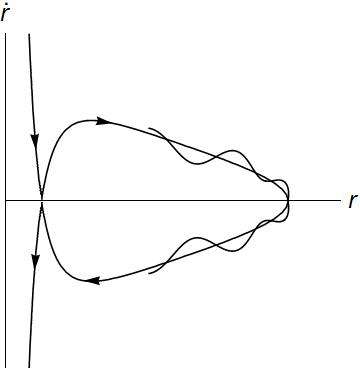}
\subcaption{Separatrix splitting}
\label{fig:side:bHomo}
\end{minipage}
\caption{Integrable vs nearly integrable}
\end{figure}

 %We next consider the $\theta$- and $\varphi$-equations. Note that a Schwarzschild geodesics moves on a plane passing through the origin, therefore we can choose this plane as the equator by setting $\theta=\pi/2$. Therefore $\frac{d\theta}{d\lambda}=0, \ p_\theta=0$ and $\frac{d\varphi}{d\lambda}=\frac{1}{L_z^2}$.

Similarly, for the variable $\theta$, we get (defining $b=\frac{L}{L_z}\geq 1$) $$\lambda=\lambda_0+\int_{\theta(0)}^\theta\frac{\sin\theta\,d\theta}{\sqrt{b^2\sin^2\theta-1}}=\lambda_0-\frac{1}{b}\arcsin\left(\frac{b\cos x}{\sqrt{b^2-1}}\right)\Big|_{\theta(0)}^\theta.$$
 Choosing $\lambda_0=0$, then by solving the inverse function, we get \begin{equation}\label{EqVertical}\cos\theta=\cos\theta(0)+\frac{\sqrt{b^2-1}}{b}\sin(b\lambda)\end{equation} as a function of $\theta(0)$ and $\lambda$. From \eqref{EqHamS}, we get $p_\theta=L_z\frac{d\theta}{d\lambda}$, thus we obtain $p_\theta$ as a function of $\theta(0),p_\theta(0),\lambda$ from \eqref{EqHomo}. We also solve from \eqref{EqHomo} $\varphi$ as a function of $\theta(0)$ and $\lambda$.

 %Note that a translation $r(\cdot +a)$ is also a solution. For the sake of definiteness, we suppose $r(0)=\max$. This gives a unique solution $r(\lambda)$. With this we relate $t$ and $\lambda$ by integrating the equation $d\lambda=\frac{1}{r^2}dt$. Again for definiteness, we choose $t=0$ corresponds to $\lambda=0$. Then we get $t=\int_0^\lambda r(\lambda)\,d\lambda:=t(\lambda)$ whose inverse function is $\lambda(t)$.

For future reference, we introduce some notations.
\begin{Not}\label{NotS1}
We denote by $\Gamma(\lambda)=(r(\lambda), p_r(\lambda)),$ $\lambda\in \R,$ the homoclinic orbit in the $(r,p_r)$-plane and by $\Gamma_c=(r_c,0),\ r_c=r_-$. We denote by $\Xi=(\theta, p_\theta, \varphi,p_\varphi)$ and by $\Xi_0$ %=(\theta(0)=\pi/2, p_\theta(0)=0, \varphi(0),p_\varphi(0))$  so we have $\Xi(\lambda)=\Xi_0+(0, 0, \frac{\lambda}{L_z^2},0)$.
the initial condition for the $\Xi$ variables. We further denote by $(\Gamma(\nu+\lambda), \Xi(\Xi_0,\lambda)),\ \lambda\in \R$ an orbit of Schwarzschild Hamiltonian system with initial condition $(\Gamma(\nu),\Xi_0)$. %The orbit can be expressed explicitly as a function of $\lambda$ and the initial conditions.
 \end{Not}
\subsection{Arnold diffusion via zoom-whirl orbits, stationary perturbations}\label{SSADStat}
In this section, we consider perturbation to the Schwarzschild metric and obtain our first result on Arnold diffusion. The perturbations that we consider in this section are stationary, i.e. independent of $\tau$, hence the particle's energy $E$ remains a conserved quantity. We shall fix a value of $E\in (4/9,1/2)$, so that we are in case (A) in the classification in Section \ref{SSSABC}. Then we get a Hamiltonian system of three degrees of freedom.
Let $g_{\mu\nu}dx^\mu dx^\nu$ be the Schwarzschild metric and $\eps h_{\mu\nu}dx^\mu dx^\nu$ be a perturbation. We treat $1/2$ of the the perturbed metric $\tilde g_{\mu\nu}dx^\mu dx^\nu=g_{\mu\nu}dx^\mu dx^\nu+\eps h_{\mu\nu}dx^\mu dx^\nu$ as the Lagrangian and obtain the corresponding Hamiltonian
$$H_\eps=\frac{1}{2}\tilde g^{\mu\nu}p_\mu p_\nu=\frac{1}{2}g^{\mu\nu}p_\mu p_\nu-\eps \frac12 g^{\al\mu}h_{\mu\nu}g^{\beta\nu}p_\al p_\beta+O(\eps^2). $$
We denote by $$H_1(r,p_r, \theta,p_\theta, \varphi,p_\varphi)=\eps^{-1}\frac12(\tilde g^{\mu\nu}-g^{\mu\nu})dx^\mu dx^\nu=-\frac12  g^{\al\mu}h_{\mu\nu}g^{\beta\nu}p_\al p_\beta+O(\eps).$$
 %Denote by   $(p_r(t),r(t))$ the homoclinic orbit to the fixed point $\{p_r=0,r=r_-\}$ in the original Schwarzschild Hamiltonian.
 We will consider perturbations depending on $\varphi$ so that $L_z$ is no longer a constant of motion.

 By the theorem of NHIM (c.f. Theorem \ref{ThmNHIM}), the NHIM $N$ is slightly perturbed. To apply the theorem of NHIM, we shall choose a large constant $C$ and consider a smooth cut-off function $\chi$ defined on the phase space that is one on the ball of radius $C$ and is zero outside the radius $C+1$, then we modify the perturbation to $-\frac12\eps \chi H_1$. We will only focus on the dynamics within the $C$-ball. This is a routine procedure to localize to a compact part of the phase space when applying the theorem of NHIM, and we will do it all the time later without explicitly mentioning it.

 We denote by $N_\eps$ the perturbed NHIM for the perturbed system $H_\eps$ and denote by $H_{N_\eps}$ the restriction of $H_\eps$ on $N_\eps$, which is a small perturbation of the Hamiltonian $H_N$. Now the situation is analogous to Arnold's example in Section \ref{SSArnold}. We can fix the total energy level $H_\eps=-\frac12$. Then the restricted system on $N_\eps$ has one and half degrees of freedom, which is analogous to the cylinder $T^*\T^1\times \T^1(\ni(I,\theta,t))$ in Arnold's example.  The $r$-component is analogous to the pendulum subsystem in Arnold's example, due to the existence of a hyperbolic fixed point and a homoclinic orbit in either case. The union of the homoclinic orbits is the stable and unstable manifolds of the unperturbed $N$. When generically perturbed,  %Let us consider a fixed energy level set $\{H_N=L^2\}$ of the unperturbed Hamiltonian system $H_N$ restricted to the unperturbed NHIM $N$. The level set is a three dimensional submanifold and a further level set $L_z$=const gives a two dimension invariant submanifold whose stable and unstable manifolds merge. When we introduce the perturbation, the main difference with Arnold's example is that some of the $L_z$-submanifolds may be broken so that it does not make sense to talk about their stable and unstable manifolds. However,
 we have that  the stable manifold $W^s(N_\eps)$ of $N_\eps$ and the unstable manifold $W^u(N_\eps)$ intersect transversally.  This gives the necessary ingredient to implement Arnold's mechanism.
  A point in the intersection $W^s(N_\eps)\cap W^u(N_\eps)$ gives rise to an orbit converging to $N_\eps$ under both forward and backward iterates but its $L_z$-value in the future may differ from that in the past. This motivates the following definition.
 \begin{Def}[The scattering map]\label{DefScatter}
Consider the flow $\Phi_t:\ M\to M$ and a NHIM $\widetilde N\subset M$. Assume that $\Gamma\subset W^s(\widetilde N)\cap W^u(\widetilde N) $ is a homoclinic manifold and assume that the intersection of $W^s(\widetilde N)$ and $W^u(\widetilde N)$ is transversal. We define the scattering map $\mathcal S:\ \widetilde N\to \widetilde N$ via $\mathcal S(z_-):=z_+$ for $z_\pm\in \widetilde N$ if there exists $z\in \Gamma$ such that
$$\mathrm{dist}(\Phi_t(z),\Phi_t(z_\pm))\to 0,\quad \mathrm{for\ } t\to \pm\infty.$$
\end{Def}
The next theorem gives an approximation of the scattering map in terms of the Melnikov function. We remark that our expression of the Melnikov function is nonstandard for the reason that the time reparametrization $\lambda$ changes the Hamiltonian structure. We postpone the proof to Appendix \ref{SSScattering}.

\begin{Thm}\label{ThmScatteringS1}
The scattering map $\mathcal S:\ N_\eps\to N_\eps$ is a well-defined symplectic map preserving the energy level set $\{H_{N_\eps}=-\frac12\}$ and $O(\eps)$-close to identity in the $C^1$ norm. Suppose that there exists an open set $U\subset N_\eps\cap H_{N_\eps}^{-1}(-\frac12)$ such that the map
$$\nu\mapsto J(\nu,\Xi_0)$$
has a nondegenerate critical point $\nu^*=\nu^*(\Xi_0)$ for all $\Xi_0\in U$, where  \begin{equation}
\label{EqMelnikov}J(\nu,\Xi_0)=\int_\R  \frac{r(\nu+\lambda)^2}{L_z}H_1(\Gamma(\nu+\lambda),\Xi(\Xi_0,\lambda))-\frac{r^2_c}{L_z}H_1(\Gamma_c,\Xi(\Xi_0,\lambda))\,d\lambda.\end{equation}

Then explicitly up to an error $o_{C^1}(1)$ we have for $\star=\varphi$
\begin{equation}\label{EqScattering}
\begin{aligned}
&\frac{1}{\eps}(p_\star(z_+)-p_\star(z_-))\\
&=-\int_\R \frac{r(\nu^*+\lambda)^2}{L_z}\partial_\star H_1(\Gamma(\nu^*+\lambda),\Xi(\Xi_0,\lambda))-\frac{r^2_c}{L_z}\partial_\star H_1(\Gamma_c,\Xi(\Xi_0,\lambda))\,d\lambda.
%&L_z(z_+)-L_z(z_-)\\
%&=-\eps\int_\R \frac{r(\al^*+\mu)^2}{L_z}\left[\partial_\varphi H_1(R(\al^*+\mu),\Xi(\Xi_0,\mu))-\partial_\varphi H_1(R_c,\Xi(\Xi_0,\mu))\right]\,d\mu,
\end{aligned}
\end{equation}

%Moreover, the map $\mathcal S$ preserves the energy levels of $H_{N_\eps}$.
\end{Thm}
%\begin{Rk}The RHS of \eqref{EqScattering} does not equal to $\partial_{\theta_0}J(\al,\Xi_0),\partial_{\varphi_0}J(\al,\Xi_0)$ since we are not in the action-angle coordinates and we do not have $\frac{\partial\theta}{\partial\theta_0}=1$ and $\frac{\partial\varphi}{\partial\varphi_0}=1.$
%\end{Rk}
With the scattering map, we are now ready to state our first result on Arnold diffusion.
\begin{Thm}\label{ThmGLS}
Let $\eps h_{\mu\nu}dx^\mu dx^\nu$  be a stationary perturbation to the Schwarzschild metric. Suppose
\begin{enumerate}
\item the assumption of Theorem \ref{ThmScatteringS1} holds;
%\item $L^2\in (L^2_{isco}, L^2_{homo})$;
\item the RHS of \eqref{EqScattering} with $\star=\varphi$ is nonvanishing at some point $\Xi_0\in U$. %for each set $\{L^2_z=\mathrm{constant}\in (a,b)\subset (0,L^2)\} $ as a subset of $ \{L^2=p_\theta^2+\frac{L_z^2}{\sin^2\theta}\}$, its image under the scattering map intersects itself transversally $($in \eqref{EqScattering} in Theorem \ref{ThmScatteringS1}, choosing $L_z(z_-)=\mathrm{const}$, we require $L_z(z_+)$ is a nondegenerate function of $\varphi$$)$.
\end{enumerate}
Then  there exist $\rho>0$ and $\eps_0>0$, such that for all $|\eps|<\eps_0$, there exist  $T>0$ and an orbit of the perturbed system such that $L(t)^2\in (L^2_{isco}, L^2_{homo})$ for all time $t\in [0,T]$ and
$$|L_z(T)-L_z(0)|> \rho. $$
\end{Thm}
\begin{proof}
 We first rewrite the Hamiltonian system $H=H_0+\eps H_1$ into a form that is separable in the unperturbed part  $H_0$ \eqref{EqHamS}. We shall fix the total energy $H=-\frac12$, multiply the equation by $\frac{r^2}{L_z}$ and introduce $$\mathcal H_\eps=\frac{r^2}{L_z}(H_0+\frac12+\eps H_1)=\frac{1}{L_z}(\al r^2 p_r^2-\al^{-1}r^2E^2+
 r^2+L^2+\eps r^2 H_1).$$ Then from the Hamiltonian equation we see that the dynamics of $H$ on the energy level set $H=-\frac12$ is the same as that of the Hamiltonian $\mathcal H_\eps$ on the energy level set $0$ up to a time reparametrization. Then the statement is a straightforward application of Theorem 4.1 of \cite{GLS} to the system $\mathcal H_\eps$ with the only difference being that Theorem 4.1 of \cite{GLS} is stated for a nonautonomous system while here we have an autonomous one. Indeed, in their proof first converts their nonautonomous system to an autonomous one. In our case, we can first transform to action-angle coordinates $(\al_\theta,J_\theta,\al_\varphi, J_\varphi)\in \T\times \R_+\times \T\times \R$ for the dynamics on $N$. We fix an energy level $H_N(J_\theta,J_\varphi)=\mathrm{const}$ and solve for $J_\theta$ then we can treat $J_\theta$ as a new Hamiltonian and $\al_\theta$ as the new time, so the $\al_\theta$ dependence in the perturbation becomes nonautonomous. This procedure is called energetic reduction (c.f. Section 45 of \cite{A89}. The existence and transversal intersection of the stable and unstable manifolds are independent of the coordinates. Then we are in the setting of \cite{GLS}. The idea of \cite{GLS} is to use the nondegeneracy of the scattering map, i.e. the transversal intersection of the stable and unstable manifolds to create a change in $L_z$
similar to Arnold's example. The nondegeneracy of the scattering map holds in a neighborhood by continuity. On the other hand, for the dynamics on the NHIM $N_\eps$, the authors use Poinicar\'e recurrence and developed a version of inclination lemma to find shadowing orbit.  We refer readers to \cite{GLS} for more details.
%fix a Poincare section $\al_\theta=0$ corresponding to $\theta=\pi/2$ and
\end{proof}
We will verify the assumptions, in particular the Melnikov nondegeneracy, for a perturbation of Schwarzchild in Section \ref{SReggeWheeler}. The assumptions are expected to hold for generic perturbations.

The phase space dynamics of the diffusion orbit is similar to that of Arnold's example, i.e. the orbit shadows orbits on the NHIM with occasionally excursion along homoclinic orbits. The configuration space dynamics of the orbit looks as follows. In the radial direction, it moves along the homoclinic orbit $\Gamma$ and when it gets close to the photon sphere, it moves along the photon sphere for some time before the next excursion to the homoclinic orbit, etc. Moreover, the orbit precesses during each excursion. In literature, this type of orbits are called {\it zoom-whirl orbits} (c.f. \cite{GK,LP1,LP2} etc). During the process of zoom-whirling, the $z$-component of the angular momentum $L_z$ undergoes a noticeable change, so the orbital plane changes from time to time, in sharp contrast to the unperturbed Schwarzschild case, where all orbits necessarily lie on a fixed plane passing through the origin due to the conservation of $L$ and $L_z$. %Note that this shows that an arbitrary small perturbation may lead to a drastic change to the dynamics.
% However, the Arnold diffusion orbit may wander on the entire photon sphere.

\begin{Rk}\label{RkADLong}
The main defect of the statement is that the size $\rho$ of the diffusing orbit depends a priori on the perturbation $H_1$, though independent of $\eps$. We expect that the following result holds for generic perturbations.
 %In fact, under a list of checkable conditions that are expected to be satisfied generically, stronger conclusion holds: \\

\emph{\qquad For all $a,b\in(0,|L|)$ with $a<b$, and all $\eps$ sufficiently small, there exists an orbit and time $T$ such that $|L_z(T)|\geq b$ and $|L_z(0)|\leq a$. }\\
However, most known machineries of proving Arnold diffusion in the presence of normal hyperbolicity $($c.f. \cite{CY1,CY2,DLS06,Tr1,Tr2}$)$ require certain twist property for the dynamics on $N$, which is absent in our case $($c.f. Proposition 3.1 of \cite{X}$)$. We list it as a conjecture in Section \ref{SSConjS}. The work \cite{GLS} that we cite above do not need any understanding of the dynamics on $N$, which may be applied to verify the above statement for concrete perturbations.
%Such conditions can be found in the works \cite{DLS06,GLS,Tr1,Tr2} etc.
If the statement is true, then we get orbit that visits any $\dt$-ball centered on the photon sphere, provided $\eps$ is chosen sufficiently small.
We do not pursue the generality here in order to keep the statement succinct, considering that there is no natural perturbation to the Schwarzschild metric. %In practice, when a perturbation is given, for example, the one we give in Appendix \ref{SReggeWheeler}, one can check the conditions to see if a long diffusing orbit exists or not.
\end{Rk}

%Note that the geodesics of the Schwarzschild spacetime with constant radius necessarily lie on a fixed plane passing through the origin. The theorem implies that excursions along the homoclinic orbits may significantly change the dynamics of an orbit, in particular, make it nonplanar.

\subsection{Arnold diffusion via zoom-whirl orbits, nonstationary perturbations}\label{SSADNonStat}
In this section, we study nonstationary perturbations so that the particle's energy $E$ is the no longer a constant of motion. We consider the case of $\tau$-periodic perturbations with period 1, which is the physically interesting case, see \cite{RW}.
For this purpose, we consider the Schwarzschild Hamiltonian $H$ \eqref{EqHamS} as a system of four degrees of freedom. Since $H$ has no explicit dependence on $\tau$, we think $\tau$ is defined on $\T=\R/\Z$. We thus obtain the six-dimensional NHIM $$\hat N=\{p_r=0,\ r=r_-\}$$ parametrized by the variables $(\tau, p_\tau, \theta, p_\theta, \varphi, p_\varphi)$, restricted to which, the Hamiltonian system has three degrees of freedom
\begin{equation}\label{EqHamHatN}
2H_{\hat N}(\tau,p_\tau, \theta,p_\theta, \varphi,p_\varphi)=-\frac{p_\tau^2}{\al(r_-)}+r_-^{-2}\left(p_\theta^2+\frac{1}{\sin^2\theta} p_\varphi^2\right).
\end{equation}
When the Schwarzschild metric is perturbed by a nonstationary $1$-periodic perturbation in $\tau$, the NHIM $\hat N$ will be slightly deformed into $\hat N_\eps$, which is a NHIM for the perturbed Hamiltonian system and can be written as a graph over $\hat N$. The perturbed Hamiltonian restricted to $\hat N_\eps$ is again Hamiltonian, denoted by $H_{\hat N_\eps}$. The stable manifold $W^s(\hat N_\eps)$ and unstable manifold $W^u(\hat N_\eps)$ of $\hat N_\eps$ in general do not coincide so that we can define a scattering map $\hat{\mathcal S}:\ \hat N_\eps\to \hat N_\eps$ by Definition \ref{DefScatter}.

Similar to Theorem \ref{ThmScatteringS1}, we can obtain a first order approximation of the scattering map.
\begin{Not}
We redefine the initial condition $\hat\Xi_0=(\tau(0),E(0),\Xi_0)\in \hat N$ where $\Xi_0=(\theta(0),p_\theta(0),\varphi(0), L_z(0))\in N$ as defined in Notation \ref{NotS1} and the solution $$\hat \Xi(\hat \Xi_0,\lambda)=\left(\tau(0)-\frac{r_-^2E(0)}{\al(r_-)L_z(0)}\lambda,E(0),\Xi(\Xi_0,\lambda)\right),\ \lambda\in \R,$$ to the Hamiltonian equations of the system $H_{\hat N}$  with initial condition $\hat\Xi_0$. As before, we still denote by $\Gamma(\nu+\lambda)$ the $(r,p_r)$-components of the homoclinic orbit with initial condition $\Gamma(\nu),\ \lambda\in \R$.
\end{Not}
With this modification, we have the Melnikov potential $J(\nu,\hat\Xi_0)$ formally the same as  \eqref{EqMelnikov}, and a similar statement to Theorem \ref{ThmScatteringS1} ($\hat N_\eps$ is now homeomorphic to $\R^3\times\T^3$ and $\star=\tau,\theta,\varphi$). For nonstationary perturbations, we have the following result similar to Theorem \ref{ThmGLS}.
\begin{Thm}\label{ThmDiffusionNonStatS} Let $\eps h_{\mu\nu}dx^\mu dx^\nu$  be a $1$-periodic in $\tau$ perturbation to the Schwarzschild metric.
Suppose
\begin{enumerate}
\item the assumption of the modified version of Theorem \ref{ThmScatteringS1} holds;
\item $L^2\in (L^2_{isco}, L^2_{homo})$, $E\in(E_{isco}, E_{homo})$,\ $L_z^2\in (0,L^2)$;
\item for $\star=\tau,\varphi$, the RHS of \eqref{EqScattering} is nonvanishing at some point $\Xi_0\in U$.
\end{enumerate}
 Then there exist $\rho>0,\eps_0>0$, such that for all $|\eps|<\eps_0,$ there exists an orbit and a time $T$ with $L(t)^2\in (L^2_{isco}, L^2_{homo})$, $E(t)\in(4/9, 1/2)$,\ $L_z(t)^2\in (0,L(t)^2)$ for all $t\in[0,T]$ and
$$|(E,L^2,L^2_z)(0)-(E,L^2,L^2_z)(T)|>\rho.$$
\end{Thm}

\subsection{Arnold diffusion for generic perturbations, conjectural picture}\label{SSConjS}
Our previous results on Arnold diffusion utilize the photon sphere and the homoclinic orbits and the diffusion mechanism is similar to Arnold's original one. In the presence of normal hyperbolicity, we expect the following more general statements are true.

 \begin{Con}\label{ConS1}
 \begin{enumerate}
\item Let $\eps h_{\mu\nu}dx^\mu dx^\nu$  be a generic stationary perturbation to the Schwarzschild metric, then for any $\dt>0$, there exists $\eps_0$ such that for all $|\eps|<\eps_0$,
 there exists an orbit on $N_\eps$ and times $T,T'$ such that $$|L_z(T)-L(0)|<\dt, \quad |L_z(T')+L(0)|<\dt.$$
\item Let $\eps h_{\mu\nu}dx^\mu dx^\nu$  be a generic perturbation $($1-periodic in $\tau)$ to the Schwarzschild metric, then  for any $\dt>0$,
there exists an orbit that visits any $\dt$-ball centered on $N\times \T^1\times (4/9,1/2)(\ni (\Xi,\tau,E))$,  provided $\eps$ is sufficiently small.
 \end{enumerate}
 \end{Con}

On the other hand, without normal hyperbolicity, we also expect the following statement on Arnold diffusion away from the photon sphere is true in the spirit of Arnold's conjecture.% For Conjecture \ref{ConS1}(2), since the Hamiltonian $H_{\hat N_\eps}$ has three degrees of freedom, the diffusion orbit may also be constructed without using the normal hyperbolicity.

 We introduce the following  part of the phase space with bounded motions, where Liouville-Arnold theorem applies
 \begin{equation}\label{EqBS}\mathcal{B}(C):=\left\{ V(r_+)<\frac{E^2}{2}<\min\{V(r_-),V(\infty)\},\quad |L_{isco}|<|L|<C\right\}\cap H^{-1}(-\frac12).\end{equation}
 We remark that the first inequality of the definition of $\mathcal{B}(R)$ is an inequality between $E$ and $L$. When we consider stationary perturbations, the particle's energy $E$ is a constant, so we treat the system as one with three degrees of freedom and $\mathcal B(C)$ is a bounded set  on the energy level $H=-\frac12$ of dimension 5. When we consider 1-periodic in $\tau$, the energy $E$ is no longer conserved, then the Hamiltonian system has four degrees of freedom and $\mathcal B(C)$ is a bounded set  on the energy level $H=-\frac12$ of dimension 7.
 \begin{Con}\label{ConS2}
 \begin{enumerate}
\item Let $\eps h_{\mu\nu}dx^\mu dx^\nu$  be a generic stationary perturbation to the Schwarzschild metric, then for any $C>0$
and any $\dt>0$, there exists an orbit that is $\dt$-dense on $\mathcal B(C)$, provided $\eps$ is small.
\item Let $\eps h_{\mu\nu}dx^\mu dx^\nu$  be a generic perturbation $($1-periodic in $\tau)$ to the Schwarzschild metric, then for any $C>0$
and any $\dt>0$, there exists an orbit that is $\dt$-dense on $\mathcal B(C)$, provided $\eps$ is small.
 \end{enumerate}
 \end{Con}

\begin{Rk}\label{RkConS}
\begin{enumerate}
\item As we have remarked before, one main difficulty in proving Conjecture \ref{ConS1} is the lack of twist for the dynamics on the NHIM.
 One difficulty for proving Conjecture \ref{ConS2} is that the system is nonconvex, so our method in \cite{CX} does not apply. %works only for convex Hamiltonian systems and not yet suitable for metric perturbations.
\item  In both conjectures, an extra difficulty comes from the requirement of  physical perturbations that solves the Einstein equation. Though it is hard to give a mathematical proof, we expect that the conjectured physical phenomena may be observable in reality.
\end{enumerate}
\end{Rk}

\section{Oscillatory orbits in perturbed Schwarzschild spacetime}\label{SOsc}
In this section, we explore case (B) in Section \ref{SSMetricS}. Far away from the blackhole, we can use the post Newtonian approximation to study the dynamics of a particle.
The orbit escaping to infinity in case (B) can be considered as an analogue of the Kepler parabolic orbit, which can be utilized to create a special orbit called oscillatory orbit.

\subsection{Oscillatory orbits in Newtonian mechanics}
 For the three-body problem with the Hamiltonian
$$H=\sum_{i=1}^3\frac{p_i^2}{2m_i}-\sum_{i<j}\frac{m_im_j}{|x_i-x_j|},\quad (x_i,p_i)\in \R^3\times \R^3,\ i=1,2,3$$
with properly chosen masses, there exists a special kind of solution $\{(x(t),p(t)),\ t\in \R\}$ such that $$\liminf_{t\to \infty} |x(t)|<C,\quad \limsup_{t\to \infty} |x(t)|=\infty. $$
One model is called Sitnikov-Alekseev model. Consider a pair of equal masses $Q_1$ and $Q_2$ moving on the $x$-$y$ plane along Kepler elliptic orbits and a massless particle (Alekseev proved that a small massive particle also works) moving on the $z$-axis attracted by the pair. Another classical model which exhibits oscillatory motions is the restricted planar circular three-body problem, which we explain in more details. This is a configuration of a massless particle moving in the gravitational field of two massive bodies called Sun and Jupiter. We assume the mass of Sun is $1-m$ and that of Jupiter is $m$, $m\in (0,1/2]$, and they move on a circular orbit. We may put the system in rotating coordinates so that Sun is fixed at $(-m,0)$ and Jupiter is fixed at $(1-m,0)$. The Hamiltonian of this system has the form in polar coordinates
$$H_m(r,R,\theta,\Theta) =\frac{R^2}{2}-\frac{1}{r}+\frac{\Theta^2}{2r^2}-\Theta+U(r,\theta),\quad (r,R,\theta,\Theta)\in \R_{\geq0}\times \R\times \T\times \R$$
with $U=\left(\frac{1}{r}-\frac{1-m }{|r(\cos\theta,\sin\theta)-(-m,0)|}-\frac{m}{|r(\cos\theta,\sin\theta)-(1-m,0)|}\right)$. Here $(r,\theta)$ is the standard polar coordinates on the plane, the variable $R$ is the radial momentum and $\Theta$ is the angular momentum. The first three terms on the RHS of $H_m$ is Kepler's two-body problem in polar coordinates and the extra term $-\Theta$ means that the system is written in rotating coordinates. The potential $U\sim \frac{ m(1-m)}{r^3}$ as $r\to\infty. $

The mechanism of having the oscillatory motion is to treat the Kepler parabolic orbit as a homoclinic orbit to the (degenerate) hyperbolic fixed point at infinity. Then in general a small perturbation will cause separatrix splitting by Poincar\'e theorem (c.f. Theorem \ref{ThmPoincare}). To elaborate a bit, let us consider $m=0$, i.e. the Kepler two-body problem written in rotating and polar coordinates
$$H_0(r,R,\theta,\Theta)=\frac12R^2-\frac{1}{r}+\frac{\Theta^2}{2r^2}-\Theta$$
with equations of motion $\begin{cases}
\frac{dr}{dt}=R\\
\frac{dR}{dt}=-\frac{1}{r^2}+\frac{\Theta^2}{r^3}
\end{cases}$ in the radial components.
We next introduce the transform $r=v^{-2}$ so that $r=\infty$ corresponds to $v=0$. In terms of $v$ and $R$ we have the following equations of motion
$\begin{cases}
\frac{dv}{dt}=-\frac{1}{2} v^3R\\
\frac{dR}{dt}=-v^4+\Theta^2 v^6.
\end{cases}$
If we make a time reparametrization $d\tau=v^3 dt$, then we get $\begin{cases}
\frac{dv}{d\tau}=-\frac{1}{2}R\\
\frac{dR}{d\tau}=-v+\Theta^2 v^3
\end{cases}.$ Therefore the point $(v,R)=(0,0)$ is a hyperbolic fixed point. Its stable and unstable manifolds still exist and coincide in the case $m=0$. For sufficiently small positive $m$, the perturbation from the potential $U$ causes the splitting of separatrix hence gives rise to a horseshoe and symbolic dynamics. This gives a broad outline of the existence of oscillatory motions. We refer readers to the literatures \cite{AKN, M,LS, Mc,GMS} etc. %The existence of oscillatory motion is then proved by evaluating the Melnikov function along the parabolic Kepler orbit Graph$\left\{R(r)=\pm\sqrt{\frac{2}{r}-\frac{\Theta^2}{r^2}}\right\}$. We refer readers to \cite{M2} for more details.
\subsection{Oscillatory orbits in perturbed Schwarzschild spacetime}\label{SSOsc}
In general relativity, the orbit in case (B) in Section \ref{SSSABC} plays a similar role of Kepler parabolic orbit.  We consider only the massive particles moving on timelike geodesics with $\mu^2=1$ and with energy $E=1$.
\subsubsection{The homoclinic orbit to infinity}\label{SSSOscHomo}
\begin{figure}
\centering
\includegraphics[width=0.4\textwidth]{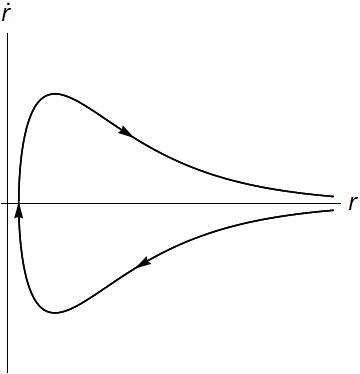}
\caption{Homoclinic orbit to infinity}
\label{FigOsc}
\end{figure}

Assuming we are in case (B) with $V(r_-)>V(\infty)$, we get the phase portrait of the orbit homoclinic to infinity for the radial equation (see Figure \ref{FigOsc})
\begin{equation}\label{EqHomoInfty} \mathrm{Graph}\left\{\frac{ dr}{dt}=\pm\sqrt{ E^2-2V(r)},\quad E^2=2V(\infty)=1\right\}.\end{equation}
Here the domain for $r$ is $[r_*,\infty)$ where $r_*$ is the largest root of $E^2-2V(r)$ and $r_*=\frac{L^2}{2M}-2M+o(1)$ when $L^2$ is large.

Differentiating \eqref{ThmRadialS}, we obtain the equations of motion is $$\begin{cases}
\frac{dr}{dt}:=R\\
\frac{dR}{dt}=-\frac{M}{r^2}+\frac{L^2}{r^3}-\frac{3ML^2}{r^4}
\end{cases}.$$
Making the change of variable $r=v^{-2}$, we obtain
$$\begin{cases}
\frac{dv}{dt}=-\frac12v^3R\\
\frac{dR}{dt}=-Mv^4+L^2v^6-3ML^2v^8
\end{cases}.$$ Now the situation is similar to the Newtonian case and we can treat $(v,R)=(0, 0)$ as a degenerate hyperbolic fixed point. As $t\to\pm\infty$, we have the asymptotic behavior $|x|,|y|\sim \frac{1}{|t|^{1/3}}$ up to a scalar multiple.

The orbit can be solved explicitly as follows (c.f. \cite{C} Section 19(c), Chapter 3). Introducing $u=1/r$ and the Mino time $d\lambda=\frac{L_z}{r^2}dt$ as before, we get from \eqref{EqHomo} by setting $E=1$
\begin{equation}\label{EquOsc}\left(\frac{du}{d\lambda}\right)^2=2Mu^3-u^2+\frac{2M}{L^2} u=2Mu(u-u_2)(u-u_3)\end{equation}
with $ u_1=\frac{2}{\ell},\ u_2=\frac{1}{2M}-\frac{2}{\ell}$ where $\ell$ is a constant such that $u_1u_2=L^{-2}$. We next introduce the Keplerian representation $u=\frac{1}{\ell}(1+\cos\chi)$ assuming the orbit is parabolic, where $\chi$ is a function of $\lambda$. The domain of $\chi$ is $[0,2\pi)$ and $\chi=\pi$ corresponds to $r=\infty$. Substituting the representation of $u$ into \eqref{EquOsc}, we get
$$\left(\frac{d\chi}{d\lambda}\right)^2=(1-4\mu)-4\mu \cos^2\frac{\chi}{2},$$
where $\mu=M/\ell\leq 1/8$ if we assume $u_1\le u_2$, whose solution is
\begin{equation}\label{Eqlambda}\lambda(\chi)=\frac{2}{(1-4\mu)^{1/2}}\left[K(k)-F\left(\frac\pi2-\frac\chi2,k\right)\right],\end{equation}
where  $k^2=\frac{4\mu}{1-4\mu}\leq 1$, and $K$ and $F$ are complete and incomplete elliptic integrals of the first kind respectively  (i.e. $F(\varphi, k)=\int_0^\varphi\frac{d\theta}{\sqrt{1-k^2\sin^2\theta}}$ and $K(k)=F(\pi/2,k)$, with $x = \mathrm{sn}(u,k)$ one has $F(x,k) = u)$. Inverting the function $\lambda(\chi)$, we obtain $\chi$ as a function of $\lambda$, hence we finally obtain $u$ hence $r$ as a function of the Mino time $\lambda$. Denote by $[-\Lambda,\Lambda]$ the domain of $\lambda$, which is a bounded interval.

We notice that an orbit of the Schwarzschild dynamics necessarily lies on a plane passing though the origin due to the conservation of $L$ and $L_z$. We thus have the freedom to choose $\theta=\pi/2$ then $p_\theta=0$ by setting the orbital plane as the equatorial plane. For simplicity, we will assume that the perturbed system also preserves the equatorial plane, i.e. we have $\theta=\pi/2$, $p_\theta=0$ and $L_z=L$ for all time. This reduces the system to one with two degrees of freedom whose phase space is parametrized by $(r,p_r,\varphi,L)$.
\begin{Not}
%We next denote by $\Xi=(\theta, p_\theta, \varphi,p_\varphi)$ and by  $\Xi_0=(\theta(0), p_\theta(0), \varphi(0),p_\varphi(0))$ the initial condition with $\theta=\theta(0)=\frac\pi2,\ p_\theta=p_\theta(0)=0,\ p_\varphi=p_\varphi(0)=L_z=L$. Then we get $\Xi(\Xi_0,\lambda)=\Xi_0+(0,0,\frac{\lambda}{L^2},0)$ as the values of $\Xi$ as a function of the Mino time $\lambda$ and the initial condition $\Xi_0$. We next
Denote by $\Gamma(\nu+\lambda)=(r(\nu+\lambda), p_r(\nu+\lambda))$ the homoclinic orbit to infinity on the $(r,p_r)$-plane parametrized by $\lambda$ and with initial condition $\Gamma(\nu)$, and by $\varphi(\varphi_0,\lambda)=\varphi_0+\frac{\lambda}{L^2}$. Therefore $(\Gamma(\nu+\lambda), \varphi(\varphi_0,\lambda)),\ \lambda\in [-\Lambda,\Lambda]$ is the orbit parameter of the unperturbed system with initial condition $(\Gamma(\nu),\varphi_0)$.
\end{Not}
With these notations, we again have the following Melnikov potential $J(\nu,\Xi_0)$ formally the same as \eqref{EqMelnikov}
\begin{equation}
\label{EqMelnikovOsc}J(\nu,\varphi_0)=\frac{1}{L_z}\int^{\Lambda-\nu}_{-\Lambda-\nu}  r(\nu+\lambda)^2H_1(\Gamma(\nu+\lambda),\varphi(\varphi_0,\lambda))-\lim_{r\to\infty } r^2 H_1(r,\varphi(\varphi_0,\lambda))\,d\lambda.\end{equation}

\subsubsection{The degenerate NHIM at infinity}
We denote by $N_\infty=\{(r,R)=(\infty,0)\}$ the invariant manifold at infinity with degenerate normal hyperbolicity and \eqref{EqHomoInfty} gives the homoclinic orbit to $N_\infty$. We consider only stationary perturbations so the Hamiltonian is independent of $\tau$. Moreover, the equatorial plane is preserved, so fixing $E=1$, then the phase space has four dimensions and $N_\infty$ has two dimensions parametrized by $(\varphi,L_z)$. Moreover, from \eqref{EqHamS}, we see that each point in $N_\infty$ is fixed and the homoclinic orbit \eqref{EqHomoInfty} is attached to each individual point in $N_\infty$.  We next consider a perturbation $\eps h_{\mu\nu}dx^\mu dx^\nu$ to the Schwarzschild metric, whose perturbation to the Schwarzschild Hamiltonian has the form $\eps H_1= -\frac\eps2 g^{\mu\nu}h_{\mu\nu}g^{\mu\nu}p_\mu p_\nu+O(\eps^2)$. We assume that $H_1$ converges to 0 asymptotically like $O(\frac{1}{r^2})$ as $r\to\infty.$ This guarantees the convergence of the Melnikov functional $J(\nu,\varphi_0)$, since we have that $H_1=O(1/r^2)=O(v^4)=O(1/|t|^{4/3})$ along the homoclinic orbit as $t\to\infty$, which is integrable with respect to $t\in \R$ (or with respect to the Mino time $\lambda$, the function $\lim_{r\to\infty }r^2H_1(r,\varphi(\varphi_0,\lambda))$ is bounded and the domain of integration with respect to $\lambda$ is bounded seen from \eqref{Eqlambda}). %The result of McGehee \cite{Mc} then applies to yield that the stable and unstable manifolds still exist in the perturbed system, despite of the degeneracy of the hyperbolicity at infinity.

\subsubsection{Existence of oscillatory orbits}
We next state our result on the existence of oscillatory orbits. We will construct explicit perturbation in Appendix \ref{SReggeWheeler} satisfying the assumptions.

\begin{Thm}\label{ThmOsc}
Let $\eps h_{\mu\nu}dx^\mu dx^\nu$ be a stationary perturbation of the Schwarzschild metric satisfying
\begin{enumerate}
\item the perturbation $H_1= \frac12  g^{\al\mu}h_{\mu\nu}g^{\beta\nu}p_\al p_\beta+O(\eps)$ to the Hamiltonian converges to 0 asymptotically like $O(\frac{1}{r^2})$, and $\lim_{r\to\infty} r^2 H_1$ is a constant;
\item The subspace $\{\theta=\pi/2,\ p_\theta=0\}$ is invariant under the perturbed Hamiltonian system;
\item $L^2$ is sufficiently large and $\partial_\nu J(\nu,\varphi_0)$, has a nondegenerate zero for $\varphi_0=0$,
\end{enumerate}
 then there exists an orbit of oscillatory type for $\eps>0$ sufficiently small, i.e. there exists a constant $C$ such that along the orbit we have
$$\liminf_{t\to\infty} r(t)<C,\quad \limsup_{t\to\infty} r(t)=\infty.$$
\end{Thm}
\begin{proof}
Since we consider only stationary perturbations, we can fix $E=1$ thus reduce the Hamiltonian system to a system of three degrees of freedom. Next by assumption (2), we further ignore the $\theta,p_\theta$. The unperturbed Schwarzschild system has in fact only one degree of freedom that is the radial motion and the perturbation introduces the time periodic perturbation through $\varphi$, so the situation is similar to the pendulum case in Section \ref{SSSPendulum}. % components to reduce the system to one with two degrees of freedom with coordinates $(r,p_r,\varphi,p_\varphi)$.

The existence of stable and unstable manifolds to the degenerate hyperbolic fixed point $(v,R)=(0,0)$ is given by \cite{Mc}. The transversal intersection of the stable and unstable manifolds is proved in the same way as Appendix \ref{AppPoincare} under the assumption (3). The fact that the transversal intersection gives rise to oscillatory orbit follows from a classical symbolic dynamics argument in \cite{M,LS}. We require $L^2$ is sufficiently large, which would become clear in Appendix \ref{AppOsc} where we verify the assumptions for a concrete perturbation. This assumption guarantees that the $\varphi$-motion is fast rotating and in physical space, this means that the particle's parabolic orbit is sufficiently far from the blackhole. We require $\eps$ small so that the Melnikov nondegeneracy implies the transversal intersection of the stable and unstable manifolds (c.f. Appendix \ref{AppPoincare}).
\end{proof}
%When a concrete perturbation is given, the statement can be verified by computing the Melnikov function. When the Melnikov function is nonvanishing, then the proof becomes standard following \cite{M2}.
%\begin{Thm}\label{ThmOscII}
%Let $\eps h_{\mu\nu}dx^\mu dx^\nu$ be a stationary perturbation of the Schwarzschild metric solving the linearized vacuum Einstein equation and satisfying
%\begin{enumerate}
%\item $H_1=\frac12  g^{\al\mu}h_{\mu\nu}g^{\beta\nu}p_\al p_\beta$  converges to 0 asymptotically like $O(\frac{1}{r^2})$ as $r\to\infty;$
%\item $L^2$ is sufficiently large and $\partial_\nu J(\nu,\Xi_0)$, has a nondegenerate zero for each $\Xi_0\in U\subset N_\infty$.
%\end{enumerate}
%Then for $\eps$ sufficiently small, there exists an orbit of oscillatory type, i.e. there exists a constant $C$ such that along the orbit we have
%$$\liminf_{t\to\infty} r(t)<C,\quad \limsup_{t\to\infty} r(t)=\infty.$$
%Moreover, there exists constant $\rho$ and time $T$ such that along the orbit we have
%$$|L_z(0)-L_z(T)|>\rho. $$
%\end{Thm}
%\begin{proof}
%Assumption (2) allows us to define the scattering map $\mathcal S:\ N_\infty\to N_\infty$ satisfying Theorem \ref{ThmScatteringS1} adapted to our current setting. The diffusing behavior in the ``Moreover" part follows from Theorem 4.1 of \cite{GLS} as before.
%\end{proof}

\begin{Rk}
\begin{enumerate}
\item When we remove the assumption (2) in the last theorem, the Hamiltonian system has six degrees of freedom and $N_\infty$ is four dimensional parametrized by $(\theta,p_\theta, \varphi,p_\varphi)$. The mechanism of Arnold diffusion also works by exploiting the homoclinic orbits to infinity to yield an orbit of oscillatory type and it simultaneously has a big oscillation in the $L_z$ component. We skip the statement here, but it is not hard to reproduce corresponding results in the Newtonian case to our setting here (c.f. \cite{DKRS}). We can also consider 1-periodic perturbations in $\tau$  similar to Theorem \ref{ThmDiffusionNonStatS}.%The dynamics in the configuration space is as follows. Arnold diffusion applied to the oscillatory motion produces orbits that are oscillatory in the radial component and almost dense when projected to the spherical components.

\item We remark that the reasoning fails for massless particles.
\end{enumerate}
\end{Rk}

\subsection{Case (C) in Section \ref{SSSABC}}
We finally remark briefly on Case (C) in Section \ref{SSSABC}. This case occurs when $E^2=1$ and $L^2=L_{homo}^2$ so that $V(r_c)=V(\infty)$. In this case, the NHIM $N=\{r=r_c, p_r=0\}$ appear simultaneously as the NHIM $N_\infty=\{r=\infty, p_r=0\}$ at infinity. There are heteroclinic orbits connecting them in the Schwarzschild Hamiltonian system. When a stationary perturbation that decays at least inverse quadratically in $r$ is considered, the NHIM $N$ is deformed into $N_\eps$ and the heteroclinic orbit is broken so that $W^s(N_\eps)\pitchfork W^u(N_\infty)$ and $W^u(N_\eps)\pitchfork W^s(N_\infty).$ This induces scattering maps $\mathcal S_1:\ N_\eps\to N_\infty$ as well as $\mathcal S_2:\ N_\infty\to N_\eps$, whose first order approximations can be obtained by the Melnikov method as in Theorem \ref{ThmScatteringS1}. In this case, Arnold diffusion behavior (by zoom-whirl orbit around the photon sphere) and oscillatory behavior coexist. Statement can be formulated by combining Theorem \ref{ThmGLS} and Theorem \ref{ThmOsc}. We skip the details.

\section{Arnold diffusion in perturbed Kerr spacetime}\label{SK}
In this section, we study Arnold difffusion in perturbed Kerr spacetime.
\subsection{The Kerr spacetime and the Penrose process}\label{SSK}
The Kerr spacetime in the standard Boyer-Lindquist coordinates has the form
$$ds^2=-\left( 1-\frac{2Mr}{\Sigma}\right)d\tau^2-2\left( \frac{2Mr}{\Sigma}\right) a\sin^2\theta d\tau d\varphi+\Sigma \left( \frac{dr^2}{\Delta}+d\theta^2\right) +\frac{\mathcal A}{\Sigma}\sin^2\theta d\varphi^2,$$
where
\begin{equation}
\begin{aligned}
\Sigma&=r^2+a^2\cos^2\theta,\\
\Delta&=r^2+a^2-2Mr,\\
\mathcal A&=(r^2+a^2)^2-\Delta a^2\sin^2\theta.
\end{aligned}
\end{equation}
Here $M$ is the mass and $a$ is the angular momentum of the blackhole with $0\leq|a|\leq M$. When $a=0$, the Kerr metric reduces to the Schwarzschild metric.
The event horizon is where the metric coefficient $g_{rr}=\frac{\Sigma}{\Delta}$ becomes singular. Let $r_h=M+ \sqrt{M^2-a^2}$ be the larger root of $\Delta$, hence the outer event horizon is the sphere $r=r_h$.
The metric coefficient $g_{tt}=\left( 1-\frac{2Mr}{\Sigma}\right)$ changes sign when $r=r_e,$ $r_e=M+\sqrt{M^2-a^2\cos^2\theta}$ be the larger root of $\Sigma-2Mr$. The region $\{r_h<r<r_e\}$ is called the ergosphere.
Both singularities are coordinate singularities.

Note that in the interior of the ergosphere, the $\tau$-component becomes spacelike.  Therefore, a moving massive particle within the ergosphere must co-rotate with the blackholewith an angular velocity $$\Omega=-{\frac {g_{t\phi }}{g_{\phi \phi }}}={\frac {2Mra}{\Sigma \left(r^{2}+a^{2}\right)+2Mra^{2}\sin ^{2}\theta }}$$
 in order to retain its timelike character (the negativity of $ds^2=-\mu^2$ comes from the $dtd\varphi$ term of the metric).
Therefore the particle gains energy and angular momentum. Because it is still outside the event horizon, it may escape the black hole. The net process is that the rotating black hole emits energetic particles and loses its own energy and angular momentum. This is called the Penrose process (c.f. \cite{C}).

 The geodesic equations are of the following form (assuming $ds^2=-\mu^2$).
\begin{equation}\label{EqofMotion}
\begin{cases}
%&\mu\rho^2\dfrac{dt}{d\tau}=\dfrac{r^2+a^2}{\Delta} P-a\\
\Sigma^2(\frac{dr}{dt})^2&= R,\\
\Sigma^2(\frac{d\theta}{dt})^2&=\Theta,\\
\Sigma\frac{d\tau}{dt}&=-a(aE\sin^2\theta-L_z)+\frac{r^2+a^2}{\Delta}P(r),\\
\Sigma\frac{d\varphi}{dt}&=\frac{a}{\Delta}P(r)-\left(aE-\frac{L_z}{\sin^2\theta}\right),
\end{cases}
\end{equation}
where we have
\begin{equation}\label{EqRTheta}
\begin{aligned}
&R(r)= P(r)^2-\Delta(\mu^2 r^2+(L_z-aE)^2+Q),\\
&\Theta(\theta)=Q-\left[(\mu^2-E^2)a^2+\dfrac{L_z^2}{\sin^2\theta}\right]\cos^2\theta,\\
&P(r)=E(r^2+a^2)-aL_z,
\end{aligned}
\end{equation}
and the constant $Q=p_\theta^2+\cos^2\theta\left(a^2(\mu^2-E^2)+\left(\frac{L_z}{\sin\theta}\right)^2\right)$ is a constant of motion called \emph{Carter constant}.

Introducing a time reparametrization $d\lambda=\Sigma(r,\theta)^{-1}dt$ called Mino time, we get the equations of motion
\begin{equation}\label{EqMino}
\begin{cases}
\frac{dr}{d\lambda}&=\pm\sqrt{R(r)}=p_r\Delta,\\
\frac{d\theta}{d\lambda}&=\pm\sqrt{\Theta(\theta)}=p_\theta,\\
\frac{d\tau}{d\lambda}&=-a(aE\sin^2\theta-L_z)+\frac{r^2+a^2}{\Delta}P(r),\\
\frac{d\varphi}{d\lambda}&=\frac{a}{\Delta}P(r)-\left(aE-\frac{L_z}{\sin^2\theta}\right),\quad \dfrac{dL_z}{d\lambda}=0.
\end{cases}
\end{equation}

We next introduce the Hamiltonian formalism. The metric can be thought of as twice a Lagrangian and we get its corresponding Hamiltonian via Legendre transform $H(p_{x^\mu},x^\mu)=p_{x^\mu}\frac{d x^\mu}{dt}-L(x^\mu,\frac{dx^\mu}{dt})$ where $x^\mu=(\tau, r,\theta,\varphi)$ and $p_{x^\mu}=(p_\tau,p_r, p_\theta,p_\varphi)$ with
\begin{equation}\label{EqLegendreK}
\begin{cases}
p_\tau&=E=-\left( 1-\frac{2Mr}{\Sigma}\right)\dot\tau-\left( \frac{2Mr}{\Sigma}\right) a\sin^2\theta  \dot\varphi,\\
p_r&=\Sigma \frac{\dot r}{\Delta},\\
 p_\theta&=\Sigma\dot \theta,\\
 p_\varphi&=L_z=-\left( \frac{2Mr}{\Sigma}\right) a\sin^2\theta \dot\tau +\frac{\mathcal A}{\Sigma}\sin^2\theta \dot \varphi.
 \end{cases}
 \end{equation}
 The Hamiltonian $H$  can be written as
\begin{equation}\label{EqHamK}
\begin{aligned}
& H=\frac12\Sigma^{-1}(H_r+H_\theta),\\%\quad \rho^2=r^2+a^2\cos^2\theta.\\
& H_r=\Delta p_r^2-\dfrac{1}{\Delta}((r^2+a^2)E-aL_z)^2,\\
& H_\theta=p_\theta^2+\dfrac{1}{\sin^2\theta}(a\sin^2\theta E-L_z)^2.\\
%&\Delta_r=(r^2+a^2)-2Mr.\\
%&z=\cos\theta
\end{aligned}
\end{equation}
%The Carter constant
%$$Q = \Theta^{2} + \cos^{2}\theta \Bigg( a^{2}(m^{2} - E^{2}) + \left(\frac{L_z}{\sin\theta} \right)^{2} \Bigg),$$
The system has four independent constants of motion $H,E,p_\varphi,Q$. The conservation of $E,\ p_\varphi$ is obvious since  the Hamiltonian does not depend on $\tau$ and $\theta$ explicitly. They have the physical meanings of \emph{energy} and \emph{the  $z$-component of the angular momentum} of the moving particle in Kerr background respectively. %In the following, we use $L_z$ for $p_\varphi$ and $E$ for $p_\tau$.  %Moreover, both $H_r$ and $H_\theta$ are Hamiltonian systems of one degree of freedom for fixed $E$ and $L_z$ hence are both integrable.

Let $\mathcal H=H_r+H_\theta+\Sigma$. We treat $\mathcal H$ as a Hamiltonian with time variable $\lambda$. Then the Hamiltonian flow of $\mathcal H$ on energy level $\mathcal H=0$ coincide with that of $H$ on the energy level $-\frac12$. Note that $\mathcal H= \Delta p_r^2-\frac{R}{\Delta}+Q$, hence the radial and latitudinal motions separate.
\subsection{The photon shell}

 Similar to the photon sphere in Schwarzschild case, the Kerr spacetime also admits unstable orbits with constant radial components. However, the dynamics of constant radius orbits in the Kerr spacetime differs drastrically from the Schwarzschild case. Such an orbit in the Schwarzschild case lies on a fixed plane through the origin thus is periodic and all these orbits form a sphere. However, in the Kerr case, such an orbit is no longer  confined to a plane but may oscillate in a spherical strip and may be quasiperiodic. Moreover, depending on the parameters $E,L_z, Q$, the radius varies, so the union of all such orbits forms a ring called \emph{photon shell}.
 The photon shell for Kerr was first studied by \cite{Wi} in the extreme case $a=M$ for timelike geodesics. The lightlike case was studied by \cite{T}. The equatorial case of $Q=0$ was studied by \cite{GLP}. In this section, we perform a study in the $a\leq M$ case for timelike geodesics extending the analysis in Section 64 of \cite{C}. For the purpose of Arnold diffusion, we do not consider the null geodesics since in that case there is no homoclinic orbits associated to the photon shell (c.f. Lemma 4.3 of \cite{X}).

 %Under the assumptions that either $a = 0,\ 9\Lambda M^2 <1$ or $\Lambda=0,\ |a|<M$, and for nearby values of $M, \Lambda, a$.
%To simplify the study, we only consider the case $\Lambda =0$. The $\Lambda>0$ but very small case will follow immediately.
%The Kerr metric is a generalization to the case of a positive cosmological constant $\Lambda$. A positive cosmological constant can be introduced in which case, the metric is called Kerr-de Sitter metric. In this paper we consider only the case of zero cosmological constant to simplify the formulas. Our result applies directly to the case of nonzero but small cosmological constant by continuity. However, if the cosmological constant is not small, the phase space structure may be deformed and case study is needed.

%We consider time-like geodesic with total energy $-\mu^2/2<0$. Our method works as well in the Schwartzchild case ($a=0$) for null-geodesic ($\mu=0$). However, in the strict Kerr case $(a>0)$ and null-geodesic ($\mu=0$), we do not know how to construct diffusing orbit since we are not able to find homoclinic orbits for the unperturbed system.

%\subsubsection{The radial and polar dynamics}
%From the metric or the Hamiltonian, we derive the equations of motion for a particle moving in the Kerr spacetime.

\subsubsection{The radial dynamics}
Orbits in the photon shell have constant radial component. So from the radial equation of \eqref{EqofMotion}, the following equations should be satisfied.
\begin{equation}\label{EqCrit}R(r)=0,\quad \frac{d}{dr}R(r)=0. \end{equation}
Let $r_c$ be a solution. Moreover, orbits in the photon shell are unstable, which implies the second order derivative $\frac{d^2}{dr^2}R(r_c)>0$.

Introducing $\ell=\frac{L_z}{E}$ and $q=\frac{Q}{E^2}$, we write explicitly equations \eqref{EqCrit} as
\begin{equation*}
\begin{aligned}
& 3r^4+r^2a^2-q(r^2-a^2)-\frac{r^2}{E^2}(3r^2-4Mr+a^2)=r^2\ell^2,\\
& r^4-a^2Mr+q(a^2-Mr)-\frac{r^3}{E^2} (r-M)=rM(\ell^2-2a\ell).
\end{aligned}
\end{equation*}
The two equations involve four variables $r,E,\ell,q$. We solve for $\ell$ and $q$\begin{equation}\label{Eqellq}
\begin{aligned}
\ell&=\frac{1}{a(r-M)}\left( M(r^2-a^2)\pm r\Delta \left( 1-\frac{1}{E^2}(1-\frac{M}{r})\right)^{1/2}\right),\\
q &=\frac{1}{a^2(r-M)}\left[\frac{r^3}{r-M}(4a^2M-r(r-3M)^2)+\frac{r^2}{E^2}\left[r(r-2M)^2-a^2M\right]\right.\\
&\left.-\frac{2r^3M}{r-M}\Delta\left(1\pm \left(1-\frac{1}{E^2}(1-\frac{M}{r})\right)^{1/2}\right)\right].
\end{aligned}
\end{equation}
In both equations we pick the solution with the minus sign since we require $Q$ to be nonnegative for the photon spherical orbits. The functions $\ell$ and $q$ are plotted in Figure \ref{fig:side:aChandra} with data $M=1, a=0.8, E=1$ and Figure \ref{fig:side:bChandra} with $M=1, a=0.8, E=0.92$.  For nearby choices of parameters the picture is qualitatively similar since the functions are continuous in $E^2$ and $a$ when they are real functions. In order for $q$ and $\ell$ to be real functions for $r>r_h$, we have to assume $E^2>1-\frac{M}{r_h}$.

\begin{figure}
\begin{minipage}[t]{0.5\linewidth}
\centering
\includegraphics[width=2.in]{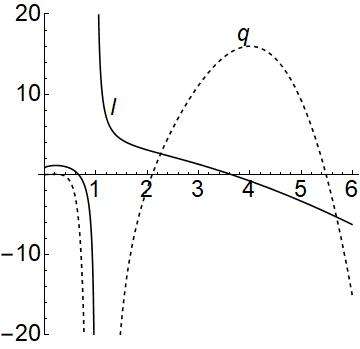}
\subcaption{$M=1,a=0.8,E=1$}
\label{fig:side:aChandra}
\end{minipage}%
\begin{minipage}[t]{0.5\linewidth}
\centering
\includegraphics[width=2.in]{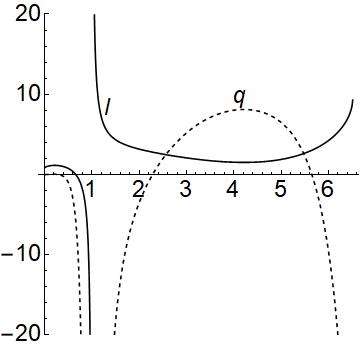}
\subcaption{$M=1,a=0.8,E=0.92$}
\label{fig:side:bChandra}
\end{minipage}
\caption{Graphs of $\ell$ and $q$}
\end{figure}

%We next determine the admissible range for the parameters $E,r,\ell$.
%\begin{Def}\label{DefAdmL}\begin{enumerate}
%\item We consider only bound orbit so we assume $E^2<1$. To make sure the above functions $q$ and $\ell$ take real values for all $r>r_h$, we define the initial admissible range of $E^2$ to be $[ 1-\frac{M}{r_h},1)$.
%\item For each $E^2$ in its initial admissible range, we determine $r_-(E)<r_+(E)$ as zeros of the function $q(r)$ and define the admissible range $[r_-(E),r_+(E)]$ for $r$. In particular, we have for all $a\in [0,1],$ the value of $q$ is strictly negative for $r=r_h=M+\sqrt{M^2-a^2}$, so that the radial interval $[r_-(E),r_+(E)]$ is strictly outside the event horizon. This is also the radial range for the photon shell if we have $R''(r)>0$ for $r\in [r_-(E),r_+(E)]$ and $\ell,q$ be as in \eqref{Eqellq}.
%\item Evaluating the function $\ell(r)$ on the interval $[r_-(E),r_+(E)]$, we then obtain the admissible range $[\ell_-(E),\ell_+(E)]$ for $\ell$. These ranges are explicit in Figure.
%\item The admissible range for $q$ is defined similarly as item $(3)$.
%\end{enumerate}
%\end{Def}

Given $r=r_c$, it determines the values of $\ell=\ell_c,q=q_c$ such that $r_c$ solves equation \eqref{EqCrit} with parameters $\ell_c$ and $q_c$. So for this choice of parameters $\ell_c, q_c$, the radius $r_c$ is a double root of $R(r)=0$, thus $R$ can be put in the form
\begin{equation}\label{EqRFactor}
R(r)/E^2=(r-r_c)^2\left(Ar^2 +2rB+C\right),\end{equation}
where we define $A=1-1/E^2<0, B=r_c\left( 1-\frac{1}{E^2}\left(1-\frac{M}{r_c}\right)\right), C=-\frac{a^2}{r_c^2} q_c.$

We next introduce the following admissible set of radii of bound spherical orbits.

\begin{Def}\label{DefAdmE}
Let $M>0,\ |a|\in (0,M]$ and $E\in (1-Mr_h^{-1},1)$ be given parameters. We define the admissible set $\mathcal R:=\mathcal R(M,a,E)$ as the radius $r_c$ satisfying the following:
\begin{enumerate}
\item $r_h<r_c\leq \frac{M}{1-E^2}$ so that $r_c$ lies outside the event horizon and the square root in \eqref{Eqellq} makes sense;
\item $q(r_c)\geq 0;$
\item inequality $Ar_c^2+2Br_c+C>0$ holds, so that $r_c$ is a local max for $-R$.
\end{enumerate}
\end{Def}

%We are interested in the unstable case, i.e. \eqref{EqCrit} with $\frac{d^2}{dr^2}R(r_c)>0$. The equation \eqref{EqCrit} is quite complicated to analyze in general.

%We are interested in the case of $0<a\leq M$ and $\mu\neq 0$ without the condition $Q=0$. The reason is as follows. When $\mu=0$, there is no homoclinic orbits to the photon sphere as we can see from the Schwarzschild case. When $Q=0$, all orbits are on the equatorial plane, so we have to remove this restriction if we would like to see the effect of Arnold diffusion on the angular momentum.

There is a way to study the equations \eqref{EqCrit} in a way analogous to the Schwarzschild case. Let us elaborate it as follows.
The function $R$ is a function of $r,E,L_z,Q$. We treat $Q$ and $L_z$ as fixed constants and think $R=R(r,E)$ as a function of $r,E$. We solve for $E=E(r)$ in the equation $R(r,E)=0$ and introduce the function $\frac{E^2}{2}=V_{\mathrm{eff}}(r)$ called the effective potential. From the implicit function theorem, we have $$\partial_r R+E'\partial_E R=0,\quad \partial^2_r R+(E')^2\partial^2_E R+E'\partial^2_{Er} R+E''\partial_E R=0.$$
If $r_c$ is such that $\partial_rR(r_c,E)=0$, and $E\neq 0$, $\partial_E R(r_c,E)\neq 0$, then we get that $V_{\mathrm{eff}}'(r_c)=0=E'$ so $r_c$ is a critical point of $V_{\mathrm{eff}}$.  Furthermore, if $E\cdot \partial_E R(r_c,E)<0$, then $V''_{\mathrm{eff}}(r_c)$ has the same sign as $\partial^2_r R$ and if $E\cdot \partial_E R(r_c,E)>0$, then $V''_{\mathrm{eff}}(r_c)$ has the opposite sign as $\partial^2_r R$. So in the case $E\cdot \partial_E R(r_c,E)<0$ we can think $V_{\mathrm{eff}}$ effectively as the $V$ in the Schwarzschild case (Figure \ref{FigV}) for the purpose of locating critical points and determining the stability. %See Figure.
%We have
%$$\frac{\partial R}{\partial E}=2P(r)(r^2+a^2)-2\Delta (a^2 E-aL_z). $$
We refer readers to \cite{Wi,GLP} for such a treatment.  %to introduce new conserved quantities $L^2=L_z^2+Q,\ \cos\iota =\frac{L_z}{L}$ as the effective total angular momentum and inclination angle. Then we have explicitly
%$$E=\frac{\left(2a\cos\iota L r+\left(r(a^2+(-2+r)r)(r^3(L^2+r^2)+a^2(2+r)(L^2-\cos^2\iota L^2+r^2)) \right)^{1/2}\right)}{r(r^3+a^2(2+r))}.$$

%The range of parameters is determined as follows. We require $E^2<\mu^2$ and $L>L_{isco}$.

%In the two equations \eqref{EqCrit}, there are variables $r,E,L_z,Q$. We perform a rescaling $\hat \mu=\mu/E,\ q=Q/E^2, \ell=L_z/E$, then we obtain
%\begin{equation}
%\begin{aligned}
%R(r)=0\Longrightarrow 0&=(r^2+a^2-a\ell)^2-\Delta(\hat\mu^2 r^2+(\ell-a)^2+q);\\
%&=(2Mr-r^2) \ell^2 -4Mar\ell -\Delta q+rP_1(r)\\
%P_1(r)&=((1-\hat \mu^2)r^3+2M\mu^2 r^2+a^2r(1-\hat\mu^2)+2Ma^2)\\
%R'(r)=0\Longrightarrow 0&=2r(r^2+a^2-a\ell)-(r-M)(\hat\mu^2 r^2+(\ell-a)^2+q)-(\hat\mu^2 r)\Delta\\
%&=-(r-M)\ell^2-2aM\ell-q(r-M)+P_2(r)\\
%P_2(r)&=2(1-\mu^2)r^3+3M\mu^2 r^2+(1-\mu^2)a^2 r+Ma^2.
%\end{aligned}
%\end{equation}
%This gives
%$$\ell=\frac{2Ma(a^2-r^2)\pm 2a\Delta \sqrt{r((1-\hat \mu^2)r+M\hat\mu^2)}}{2 a^2(r-M)},\quad  a^2(M-r)q=-2Mar^2\ell+P_3$$
%$P_3=r^2(4M^2\hat \mu^2r - Ma^2\hat \mu^2 + 2Ma^2 - 4M\hat \mu^2r^2 + 3Mr^2 + \hat \mu^2r^3 - r^3)$
\subsubsection{The vertical dynamics}\label{SSSVertical}

The vertical dynamics is determined by the equation \begin{equation}\label{EqVerticalK}(\frac{d\theta}{d\lambda})^2=\Theta=Q-U(\theta),\ \mathrm{where}\ U(\theta)=\left[(\mu^2-E^2)a^2+\dfrac{L_z^2}{\sin^2\theta}\right]\cos^2\theta.\end{equation} So the dynamics can be visualized as a particle moving in the potential well of $U(\theta)$ .

 The case of $L_z=0$ differs drastically from $L_z\neq 0$ case. Indeed, if $L_z=0$ and $Q>(\mu^2-E^2)a^2>0,\ \mu^2=1$, the orbit goes through the south and north poles. When $L_z\neq0$, as we shall see next, each orbit is bounded away from the the south and north poles. In the following, we focus mainly on the $L_z\neq 0$ case.
\begin{figure}
\centering
\includegraphics[width=0.3\textwidth]{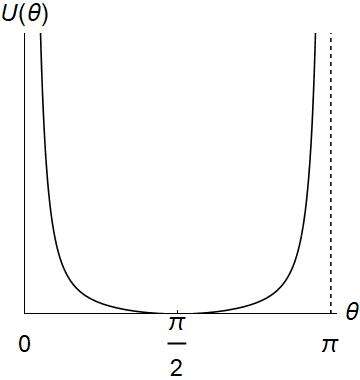}
\caption{Graph of $U$}
\label{FigU}
\end{figure}

In the case $(\mu^2-E^2)>0$ and $\mu^2=1$, we have $U(\theta)\geq 0$ is a single well potential (see Figure \ref{FigU}) and $Q\geq 0$. The minimum $\min U=0$ is attained at $\theta=\pi/2$. We also have $U(\pi/2+\theta)=U(\pi/2-\theta)$ for $\theta\in [0,\pi/2)$ and $U(\theta)\to\infty$ for $\theta\to 0,\pi.$

Let $\theta^-<\theta^+$ be two root of the equation $\Theta(\theta)=0$ with $Q>0$. Then the orbit on the photon sphere with conserved quantities $Q,E,L_z$ oscillates within the spherical strip $\theta\in [\theta^-,\theta^+]$.

The explicit solution of $\theta(\lambda)$ is obtained as follows. Denoting by $v=\cos\theta,\ \al^2=a^2(E^{-2}-1)>0$ and $$\Theta_v=q-(\ell^2+q+\al^2)v^2+\al^2 v^4,$$
and by $v^2_\pm$ the two roots of the equation $\Theta_v=0$ with $v_-^2<v_+^2$.  Then we get
$$\lambda-\nu=\int_0^\theta \frac{d\theta}{\sqrt\Theta}=E^{-1}\int_0^v \frac{dv}{\sqrt\Theta_v}=\frac{1}{E\al}\int_0^v\frac{dv}{\sqrt{(v^2-v_+^2)(v^2-v_-^2)}}=\frac{1}{E\al v_+} F(\psi,k),$$
where $\nu$ is the integral constant, $k=\frac{v_-}{v_+}$ and $\sin \psi=\frac{v}{v_-}$, and $F$ is the incomplete elliptic integral of the first kind. This gives us the function of $\lambda$ in terms of $v$ hence of $\theta$, then inverting it we obtain the function $\theta(\lambda)$.

\subsubsection{The Hamiltonian restricted to the photon shell}
When restricted to the photon shell \begin{equation}
\label{EqNHIMK}N:=\{r=r_c,p_r=0,\ r_c\in \mathcal R\}
\end{equation}
and fixing an energy $E\in (1-\frac{M}{r_h},1)$, the Hamiltonian system has two degrees of freedom. The dynamics of this Hamiltonian system is analyzed in \cite{X}. %becomes the Hamiltonian subsystem
%\begin{equation}\label{EqHamPhotonK}
%\begin{aligned}
%& H_N=\Sigma(r_c,\theta)^{-1}\left(-\dfrac{((r_c^2+a^2)E-aL_z)^2}{\Delta(r_c)}+p_\theta^2+\dfrac{1}{\sin^2\theta}(a\sin^2\theta E-L_z)^2\right),\\%\quad \rho^2=r^2+a^2\cos^2\theta.\\
%\end{aligned}
%\end{equation}
% obtained from the Hamiltonian \eqref{EqHam} by setting $r=r_c, p_r=0$. We should further eliminate the $r_c$-dependence by substituting $r_c$ as a function of $L_z$ by inverting the first equation in \eqref{Eqellq} on each of its monotone intervals.  Thus we get a Hamiltonian system of two degrees of freedom with coordinates $(\theta, p_\theta, \varphi, L_z)$ integrable in the Liouville-Arnold sense.%so we can introduce action angle coordinates as follows. %$\al_\theta, J_\theta$ and $\al_\varphi$, $\al_\varphi$ with angular variables $\al_\theta,\al_\varphi\in \T$ and $J_\theta,J_\varphi$ in a subset of $\R$, to write the Hamiltonian as a function of $J_\theta$ and $J_\varphi$ independent of the angular variables.
\subsection{Arnold diffusion via zoom-whirl orbits and Penrose process}\label{SSADK}
In this section, we study Arnold diffusion utilizing the homoclinic orbits to the photon shell. We always assume $E^2<1$.
\subsubsection{The photon shell}
We look at the radial motion in \eqref{EqofMotion}. After a time rescaling $d\lambda:=\frac{1}{\Sigma} d\tau$ by the Mino time, the radial equation of motion becomes $\left(\frac{dr}{d\lambda}\right)^2- R(r)=0$. Thus  the dynamics which can be visualized as a particle moving in the potential well of $-R(r)$.

\begin{figure}
\centering
\includegraphics[width=0.35\textwidth]{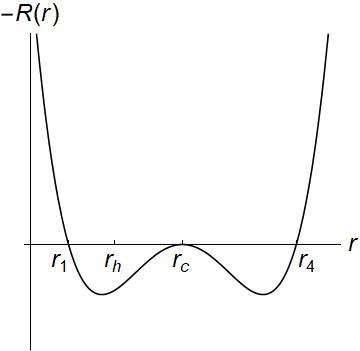}
\caption{Graph of $-R$}
\label{FigR}
\end{figure}

Suppose the quartic polynomial $R(r)$ has four roots denoted by $r_1\leq r_2\leq r_3\leq r_4$ and $R$ has the factorization,
\[R(r)=(E^2-1)(r-r_1)(r-r_2)(r-r_3)(r-r_4).\]
There are three parameters $E,L_z,Q$ to vary. When \eqref{EqCrit} holds with $\frac{d^2}{dr^2}R(r_c)>0$, the point $r_c$ a root of multiplicity 2. Since we have $-R\sim (1-E^2)r^4$ as $r\to \infty$.  This implies that $r_c$ cannot be the largest root $r_4$ or the smallest root $r_1$. Then we have $r_c=r_2=r_3$ and (see Figure \ref{FigR})
\[R(r)=(E^2-1)(r-r_1)(r-r_c)^2(r-r_4).\]

First, we know that $R(r)\geq 0$ for $r\in [r_c,r_4]$ hence we get that the graph of the function
\begin{equation}\label{EqHomoK}\frac{dr}{d\lambda}=\pm\sqrt{R(r)}=\pm|r-r_c|\sqrt{(E^2-1)(r-r_4)(r-r_1)},\ r\in [r_c,r_4]\end{equation}
is the phase portrait of a homoclinic orbit to the hyperbolic fixed point $(r,\frac{dr}{d\lambda})=(r_c,0)$.
%The point $r=r_u$ is a hyperbolic fixed point. In the $(p_r,r)$ space, we get a homoclinic orbit
%$$ \gamma_0:=\left\{(r,p_r)\ |\ p_r=\dfrac{1}{\Delta}\sqrt R,\quad r\in [r_u,r_1]\right\}.$$

We next show that we can indeed choose parameters to guarantee $r_1<r_c<r_4$ so that $R''(r_c)>0$. Indeed, from \eqref{EqRFactor} we see that this is equivalent to
\begin{equation}\label{EqAdm}Ar_c^2+2Br_c+C>0\Longleftrightarrow\ 3r_c^2(1-E^{-2})+E^{-2}M-\frac{a^2}{r_c^2}q_c>0.\end{equation}
%for $r_1=\frac{B-\sqrt{B^2-AC}}{-A},\ r_4=\frac{B+\sqrt{B^2-AC}}{-A}$ so the requirement
We examine the inequality using the parameters $M=1,E=1,a=0.8$. Then $A=0,B=1,C=\frac{-a^2}{r_c^2}q_c$ and the inequality reduces to $2r_c-\frac{a^2}{r_c^2}q_c>0$. We  can check from  \eqref{Eqellq} that  this inequality indeed holds for all $r_c$ with $q(r_c)\geq 0$ (see  also Figure \ref{fig:side:aChandra}). Therefore by continuity we have the same conclusion for nearby $E<1$ and $a$. %However, it is also clear that for smaller $E$, the inequality may fail. So we introduce the following definition.
%When $(a,E)$ is admissible, for each $r\in [r_-(E), r_+(E)]$, there are unstable spherical orbits with constant radius $r$ lying on the photon shell. In particular when $r=r_\pm(E)$, the orbit lies on the equator plane since we have $Q=0$.

%There are indeed parameters satisfying this inequality. Indeed, we may choose $E^2$ very close to 1 so $Ar_c^2$ is close to zero and $B$ is close to $M$. When we choose $a=0.8M, r_c=2M$ then $q_c$ is close to zero, so the inequality is satisfied.

Finally, we also have $R(r)\geq 0$ for $r\in [r_1,r_c]$ so there is a corresponding homoclinic orbit.  However, such a homoclinic orbit is not physical since we have that $r_1<r_h$ for the smallest root of $R(r)=0$, so that part of the homoclinic orbit lies inside the event horizon (see Figure \ref{FigR}). This can be seen from the estimate $-R(r_h)<0$ using \eqref{EqRTheta} and the fact that $\Delta(r_h)=0.$ For the same reason, the homoclinic orbit in the $E^2>1$ case intersects the event horizon,
%So we require $r_1>r_h$ in order for the homoclinic orbit to be outside the event horizon, which gives $Ar^2_h+2Br_h+C<0$ using \eqref{EqRFactor}. We examine the parameter $M=1,E=1$, hence $A=0,B=1,C=\frac{a^2}{r_c^2}q_c$, so the required inequality becomes $2r_h-\frac{a^2}{r_c^2}q_c<0$, which is violated for all $a\in [0,M]$ and all $r\in [r_-(E), r_+(E)]$.
so in the following, we do not consider these homoclinic orbit.

%The explicit expression of the homoclinic orbit is very complicated in general. In \cite{LP1}, the authors worked out explicitly the analytical expression of the homoclinic orbit in the case of $Q=0$, in which case the orbit lies entirely on the equatorial plane, in which case we have $r_1=0$. However in this case the graph of the function $$\pm \sqrt{R}=\pm|r-r_c|\sqrt{(E^2-1)(r-r_4)r},\ r\in [r_1,r_c],$$ does not give a homoclinic orbit since such an orbit necessarily  crosses the event horizon hence falls into the blackhole.

%When $E^2>1$, we have that the double root is the largest root,  $r_1<r_2<r_c$ and
%$$R=(E^2-1)(r-r_1)(r-r_2)(r-r_c)^2.$$
%In this case, we have $R(r)\geq 0$ for $r\in [r_2,r_c]$ and the graph of the function
%\[\frac{dr}{d\lambda}=\pm\sqrt{R(r)}=\pm|r-r_c|\sqrt{(E^2-1)(r-r_2)(r-r_1)},\ r\in [r_2,r_c]\]
%is a homoclinic orbit to the hyperbolic fixed point $(r,\frac{dr}{d\lambda})=(r_c,0)$.

%The $r$-equation can be solved explicitly as follows. We have $$\lambda=\lambda_0+\int_{r_0}^r\frac{1}{\sqrt{R(r)}}\,dr.$$
%When we have $r(\lambda)$ we obtain $p_r(\lambda)$ from the equation $\frac{dr}{d\lambda}=\pm\sqrt{R(r)}=p_r\Delta$. Similarly, we solve the $\theta$-equation.

\subsubsection{The homoclinic orbit to the photon shell}
We next solve for the homoclinic orbit \eqref{EqHomoK} explicitly. Similar researches can be found in literature \cite{LP1,LP2}. For this purpose, we use the equation $\lambda=\int \frac{dr}{\sqrt R}$. Using the factorization \eqref{EqRFactor}, we make a change of variable $x=(r-r_c)^{-1}$ and obtain
\begin{equation}\label{EqKerrr}\lambda-\nu=\int_0^r \frac{dr}{\sqrt R}=\frac1E\int_0^x \frac{dx}{\sqrt{F}}=\pm \frac{1}{\gamma}\ln(2\sqrt{\gamma F}+2\gamma x+\beta),\quad \gamma>0,\end{equation}
where $\nu$ is the integral constant, $F=\al+\beta x+\gamma x^2$ and
$\al=1-E^{-2}, \ \beta= 4 r_c(1-E^{-2})+2ME^{-2},\ \gamma=3r_c^2(1-E^{-2})-a^2q_cr_c^{-2}+2Mr_c E^{-2}.$ Along the homoclinic orbit $r(\lambda)\to r_c$ exponentially as $\lambda\to\pm\infty$, which corresponds $x(\lambda)\to +\infty$ exponentially. If we have in addition $D=-4\left(\left(r_c(1-E^{-2})+ME^{-2}\right)^2+\al^2q_cr_c^2(1-E^{-2})\right)>0$, then the above integral can be written more explicitly as $\pm\frac{1}{E\sqrt\gamma}\sinh^{-1}\frac{2\gamma x+\beta}{\sqrt D}$.
Solving the inverse function, we find $r$ as a function of $\lambda$.

%=\pm\frac{1}{E\sqrt\gamma}\sinh^{-1}\frac{2\gamma x+\beta}{\sqrt D},\quad \gamma>0,D>0,$$where $D=-4\left(\left(r_c(1-E^{-2})+ME^{-2}\right)^2+\al^2q_cr_c^2(1-E^{-2})\right)$ and

Note that the $\theta$-equation does not depend on $r$ so the solution to the $(\theta,p_\theta)$ equation in Mino time and with initial condition $(\theta(0), p_\theta(0))$ can be expressed as $(\theta,p_\theta)(\lambda,(\theta(0), p_\theta(0)),L_z).$ However, the $\varphi$-equation depends on both $\theta$ and $r$, so to solve the $\varphi$-equation, we have to substitute a solution $r(\nu+\lambda),\theta(\lambda,(\theta(0),p_\theta(0)),L_z)$. Therefore the solution of the $\varphi$-equation with initial condition $r(\nu),\theta(0),\varphi(0)$ has the form $\varphi(\lambda,\nu,\theta(0),\varphi(0),L_z)$.

\begin{Not}
We denote $\Gamma(\lambda+\nu)=(r(\lambda+\nu), p_r(\lambda+\nu))$ the homoclinic orbit on the $(r,p_r)$-plane with initial condition $\Gamma(\nu)$ solved from \eqref{EqKerrr}, $\Gamma_c=(r_c,0)$ and by $\Xi=(\theta, p_\theta, \varphi,p_\varphi)$. We denote $\Xi=(\theta,p_\theta,\varphi,p_\varphi)$ and write the solution with initial value $r(\nu)$ and $\Xi_0$ as $\Xi(\nu,\Xi_0,\lambda)$. Here the dependence on $\nu$ enters through the $\varphi$-equation and does not enter $\theta,p_\theta,L_z$. Such a dependence can be solved from differentiating the $\varphi$-equation $\frac{d\partial_\nu \varphi(\lambda)}{d\lambda}=\left(\left(\frac{aP}{\Delta}\right)'(r(\nu+\lambda))\right)r'(\nu+\lambda).$

\end{Not}

\subsubsection{The scattering map}
Suppose the Kerr metric $g_{\mu\nu}dx^\mu dx^\nu$ is perturbed to $\tilde g_{\mu\nu}dx^\mu dx^\nu=(g_{\mu\nu}+\eps h_{\mu\nu})dx^\mu dx^\nu$, then the corresponding Hamiltonian is $$H=\frac12\tilde g^{\mu\nu}p_\mu p_\nu=\frac{1}{2}(g^{\mu\nu}p_\mu p_\nu-\eps  g^{\al\mu}h_{\mu\nu}g^{\beta\nu}p_\al p_\beta)+O(\eps^2).$$
We denote by $H_1=-\frac12 g^{\al\mu}h_{\mu\nu}g^{\beta\nu}p_\al p_\beta+O(\eps)$ the perturbation and proceed.

%We further denote by $(R(\alpha+\lambda), \Xi(\Xi_0,\lambda)),\ \lambda\in \R$ an orbit of the unperturbed system with initial condition $(R(\alpha),\Xi_0)$.

The set $N$ in \eqref{EqNHIMK} is a NHIM for the Kerr Hamiltonian. When the perturbation is added, the theorem of NHIM implies that $N$ persists and gets only slightly deformed into a nearby submanifold denoted by $N_\eps$ that is a NHIM for the perturbed system. The Hamiltonian restricted to $N_\eps$ is still Hamiltonian denoted by $H_{N_\eps}$, and when pulled back to a function on $N$, it is a perturbation of $H_N$. The perturbation may create transversal intersections between the stable $W^s(N_\eps)$ and unstable $W^u(N_\eps)$ manifolds of $N_\eps$, which enables us to define the scattering map.
%The perturbed Hamiltonian when restricted to $N_\eps$.

The following theorem is proved completely analogously to Theorem \ref{ThmScatteringS1} (see Section \ref{SSScattering} and Remark \ref{RkK}).
\begin{Thm}\label{ThmScatteringK}
The scattering map $\mathcal S:\ N_\eps\to N_\eps$ is a well-defined symplectic map preserving the level set $\{H_{N_\eps}=-1/2\}$, and $O(\eps)$-close to identity in the $C^1$ norm. Suppose that the function $\nu\mapsto J(\nu,\Xi_0)$ has nondegenerate critical point denoted $\nu^*=\nu^*(\Xi_0)$ for all $\Xi_0\in U\subset N_\eps\cap H_{N_\eps}^{-1}(-1/2)$ where $U$ is an open set and
\begin{equation}
\begin{aligned}
J(\nu,\Xi_0)&=\int_\R  \Sigma(r(\nu+\lambda),\theta(\lambda,\theta(0)))H_1(\Gamma(\nu+\lambda),\Xi(\nu,\Xi_0,\lambda))\\
&-  \Sigma(r_c,\theta(\lambda,\theta(0))) H_1(\Gamma_c,\Xi(\nu,\Xi_0,\lambda))\,d\lambda.
\end{aligned}
\end{equation}
 Then explicitly up to an error $o_{C^1}(1)$ we have for $\star=\varphi$
\begin{equation}\label{EqScatteringK}
\begin{aligned}
&\frac{1}{\eps}(p_\star(z_+)-p_\star(z_-))=\\
&-\int_\R \Sigma(r(\nu^*+\lambda),\theta(\lambda,\theta(0)))\partial_\star H_1(\Gamma(\nu^*+\lambda),\Xi(\nu^*,\Xi_0,\lambda))\\
&-\Sigma(r_c,\theta(\lambda,\theta(0)))\partial_\star H_1(\Gamma_c,\Xi(\nu^*,\Xi_0,\lambda))\,d\lambda.\\
\end{aligned}
\end{equation}
%where $\nu^*=\nu^*(\Xi_0)$ is a nondegenerate zero of the equation $\partial_\nu J(\nu,\Xi_0)=0$.
%Moreover, the map $\mathcal S$ preserves the energy level set of $H_{N_\eps}$.
\end{Thm}
\subsubsection{Arnold diffusion and its physical meaning}\label{SSSADK}
We also have the following theorem on Arnold diffusion in perturbed Kerr spacetime that is proved completely analogous to Theorem \ref{ThmGLS} (by applying Theorem 4.1 of \cite{GLS} to the Hamiltonian \eqref{EqHamKR}).
\begin{Thm}\label{ThmKerrAD1}
Let $\eps h_{\mu\nu}dx^\mu dx^\nu$  be a stationary perturbation to the Kerr metric. Suppose
\begin{enumerate}
\item the assumption of Theorem \ref{ThmScatteringK} holds;
\item the parameters $M,a,E$ are as in Definition \ref{DefAdmE}, and $Q,L_z$ are solved from  \eqref{Eqellq} with $r\in \mathcal R$;
\item the RHS of \eqref{EqScatteringK} with $\star=\varphi$ is nonvanishing at some point $\Xi_0\in U$. %for each set $\{L^2_z=\mathrm{constant}\in (a,b)\subset (0,L^2)\} $ as a subset of $ \{L^2=p_\theta^2+\frac{L_z^2}{\sin^2\theta}\}$, its image under the scattering map intersects itself transversally $($in \eqref{EqScattering} in Theorem \ref{ThmScatteringS1}, choosing $L_z(z_-)=\mathrm{const}$, we require $L_z(z_+)$ is a nondegenerate function of $\varphi$$)$.
\end{enumerate}
Then  there exist $\rho>0$ and $\eps_0>0$, such that for all $|\eps|<\eps_0$, there exists  $T>0$ and an orbit of the perturbed system such that
$$|L_z(T)-L_z(0)|> \rho. $$
\end{Thm}
The dynamics on the NHIM N is studied in \cite{X}, which has the twist property for some parameter range, thus many existing techniques in literature (\cite{CY1,CY2,Tr1,Tr2, DLS06,GLS} etc) apply to verify that a given perturbation can give orbits whose $L_z$-variable varies in the entire interval $\ell(\mathcal R)$ up to a small $\dt$-error near the boundary. We do not pursue that generality for the sake of simplicity. Similarly, if the perturbation is not stationary but 1-periodic in $\tau$, we have the analogue of Theorem \ref{ThmDiffusionNonStatS}. We skip the statement here.

Let us now explain the physical picture of the diffusing orbit given by the last stated theorem. The theorem implies that there exists zoom-whirl orbit in perturbed Kerr spacetime to have large oscillation of $L_z$. A stationary perturbation does not change the value of $E$. From \eqref{EqCrit} and \eqref{Eqellq}, we see that the change of $L_z=\ell E$ will lead to a change of the solution $r_c$, which in turn leads to a change in $Q=qE$. Thus during the zoom-whirling process, the whirling orbit visits photon spheres with different radii $r_c$. Moreover, we can find orbit such that $|L_z(\infty)-L_z(-\infty)|>\rho$ thus when $t\to\pm\infty$ the orbit approaches photon spheres with different radii. This phenomenon does not even occur in Schwarzschild case.

Solving the linearized vacuum Einstein equation at Kerr metric is significantly harder than at Schwarzschild, so we do not provide a concrete perturbation to verify the assumption of
Theorem \ref{ThmKerrAD1} as we did in Section \ref{SReggeWheeler}.

\subsubsection{Arnold diffusion and Penrose process}\label{SSSPenrose}
We remark that Arnold diffusion in this setting is related to the Penrose process. As explained in Section \ref{SSK} (c.f. Section 65 of \cite{C} and Section 12.4 of \cite{W} etc), the Penrose process is a process to extract energy and angular momentum from the blackhole. Consider a particle with energy $E=1$ entering the ergosphere but not yet the event horizon and we can arrange the particle to break into two pieces, one with negative energy $E_1<0$, so the other has energy $E_2>1$. The negative energy fragment falls into the blackhole while the other piece escapes to infinity, thus the mass of the blackhole is reduced to $M+E_1$. Simultaneously, the angular momentum $L_z$ of the negative energy fragment particle should also be negative (equation (12.4.7) of \cite{W}), i.e. opposite to that of the black hole, so that the angular momentum of the blackhole is also reduced.

We substitute $r=r_e=2M$ for the equatorial orbit into the $q$-equation in \eqref{Eqellq}, we find that $q(r_e)\simeq \frac{(28-8\sqrt 2)a^2-16}{a^2}$ for the parameters $M=1,E=1$. This number is positive for $a$ close to $1$, in which case  the inner boundary of the photon shell  lies inside the ergosphere. This means that the photon shell intersects the ergosphere.

We then consider a particle whose energy $E<1$ and angular momentum $L_z$ are such that  the photon shell intersects the ergosphere $2M\in \mathcal R$. We may use the mechanism of Arnold diffusion to arrange the particle to move on zoom-whirl orbit to reduce its angular momentum and simultaneously radius of the photon shell orbit will increase. We then arrange that the particle breaks into two pieces with $E_1+E_2=E,\ L_1+L_2=L_z$ and $E_1<0, L_1<0$ when $r=r_c$. The one with negative energy and angular momentum falls into the event horizon while the other one has both energy $E_2$ and angular momentum $L_2$ increased compared to $E$ and $L_z$ respectively. If the values of $E_2$ and $L_2$ are arranged precisely such that the second particle still shadows a photon shell orbit (we need $E_2<1$, \eqref{Eqellq} holds). We may apply the mechanism of Arnold diffusion again to reduce the value of $L_z$ before the next time when we apply Penrose mechanism to break the particle into two. This process, if arranged precisely, may maintain the zoom-whirl particle's angular momentum in a fixed range but keep reducing the energy and angular momentum of the blackhole, until the zoom-whirl particle's energy exceeds 1 then it escapes.

This process always increases the energy of the zoom-whirl particle, since we consider stationary perturbations which does not change the energy. In the following, we expect that when a $\tau$-periodic perturbation is considered, we can reduce both the values of energy and the angular momentum of the zoom-whirl particle so that we can keep reducing the energy and angular momentum of the blackhole (c.f. Conjecture \ref{ConK1}(2)). Moreover, the above diffusion mechanism using zoom-whirl orbit requires a delicate control of $E_2,L_2$ such that the particle after the Penrose process also moves on zoom-whirl orbit. With the following following Conjecture \ref{ConK2}(2), we expect that the above process of combining Arnold diffusion and Penrose process can work without delicate control of the values of $E_2,L_2$.

%Where $r_c>r_1>0$ is a double root corresponding to a hyperbolic fixed point of the radial equation.

%Moreover, for the range of parameters, we have that $r_1\in (r_+, 2M)$, i.e. the root $r_1$ lies in the ergosphere. Thus the homoclinic orbit enters the ergosphere and connects to the photon sphere.
\subsubsection{Conjectures}

%We introduce $u=\cos\theta$ then we obtain the equation for $u$
%$$\Sigma^2 \dot u^2=(\mu^2-E^2)a^2u^4-(Q+L_z^2+(\mu^2-E^2)a^2) u^2+Q.$$
%We solve the RHS=0 for $u^2$ to get
%$$u^2=\frac{(Q+L_z^2+(\mu^2-E^2)a^2)\pm\sqrt{(Q+L_z^2+(\mu^2-E^2)a^2)^2-4Q(\mu^2-E^2)a^2}}{2(\mu^2-E^2)a^2}$$

As in the Schwarzschild case, we formulate the following conjecture in the Kerr setting.
 \begin{Con}\label{ConK1}
 \begin{enumerate}
\item Let $\eps h_{\mu\nu}dx^\mu dx^\nu$  be a generic stationary perturbation to the Kerr metric, suppose the tuple $(a,M,E,Q,L_z)$ of parameters are as in assumption (2) of Theorem \ref{ThmKerrAD1}, then for any $\dt>0$,
 there exists $\eps_0$ such that for all $|\eps|<\eps_0$,
 there exists an orbit on $N_\eps$ and times $T,T'$ such that
$$|L_z(T)-E\ell_-(E)|<\dt,\quad |L_z(T')-E\ell_+(E)|<\dt,$$
where $\ell_-(E)<\ell_+(E)$ are two endpoints of the interval $\ell(\mathcal R)$.
\item Let $\eps h_{\mu\nu}dx^\mu dx^\nu$  be a generic perturbation $($1-periodic in $\tau)$ to the Kerr metric, then for any $\dt>0$, there is an orbit that visits any $\dt$-ball centered at  $\{H_{\hat N_\eps}=-\frac12\}$ on $\hat N_\eps$, provided $\eps$ is sufficiently small.
 \end{enumerate}
 \end{Con}
 We remark that the dynamics on $N$ in general has the twist property (see Theorem 1.2 of \cite{X}) so that the first item should be easier than its Schwarzschild counterpart.
 %We remark that the system $H_{\hat N_\eps}$ has three degrees of freedom, so we expect that the diffusion orbit in Conjecture \ref{ConK1}(2) can be constructed either using the normal hyperbolicity, or not, and the latter case is similar to Conjecture \ref{ConK2} below.

To state the next conjecture, we restrict our attention to the part of the phase space of bounded motions. Since $\varphi$ is defined on $\T$ and the domain for $\theta$ is also bounded for bounded $Q$ (see Section \ref{SSSVertical}), we consider only the radial component. We are interested in the right potential well of $-R(r)$. Recall that the function $R$ depends on the constants of motion $E,L_z,Q$. For $E^2<1$ we denote by $r^*$ the local max of $-R$ i.e.  $-R'(r^*)=0$ and $-R''(r^*)<0$ and by $r_*(>r^*)$ the local min of $-R$ to the right of $r^*$, i.e.  $R'(r_*)=0$ and $-R''(r_*)>0$. For a particle moving in the right potential well of $-R$ satisfying $(\frac{dr}{d\lambda})^2-R(r)=0$, we necessarily have $-R(r^*)>0$ and $-R(r_*)<0$. The figure of $-R$ can be obtain by lifting the graph of Figure \ref{FigR} upward slightly so that $-R$ has four zeros and we are interested in the right lobe below the $r$-axis.

Let $C$ be a large number. We define the following part of the phase space with bounded motions hence Liouville-Arnold theorem applies
\begin{equation}\label{EqBK}\mathcal B(C):=\left\{0< E^2<1,\  |L_z|<C,\ 0<Q<C\ |\ -R(r^*)>0,\ -R(r_*)<0\right\}\cap H^{-1}(-\frac12). \end{equation}
Similar to \eqref{EqBS}, in the next Conjecture \ref{ConK2}(1) we treat $\mathcal B(C)$ as a 5 dimensional set and in Conjecture \ref{ConK2}(2) we treat $\mathcal B(C)$ as a 7 dimensional set (see the paragraph before Conjecture \ref{ConS2})
 \begin{Con}\label{ConK2}
 \begin{enumerate}
\item Let $\eps h_{\mu\nu}dx^\mu dx^\nu$  be a generic stationary perturbation to the Kerr metric, then for any $\dt>0$ and $C>0$,
there exists an orbit that is $\dt$-dense on $\mathcal B(C)$ provided $\eps$ is small.
\item Let $\eps h_{\mu\nu}dx^\mu dx^\nu$  be a generic perturbation $($1-periodic in $\tau)$ to the Kerr metric, then for any $\dt>0$ and $C>0$,
 there exists an orbit that is $\dt$-dense on $\mathcal B(C)$, provided $\eps$ is small.
 \end{enumerate}
 \end{Con}
 We have similar remarks as Remark \ref{RkConS} in the Schwarzschild case.
\subsection{Oscillatory orbits in perturbed Kerr spacetime}
In this section, we study the oscillatory orbits in perturbed Kerr spacetime generalizing the analysis in Section \ref{SSOsc}.

 For simplicity, we consider only the case of $Q=0$, which implies $\theta=\pi/2$ and that the orbits all lie on the equatorial plane. Then $\Sigma=r^2$ and the radial equation \eqref{EqofMotion} $$\left(\frac{ dr}{dt}\right)^2=\frac{R(r)}{\Sigma(r)^2}=\frac{R(r)}{r^4}.$$
We next choose $E^2=1$ then from \eqref{EqRTheta} we get that the degree of $-R$ is less than $4$ and it has the form $-R(r)=-2Mr^3+(L_z^2-2a^2)r^2+O(r)$, so that as $r\to\infty$ we have $$\left(\frac{ dr}{dt}\right)^2-\frac{2M}{r}=o(r^{-1}),$$
and moreover, the largest root of $R$ is estimated as $\frac{L_z^2}{2M}+O(1)$ for $L_z^2$ large.

Let $r_*$ be the largest root of $R(r)=0$ so $R(r)>0$ for $r>r_*$. Then
$$\mathrm{Graph}\left\{\frac{dr}{dt}=\pm\frac{1}{r^2}\sqrt{R(r)} \ \Big|\ r\geq r_*\right\}$$
in the $(r,\frac{dr}{dt})$-plane gives rise to the phase portrait of the homoclinic orbit to infinity. The homoclinic orbit can be solved explicit using the method of Section \ref{SSSOscHomo}.
In Mino time, the $\varphi$-equation is estimated as $\frac{d\varphi}{d\lambda}=L_z+O(r^{-1})=L_z+O(\frac{1}{L_z^2})$ in the large $L_z$ limit.

The situation is now similar to the Schwarzschild case in Section \ref{SSOsc}, we can perform a similar coordinate change to reveal the degenerate hyperbolic fixed point at infinity. We then define the Melnikov function $J(\nu,\varphi_0)$ by adapting that in Section \ref{SSOsc}. Then we have the following result by adapting Theorem \ref{ThmOsc}.

\begin{Thm}\label{ThmOscK}
Let $\eps h_{\mu\nu}dx^\mu dx^\nu$ be a stationary perturbation of the Kerr metric satisfying
\begin{enumerate}
\item The perturbation $H_1=- \frac12  g^{\al\mu}h_{\mu\nu}g^{\beta\nu}p_\al p_\beta+O(\eps)$  converges to 0 asymptotically like $O(\frac{1}{r^2})$ as $r\to\infty$ and $\lim_{r\to\infty}r^2H_1$ is a constant;
\item The subspace $\{\theta=\pi/2,\ p_\theta=0\}$ is invariant under the perturbed Hamiltonian system;
\item $|L_z|$ is sufficiently large and $\partial_\nu J(\nu,\varphi_0)$, has a nondegenerate zero for $\varphi_0=0$,
\end{enumerate}
 then there exists an orbit of oscillatory type for $\eps>0$ sufficiently small, i.e. there exists a constant $C$ such that along the orbit we have
$$\liminf_{t\to\infty} r(t)<C,\quad \limsup_{t\to\infty} r(t)=\infty.$$
\end{Thm}

%When $\max_{r>r_h} -\frac{R}{\Sigma^2}=0$ and is attained at a local maximal point, then the orbit is actually a heteroclinic orbit between infinity and the photon sphere.

%We can denote $x=1/r$ to rewrite the equation of motion as $$\dfrac{dx}{d\tau}=-\sqrt{2M} x^{5/2}.$$ Therefore $x=0$, which corresponds to the infinity in the $r$ coordinates, is a degenerate saddle.

%\section{Problems}

%In this section, we provide a list of problems.

%Problem 1: Figure  out the exact domain of KAM nondegeneracy and convexity in Section???

%Our argument provided the the convexity in a neighborhood of the elliptic fixed point. For more global convexity, one need to resort to the computer. We have provided the formulas.

%Problem 2: In the section ???, we see that for both cases $E^2>1$ and $E^2<1$, there exists homoclinic orbits entering the event horizon. Try to figure out the physical meaning of these orbits.
\section{Chaotic carrousel motions around the event horizon}\label{SCarrousel}
In this section, we show that there exists a remarkable chaotic behavior near the event horizon of the perturbed Kerr spacetime.  The analysis in this section applies also to the Reissner-Nordstr$\phi$m metric.

\subsection{The dynamics crossing the event horizon}
We consider only the special case of $Q=0$ (so the orbit lies on the equatorial plane). For simplicity, we start by considering $L_z=aE$. This case is a  representative of the way how the orbit crosses the event horizon.

From \eqref{EqofMotion}, we have
$$\label{Eqdrdphi}(\frac{dr}{d\lambda})^2-R(r)=0,
$$
where $-R(r)=-r^2((E^2-1) r^2+2Mr-a^2)$ and $r^2dt=d\lambda$ (see Figure \ref{FigCarrouselR}). We view the radial dynamics as a particle moving in the potential well $-R$ on the zero energy level. Note that $-R(r_\pm)<0$ for $r_\pm$ solving $\Delta(r)=0$.  Let $0<r_I<r_O$ be the two positive roots of $R$, then we have $r_I<r_-<r_+<r_O$. Then we see that the radial motion is first to decrease from $r_O$, crossing the outer event horizon $r_+$, then crossing the inner event horizon $r_-$ then reach $r_I$ without hitting the singularity $r=0$. Then the radial motion will increase and cross the inner then the outer event horizons until reaching $r_O$. If we consider only the dynamics of the radial component, this motion is periodic. However, in general people treat the returning piece of orbit as entering a different domain in the Penrose diagram hence unfolds the periodic orbit. Such a behavior also exists in the Reissner-Nordstr$\phi$m metric, differring drastically from the radial dynamics of Schwarzschild metric.

Note that in the above discussion, we did not include the $\varphi$ and $\tau$-components. We have $\frac{d\varphi}{dt}=aE/\Delta$ and $\frac{d\tau}{dt}=E(r^2+a^2)/\Delta$ both become singular at the event horizons. In particular, we see that the angular variable $\varphi$ is fast rotating when approaching the event horizon. But these singular behaviors are just coordinate singularities, and can be resolved by a coordinate change. The way we resolve the singularity is to perform a time change, that is, to treat $\varphi$ as the new time in place of the proper time $t$. %Again, from equation \eqref{EqofMotion}, we obtain
%\begin{equation}\label{Eqdrdphi}(\frac{dr}{d\varphi})^2=\frac{\Delta^2 R}{(aP(r)-(aE-L_z)\Delta)^2}. \end{equation}%For simplicity, we focus on the case $aE=L_z$. It would be clear that the assumption can be relaxed.
\begin{figure}
\centering
\includegraphics[width=0.3\textwidth]{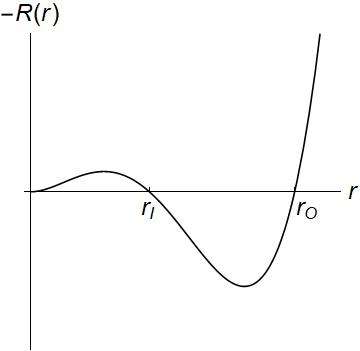}
\caption{Graph of $-R$}
\label{FigCarrouselR}
\end{figure}

From equation \eqref{EqofMotion}, we obtain
\begin{equation}\label{Eqdrdphi}(\frac{dr}{d\varphi})^2=\frac{\Delta^2 ((E^2-1) r^2+2Mr-a^2)}{a^2E^2r^2}:=-V(r).\end{equation}

%When approaching the event horizon, the $r$-equation is nonsingular so we have $r$ is nearly linear in $t$ when approaching the event horizon. However, the $\varphi$-equation in \eqref{} becomes singular since $\Delta\to0$ nearly linearly. So we get $\dot \varphi\sim \frac{1}{t-t_*}$.
%This means that $\varphi$ completes infinitely many periods within finite proper time as $t\to t_*$. However, the event horizon is simply a coordinate singularity. After performing a coordinate change resolving the singularity, we treat the outer event horizon as a hyperbolic fixed point.
%The main purpose of this section is to show that the event horizon under perturbation may lead to some remarkable new orbit that are absent when there is no perturbation.
%Denote by $r^-_h<r^+_h$ the two roots of $\Delta$ and by $u^+=1/r^+_h,\ u^-=1/r^-_h$.
%Set $u=\frac{1}{r}$, then we get
%$$\frac{du}{d\varphi}=\pm \frac{a}{E}(u-u^+)(u-u^-)((E^2-1)+2M u-a^2 u^2)^{1/2}. $$
%Denoting $V(u)= \frac{a^2}{E^2}(u-u^+)^2(u-u^-)^2((E^2-1)+2M u-a^2 u^2)$, then the equation of motion becomes
%$$(\frac{du}{d\varphi})^2-V(u)=0. $$
%Noting that $V(0)=\frac{a^2}{E^2}(u^+)^2(u^-)^2(E^2-1)<0$ if $E^2<1$. This means that there is root $0<u_*<u^+$ of $V$.
 Again we visualize this system as a particle moving in a potential well of the potential $V$ on the zero energy level (see Figure \ref{FigCarrouselV}). Now the event horizon $r=r_\pm$ becomes a local max of $V$, thus the point $(r=r_\pm,\frac{dr}{d\varphi}=0)$ is a hyperbolic fixed point of the ODE in the phase plane $(r,\frac{dr}{d\varphi})$.
 \begin{figure}
\centering
\includegraphics[width=0.3\textwidth]{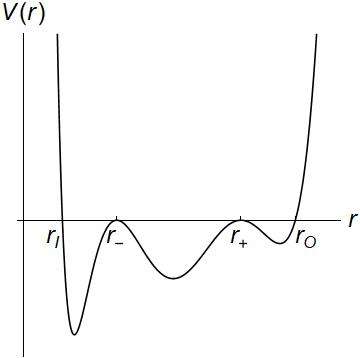}
\caption{Graph of $V$}
\label{FigCarrouselV}
\end{figure}

 This gives two homoclinic orbits approaching $r_\pm$ respectively
 \begin{equation}
 \begin{aligned}
\Gamma_1:=&\mathrm{Graph}\left\{\frac{dr}{d\varphi}=\pm\sqrt{V(r)},\ r\in(r_+,r_O]\right\},\\
\Gamma_2:=&\mathrm{Graph}\left\{\frac{dr}{d\varphi}=\pm\sqrt{V(r)},\ r\in[r_I,r_-)\right\},\\
 \end{aligned}
 \end{equation}
 and two heteroclinic orbits connecting $r_+$ and $r_-$
 \begin{equation}
 \begin{aligned}
\Gamma_3:=&\mathrm{Graph}\left\{\frac{dr}{d\varphi}=+\sqrt{V(r)},\ r\in(r_-,r_+)\right\},\\
\Gamma_4:=&\mathrm{Graph}\left\{\frac{dr}{d\varphi}=-\sqrt{V(r)},\ r\in(r_-,r_+)\right\}.\\
 \end{aligned}
 \end{equation}
These homo- or hetero- clinic orbits can be solved as function of $\varphi$ (c.f. Section 61 of \cite{C}).   We denote by $\Gamma_k(\cdot+\psi)=(r(\cdot+\psi),\frac{dr}{d\varphi}(\cdot+\psi))$ an orbit $(r,\frac{dr}{d\varphi})$ with initial condition $(r(\psi),\frac{dr}{d\varphi}(\psi))$.
%We can thus treat the event horizon as a hyperbolic fixed point, which admits a homoclinic orbit.

\subsection{Chaotic dynamics around the event horizon under perturbation}
So the general idea is that a generic perturbation may create separatrix splitting as in the Theorem of Poincar\'e hence lead to chaotic dynamics. We next outline the procedure of resolving the singularity and generating the separatrix splitting. After that we formulate precise statement and detailed proofs.

%\subsection{Perturbation of Kerr}
%\begin{equation}\label{Teukolsky}
%\begin{aligned}
%\end{aligned}
%\end{equation}

Consider next a perturbation $\eps H_1=\eps \hat h^{\mu\nu}p_{x^\mu} p_{x^\nu}$ to the Kerr Hamiltonian that is a stationary and nonaxisymmetric perturbation. We fix a value $0<E<1$ and consider equatorial motion with $\theta=\pi/2,p_\theta=0$, then the system $H_0+\eps H_1$, where $H_0$ is the Kerr Hamiltonian, is a system of two degrees of freedom with coordinates $(r,p_r,\varphi, L_z).$ We next perform an energetic reduction, that is to solve $L_z$ from the equation $H_0+\eps H_1=-\frac12$ as a function of $r,p_r, \varphi$ and treat $\Lambda:=aE-L_z$ as the new Hamiltonian and $\varphi$ as the new time. Then the Hamiltonian system $\Lambda$ is a system of one and half degrees of freedom. However, the set of coordinates $(r,p_r)$ becomes singular near the event horizon, so we use the nonsingular coordinate
$\rho:=\frac{\Delta^2}{aE r^2} p_r=\frac{dr}{d\varphi}$ which vanishes at the horizon in place of $p_r$. Then the symplectic form becomes
$$\omega=dr\wedge dp_r=dr\wedge d(\frac{aEr^2}{\Delta^2}\rho)=\frac{aEr^2}{\Delta^2(r)}dr\wedge d\rho. $$
The equations of motion in the new coordinate system $(r,\rho)$ is nonsingular then, which agrees with \eqref{Eqdrdphi} in the case of $\eps=0$ (we denote by $\Lambda_0=\Lambda|_{\eps=0}$). The perturbation $H_1=H_1(r,p_r,L_z,\varphi)$ is a function of $r,p_r, L_z,\varphi$. In the new coordinates $r,\rho, L_z,\varphi$, we denote $\tilde H_1(r,\rho,L_z,\varphi)=H_1(r,\frac{aEr^2}{\Delta^2}\rho,L_z,\varphi)$.  We shall require that $\tilde H_1$ satisfy certain boundary conditions at $r=r_+$.  We next turn on the perturbation by letting $\eps>0$ but sufficiently small. Then the system \eqref{Eqdrdphi} undergoes a time($\varphi$)-periodic perturbation, so we are in the setting of Appendix \ref{SSPoincare}, so separatrix splitting will be created if the Melnikov function $$\eps\mathcal M(\psi):=\int_\R \omega(X_\Lambda,X_{\Lambda_0})(\Gamma_1(\varphi+\psi), \varphi)\, d\varphi$$
has a nondegenerate zero by Theorem \ref{ThmPoincare}, where we denote $X_\Lambda$ the Hamiltonian vector field associated to the Hamiltonian $\Lambda$ via $\omega( X_\Lambda,\cdot)=-d\Lambda$, explicitly, $X_\Lambda=\frac{\Delta^2}{aEr^2}(\partial_\rho\Lambda, -\partial_
r\Lambda)$, similarly for $X_{\Lambda_0}$.

%We write $H_1$ in terms of the variables  $(r,\frac{dr}{d\varphi},\varphi,L_z)$.

%Then we form the Melnikov function
%$$\mathcal M(\psi):=\frac{1}{aE}\int_\R \Delta(r(\varphi+\psi))H_1(\Gamma_1(\varphi+\psi),L_z,\varphi)\,d\varphi$$

We formulate the following theorem summarizing the outcome of the above procedure.

\begin{Thm} \label{ThmHorizon}Let $\eps H_1=\eps \hat h^{\mu\nu}p_{x^\mu} p_{x^\nu}$ be a stationary and nonaxisymmetric perturbation of the Kerr Hamiltonian with $L_z=aE$ and satisfies the following
\begin{enumerate}
\item The perturbed system preserves the equator plane $\{\theta=\pi/2,p_\theta=0\}$;
\item As $r\to r_+$ and $\rho\to 0$, we have that $\tilde H_1 \Delta^3 , D_{r,\rho}(\tilde H_1 \Delta^3)$, $(D_{r,\rho})^2(\tilde H_1 \Delta^3)$ and $\partial_{L_z}\tilde H_1 \Delta,(D_{r,\rho})(\partial_{L_z}\tilde H_1 \Delta)$ are all bounded; %both %$\Delta^{-2}\left((\partial_r,\Delta^{-2}\partial_{p_r})(H_1\Delta^3)\right)$  and $(\partial_r,\Delta^{-2}\partial_{p_r})^2(H_1\Delta^3)$, $\Delta \partial_{L_z}H_1$.
%are bounded;
\item The Melnikov function $\mathcal M(\psi)$ has a nondegenerate zero.
\end{enumerate}
Then for $\eps$ sufficiently small, there exists a Smale horseshoe in neighborhoods of $r_+$ and $r_O$. In particular, there exists an orbit which visits the two neighborhoods  repeatedly as $t\to\pm\infty$.
\end{Thm}
\begin{proof}
%In the Hamiltonian, we set $\theta=\pi/2,\ p_\theta=0$.

%We write the Hamiltonian as  where the unperturbed part follows from  \eqref{EqHam}.

 We write the Hamiltonian $H$  as follows  (setting $aE-L_z=\Lambda$)
 $$ H=\frac12 r^{-2}(\Delta p_r^2-\dfrac{1}{\Delta}(r^2E+a\Lambda)^2+\Lambda^2)+\eps H_1$$
 where the unperturbed part follows from  \eqref{EqHamK}.
 %The variable $p_r=\frac{\Sigma \dot r}{\Delta}=\frac{r^2 \dot r}{\Delta}$ is singular near the event horizon, we thus introduce the new variable $\rho:=\frac{\Delta^2}{aE r^2} p_r=\frac{dr}{d\varphi}$. Then the symplectic form becomes
%$$dr\wedge dp_r=dr\wedge d(\frac{aEr^2}{\Delta^2}\rho)=\frac{aEr^2}{\Delta^2(r)}dr\wedge d\rho. $$
We also introduce an auxiliary function
$$\mathcal H=\mathcal H_0+\eps \mathcal H_1:=(\rho^2+\tilde V(\Lambda, r))+ \frac{2\eps}{a^2E^2}\tilde H_1 r^{-2}\Delta^3,$$
where $\tilde V(\Lambda, r)=-\frac{\Delta^2}{a^2E^2r^4}((r^2E+a\Lambda)^2-\Delta (\Lambda^2+r^2))$, which reduces to $V(r)$ if we set $\Lambda=0$.
We thus have $H=-\frac12$ if and only if $\mathcal H=0.$ We also have $$\frac{\partial \mathcal H_0}{\partial \Lambda}= \frac{\Delta^2}{a^2E^2r^4}2(aEr^2+a^2\Lambda-\Delta \Lambda)$$
which is reduced to $\frac{2\Delta^2}{aEr^2}$ when $\Lambda=0$.
In terms of the variables $(r,\rho)$, we get the Hamiltonian equations
$\begin{cases}\frac{dr}{dt}&= \frac{\Delta^2}{aEr^2}\partial_\rho H\\
\frac{d\rho}{dt}&= -\frac{\Delta^2}{aEr^2}\partial_r H
\end{cases}$.

Dividing the equations by $\frac{d\varphi}{dt}=\frac{\partial H}{\partial L_z}=-\frac{\partial H}{\partial \Lambda}$, we get
%We next perform the energetic reductiono by treating $\Lambda$ as the Hamiltonian and $\varphi$ as the new time, then we obtain the equations of motion as follows (apply implicit function theorem to solve for $\Lambda$ in both the $H=-1/2$ equation and the $\mathcal H=0$ equation)
\begin{equation}\label{EqHamHorizon}
\begin{aligned}
\frac{d r}{d\varphi}&= -\frac{\Delta^2}{aEr^2}\frac{\partial H}{\partial \rho}\big/ \frac{\partial H}{\partial \Lambda}=\frac{\Delta^2}{aEr^2} \frac{\partial \Lambda}{\partial \rho}=-\frac{\Delta^2}{aEr^2} \frac{\partial \mathcal H}{\partial \rho}\big/ \frac{\partial \mathcal H}{\partial \Lambda}\\
&=-aEr^2\frac{ (\rho+\frac{\eps}{a^2E^2}\partial_\rho \tilde H_1 r^{-2}\Delta^3)}{-aEr^2-a^2\Lambda+\Delta \Lambda-\eps r^{2}\Delta\frac{\partial H_1}{\partial L_z}},\\
\frac{d\rho}{d\varphi}&=\frac{\Delta^2}{aEr^2} \frac{\partial H}{\partial r}\big/ \frac{\partial H}{\partial \Lambda}=-\frac{\Delta^2}{aEr^2} \frac{\partial \Lambda}{\partial r}=\frac{\Delta^2}{aEr^2} \frac{\partial \mathcal H}{\partial r}\big/ \frac{\partial \mathcal H}{\partial \Lambda}\\
&=\frac{aEr^2}{2}\frac{ \partial_r\tilde V(\Lambda,r)-\frac{2\eps}{a^2E^2}\partial_r(\tilde H_1 r^{-2}\Delta^3)}{-aEr^2-a^2\Lambda+\Delta \Lambda-\eps r^{2}\Delta\frac{\partial  H_1}{\partial L_z}}.
\end{aligned}
\end{equation}
When $\eps=0$ and $\Lambda=0$, we see that $(r=r_+,\rho=0)$ is a hyperbolic fixed point. Note also that when $\eps=0$, the equations of motion is reduced to the Hamiltonian equations with Hamiltonian $\mathcal H$ endowed with the symplectic form $dr\wedge d\rho$. The assumption (2) in the statement guarantees that the ODE is perturbed by an $O_{C^1}(\eps)$ perturbation.% so that the hyperbolic fixed point is slightly perturbed and invariant manifolds exist in the perturbed system.

The perturbed system has one and a half degrees of freedom. We think it as a $\varphi$-periodic perturbation of the homoclinic orbit. Then the stable and unstable manifolds of the hyperbolic fixed point are still defined for the time-1 map (here the time is $\varphi$). The stable and unstable manifolds generically splits under the perturbation, which is measured by the Melnikov function by repeating the argument in Appendix \ref{AppPoincare}. %(c.f. in particular, the beginning of the proof in Appendix \ref{SSScattering}).
%$$\mathcal L_0=\frac{\Delta p_r}{a\sqrt{(\frac{r^2}{a^2}+2a)^2+2}}.$$

% we obtain that the Melnikov has a nondegenerate critical point implies that the stable and unstable manifolds of the hyperbolic fixed point intersect transversally.
We thus evaluate the Melnikov function (setting $\Lambda=0$ in the integrand)%We solve $\Lambda$ from the equation $\mathcal H=0$ and denote by $\Lambda_0$ the function solving $\mathcal H_0=0$. Then by implicit function theorem we obtain
\begin{equation}
\begin{aligned}
\eps \mathcal M(\psi)&=\int_\R \omega(X_\Lambda,X_{\Lambda_0})(\Gamma_1(\varphi+\psi),\varphi)\,d\varphi\\
&=\int_\R \frac{\Delta^2}{aEr^2}\left(\frac{\partial \Lambda}{\partial r}\frac{\partial \Lambda_0}{\partial \rho}-\frac{\partial \Lambda_0}{\partial r}\frac{\partial \Lambda}{\partial \rho}\right)\Big|_{(\Gamma_1(\varphi+\psi),\varphi)}\,d\varphi\\
&=-\int_\R \frac{\Delta^2/(aEr^2)}{\frac{\partial \mathcal H}{\partial \Lambda}\frac{\partial \mathcal H_0}{\partial \Lambda}}\left(\frac{\partial\mathcal  H}{\partial r}\frac{\partial  \mathcal H_0}{\partial \rho}-\frac{\partial  \mathcal H_0}{\partial r}\frac{\partial  \mathcal  H}{\partial \rho}\right)\Big|_{(\Gamma_1(\varphi+\psi),\varphi)}\,d\varphi\\
&=-\eps\int_\R \frac{\Delta^2/(aEr^2)}{\left(\frac{\partial \mathcal H_0}{\partial \Lambda}\right)^2}\left(\frac{\partial \mathcal H_1}{\partial r}\frac{\partial \mathcal  H_0}{\partial \rho}-\frac{\partial  \mathcal  H_0}{\partial r}\frac{\partial  \mathcal  H_1}{\partial \rho}\right)\Big|_{(\Gamma_1(\varphi+\psi),\varphi)}\,d\varphi+O(\eps^2)\\
&=-\eps\int_\R \frac{4aEr^2}{\Delta^2}\left(\frac{\partial \mathcal H_1}{\partial r}\frac{\partial \mathcal  H_0}{\partial \rho}-\frac{\partial  \mathcal  H_0}{\partial r}\frac{\partial  \mathcal  H_1}{\partial \rho}\right)\Big|_{(\Gamma_1(\varphi+\psi),\varphi)}\,d\varphi+O(\eps^2)\\
&=4\eps\int_\R \omega(D\mathcal H_0, D\mathcal H_1))\Big|_{(\Gamma_1(\varphi+\psi),\varphi)}\,d\varphi+O(\eps^2).
\end{aligned}
\end{equation}
To make sure the convergence of the last integral, we expand
 \begin{equation*}
 \begin{aligned}
 \omega(D\mathcal H_0, D\mathcal H_1)&
 =\frac{\partial\mathcal H_0}{\partial r}\frac{\partial\mathcal H_1}{\partial p_r}-\frac{\partial\mathcal H_0}{\partial p_r}\frac{\partial\mathcal H_1}{\partial r}\\
 &=\frac{2}{a^2E^2}\left(\left(2\rho\frac{\partial \rho}{\partial r}+ V'(r)\right) \frac{\partial H_1}{\partial p_r} r^{-2}\Delta^3-(2\rho \frac{\partial \rho}{\partial p_r})\partial_r(H_1 r^{-2}\Delta^3)\right)\\
 &=\frac{2}{a^2E^2}\left(\left(2\rho(\frac{\Delta^2}{aE r^2})'p_r+ V'(r)\right) \frac{\partial H_1}{\partial p_r} r^{-2}\Delta^3-(2\rho \frac{\Delta^2}{aE r^2})\partial_r(H_1 r^{-2}\Delta^3)\right)
%& =\frac{8\eps r^2}{aE\Delta^2} \left(\partial_r(H_1 r^{-2}\Delta^3)2\rho-V'(r)\partial_\rho H_1 r^{-2}\Delta^3\right)\\
 %&= \frac{8\eps r^2}{aE} \left(2\rho\left(\frac{\partial(H_1 r^{-2})}{\partial r}\Delta+3H_1 r^{-2} \Delta'(r)\right)- V'(r)\partial_\rho H_1 r^{-2}\Delta\right)\\
 %&= \frac{8\eps r^2}{aE} \left(2\rho\left(\frac{\partial(H_1 r^{-2})}{\partial r}\Delta+3H_1 r^{-2} \Delta'(r)\right)- aEV'(r)\frac{\partial H_1}{\partial p_r} \Delta^{-1}\right).
 \end{aligned}
 \end{equation*}
As $\varphi\to\pm\infty$, we have $\rho\to 0$ and $r\to r_+$ both exponentially, so do $\Delta(r)\to 0$ and $V'(r)\to 0$. So we have $(\frac{\Delta^2}{aE r^2})'p_r$ bounded and $\left(2\rho(\frac{\Delta^2}{aE r^2})'p_r+ V'(r)\right)\to 0$ exponentially.  To make sure the above integral converges, it suffices to require that $\frac{\partial H_1}{\partial p_r} \Delta^3$, $H_1\Delta^4$ and $(\partial_rH_1)\Delta^5$ are all bounded as $r\to r_+$ and $p_r=\Delta^{-1}\frac{dr}{d\lambda}\sim \Delta^{-1} $. These bounds are implied by assumption (2).
%we shall require that $\omega(D\mathcal H_0, D\mathcal H_1))$ exponentially. It is sufficiently to require that $\frac{\partial H_1}{\partial p_r} \Delta^{-1}$, $H_1 $ and $\Delta\partial_r H_1$ are all bounded as $r\to r_+$, which is implied by assumption (2).

As in the proof of Poincar\'e's theorem, that the Melnikov function having nondegenerate zero implies that the transversal intersection of the stable and unstable manifold, which in turn implies the existence of a horseshoe (c.f. Appendix \ref{SSPoincare} and \cite{KH} Chapter 6.5).
The statement then follows.
\end{proof}

%We finally remark on the assumptions.

\begin{Rk}
%\begin{enumerate}
%\item The assumption $L_z=aE$ is clearly removable. For $L_z\neq aE$, extra cares need to be taken to handle the issue that the denominator of $V$ has zero that is not 0.
%\item
The theorem can also be proved in the massless case. Indeed, we consider $-R$ with parameters $Q=0$ and $L_z/E$ such that $-R$ has a nondegenerate zero outside the event horizon (perturb Figure 6A of \cite{X} slightly). Then  we get similar phase portrait for the equation $\frac{dr}{d\varphi}$. Then we have a statement analogous to Theorem \ref{ThmHorizon} in this setting.
%\item The assumption (2) can be relaxed slightly to allow the hyperbolic fixed point to be perturbed and the integrand of the Melnikov function to converge to a bounded function as $\varphi\to\pm\infty$, in which case we shall modify the integrand of the Melnikov function to $\omega(X_\Lambda,X_{\Lambda_0})|_{(\Gamma_1)(\varphi+\psi),\varphi}-\omega(X_\Lambda,X_{\Lambda_0})|_{(r_+,0,\varphi)}$, similar to \eqref{EqMelnikov} that we have done before. We do not pursue this generality here.
%\end{enumerate}
\end{Rk}

\subsection{Outlooks}
\subsubsection{Perturbation of Kerr and assumption (2)}
We next remark  on the technical assumption (2). These are supposed to be quite easy to satisfy. For instance, for the the first three items on the bounds of the $C^2$ norm $\tilde H_1\Delta^3$, a necessary condition is that the $g^{rr}p_rp_r$ term in the Hamiltonian satisfies $g^{rr}/\Delta$ is $C^2$ bounded, which is indeed the case for the unperturbed Kerr Hamiltonian. Similarly, we get bounds on other metric coefficients. The assumptions can be checked directly if a perturbation of Kerr is given. Due to the lack of a natural perturbation and the complexity of the problem, here we only outline how to verify it in general and show its plausibility.

Let us briefly recall the perturbation theory for Kerr. To look for a solution solving the linearized Einstein equation at Kerr, the method of \cite{Te} is to solve first two linear equations called Teukolsky equations for a radial function $\hat R(r)$ and an angular equation $\hat S(\theta)$ (c.f. equation (4.9) and  (4.10) of \cite{Te}). Then form the function $\psi=e^{-\sqrt{-1}\omega \tau}e^{\sqrt{-1} m\varphi}\hat S(\theta)\hat R(r)$, from which, one can recover the metric coefficients(c.f. \cite{C} Section 82).  There are two types of boundary conditions: $\hat R\sim \Delta^{2}e^{-\sqrt{-1} k r_*} $ or $\hat R\sim \Delta^{-2}e^{-\sqrt{-1} k r_*} $ as $r\to r_+$, where $k=\omega-\frac{ma}{2Mr_+}$ and $\frac{dr_*}{dr}=\frac{r^2+a^2}{\Delta}$ (c.f. equation (5.6) of \cite{Te}). To find perturbations satisfying assumption (2), we may look for it from solutions to the Teukolsky equations satisfying  the first type of boundary condition for $R$. %, then we expect that the assumption (2) is satisfied.

\subsubsection{Chaotic dynamics across the horizon}
The general principle applies also to the orbits $\Gamma_i,\ i=2,3,4$. Thus we expect that there exist chaotic orbits visiting the event horizons $r_\pm$ repeatedly in the region $r_-<r<r_+$, and chaotic orbits oscillating between $r_-$ and $r_I$. There may even be chaotic orbits visiting all the four radii $r_I,r_\pm,r_O$ repeatedly. Similar statements to Theorem \ref{ThmHorizon} can be formulated, which we skip.
In addition to the Melnikov condition, we also need to verify the boundary condition at the inner event horizon.
\subsubsection{Arnold diffusion near the horizon}
If we introduce the dependence on $\tau$ or also consider nonequatorial orbits in the perturbed Kerr spacetime, the resulting Hamiltonian system will have high enough dimensions to admit Arnold diffusion utilizing the hyperbolic fixed point and its associated homoclinic orbit at the event horizon in the $(r,\rho)$ coordinates. The analysis is completely analogous to Arnold diffusion utilizing the photon shell. We skip the statement since it is similar to Theorem \ref{ThmGLS}.  The diffusion orbit will perform zoom-whirl motion around the event horizon and will undergo a big oscillation in the energy $E$ or angular momentum $L_z$.

\appendix
\section{Poincar\'e's theorem, Smale Horseshoe and the Scattering Map}\label{AppPoincare}

In this section, we give an introduction to Poincar\'e's theorem. Historically, Poincar\'e discovered it when studying the nonintegrability of the Newtonian three-body problem. Later Smale discovered his horseshoe from Poincar\'e's picture and introduced the symbolic dynamics. We present Smale's result in Section \ref{SHorseShoe}. After that we generalize Poincar\'e's proof to the setting of normally hyperbolic invariant manifolds to give the proof of Theorem \ref{ThmScatteringS1} in Section \ref{SSScattering}.
\subsection{Poincar\'e's theorem}\label{SSPoincare}

%The next remark is that the invariant manifold theorem actually holds for maps. Namely, let $0$ be a hyperbolic fixed point for the map $\phi_0:\ \R^2\to \R^2$, i.e. $D\phi(0)$ has two eigenvalues satisfying $\lambda_1<1<\lambda_2$ (or equivalently, the linearized flow $X'=AX$ has $A$ with two eigenvalues with opposite signs). Th en it has stable manifold and unstable manifold $W^s=\{X\in \R^2 \ |\ \phi^n(X)\to 0,\ n\to\infty\}$ and unstable manifold $W^u=\{X\in \R^2 \ |\ \phi^n(X)\to 0,\ n\to-\infty\}$.

%Therefore, we will consider the phase space $((q,p)\in)\R^2$ and the maps $\phi^1_{H_0}$ and $\phi^1_{H_\eps}$, which are called time-1 maps. We can talk about invariant manifolds for both maps.

%Finally, we remark that the stable and unstable manifold of a hyperbolic fixed point for the time-1 map of a Hamiltonian (autonomous or nonautonomous) are Lagrangian submanifolds.
%\subsubsection{Poincar\'e's theorem}
We consider a Hamiltonian $H_0:\ \R^2\to \R$ satisfying the following assumptions:
\begin{enumerate}
\item there exists a critical point $O$ of $H_0$ at which the eigenvalues of the linearized Hamiltonian equations are real and nonzero. This means that the eigenvalues have opposite sign and same modulus since the flow preserves volume.
\item There is a homoclinic solution $\gamma:\ \R\to \R^2$ such that $\gamma(t)\to O,\ t\to\pm\infty$. Points in the graph of $\gamma$ is simultaneously on the stable and unstable manifolds of $O$. %Suppose also the graph can be expressed as $(q,p=\partial_qS_0)$ where $S_0$ solves the Hamilton-Jacobi equation $H_0(q,\partial_q S_0(q))=0. $
\end{enumerate}
%We will see from the proof that the setting can be easily generalized to the case of a single hyperbolic fixed point with a homoclinic orbit.
%In the case of $H_0:\  \R^2\to \R$, either the stable or the unstable manifold consists of a solution to the Hamiltonian equation. So the Hamilton-Jacobi equation is satisfied
%$$H_0(q, \partial_q S_0)=0,$$

%We remark that the same conclusion holds for a Hamiltonian $H_0:\ \R^{2n}\to \R$, which we do not discuss here.

We next add a small time dependent perturbation to $H_0$ to get the Hamiltonian $$H_\eps:\ \R^2\times \T^1\to \R,\ H_\eps(q,p)=H_0(q,p)+\eps H_1(q,p,t).$$ We denote by $\phi^t$ the flow generated by the Hamiltonians $H_0$. Since the Hamiltonian $H_0$ is autonomous, the flow $\phi^t,\ t\in \R,$ is a one-parameter diffeomorphism group, i.e. $\phi^{t+s}=\phi^t\phi^s,\ s,t\in \R$. However, the Hamiltonian $H_\eps$ is nonautonomous and the time dependence is 1-periodic. Instead, denoting by $\phi_\eps^1$ the time-1 map, we have $\{\phi_{\eps}^n,\ n\in \Z\}$ is a group, i.e. $\phi^{m+n}_\eps=\phi^{m}_\eps\phi^{n}_\eps,\ m,n\in \Z$.   %The phase portrait of $\phi^n_{\eps=0},\ n\in \Z,$ is the same as that of $\phi^t,\ t\in \R$ on $\R^2$.

In general, for $\eps>0$ the fixed point $O$ is perturbed to $O_\eps$ that is $O(\eps)$-close to $O$. Moreover the stable and unstable manifolds $W^s_\eps(O_\eps)$ and $W^u_\eps(O_\eps)$ (c.f. equation \eqref{EqWSU}) do not coincide and complicated dynamics is created. %The next theorem discovered by Poincar\'e gives a characterization of such phenomenon. This is the main mechanism responsible for the nonintegrability.

\begin{Thm}[Poincar\'e]\label{ThmPoincare} Suppose Melnikov function
$$\mathcal M(\nu):=\int_\R\{H_0,H_1\}(q(t+\nu),p(t+\nu),t)\,dt$$
has a nondegenerate zero.
Then for $\eps>0$ sufficiently small, the perturbed stable and unstable manifolds $W_\eps^u(O_\eps)$ and $W_\eps^s(O_\eps)$  intersect transversally.
\end{Thm}
\begin{proof}
Denote by $\Lambda=\{(q_0(\nu),p_0(\nu),s),\ \nu\in \R, s\in \T\}\subset \R^2\times \T$ the two dimensional unperturbed stable and unstable manifolds in the extended phase space. We pick a line $\ell_{\tau,s}$ perpendicular to $\Lambda$ at the point $(q_0(\nu),p_0(\nu),s)$ and denote by $z^{u,s}_\eps$ the intersection point of $W^{u/s}_\eps(O_\eps)$ with $\ell_{\tau,s}$. We have $\phi_\eps^t(z^u_\eps)\to O_\eps$ as $t\to-\infty$ and $\phi_\eps^t(z^s_\eps)\to O_\eps$ as $t\to+\infty$. To see that $W^{u/s}_\eps(O_\eps)$ do not coincide, we shall compare the difference $H_0(z^u_\eps)-H_0(z^s_\eps)$. Since $DH_0$ is nondegenerate near $\gamma(\nu)$, and $\ell_{\tau,s}$ is in the direction of $DH_0(\gamma(\nu))$, we see that when the difference $H_0(z^u_\eps)-H_0(z^s_\eps)$ is nonvanishing, then we have the stable and unstable manifolds $W^{u/s}_\eps(O_\eps)$ split.

We first  have
$$H_0(z^u_\eps)-H_0(\phi_\eps^{-T}z^u_\eps)=\eps \int_{-T}^0\{H_0,H_1\}(\phi_\eps^tz^u_\eps,s+t)\,dt.$$

Next, by the invariant manifold theory, the unstable manifold $W^u_\eps(O_\eps)$ of the perturbed system remains within $O(\eps)$ distance of the unperturbed one $W^u(O)$. Then we have that for each $T$, there is $T'$ such that $|\phi^{-T}_\eps z^u-(q_0(\nu-T'),q_0(\nu-T'))|=O(\eps)$. Moreover, when $T\to\infty$, we can also choose $T'\to\infty$. We next choose $T=c\log\eps^{-1}>0$ with small $c$. Then by Gronwall estimate, we that that $|\phi^{t}_\eps z^u-(q_0(\nu-t),q_0(\nu-t))|=O(\eps^{1-d_1})$ for some $0<d_1<1$ for all $t\in [-T,0]$.
Then  we have $|(q_0(\nu-T),q_0(\nu-T))-O|< O(\eps^{d_2})$ and $|(q_0(\nu-T'),q_0(\nu-T'))-O|< o_\eps(1)$ due to the exponential convergence.

Over the time interval $[-T,0]$, we replace the integral $ \int_{-T}^0\{H_0,H_1\}(\phi_\eps^tz^u_\eps,s+t)\,dt$ by $ \int_{-T}^0\{H_0,H_1\}(\gamma(\nu+t), s+t)\,dt$ creating an error $O(\eps \eps^{1-d_1}\log\eps)=o(\eps)$. Moreover, since $\gamma(t)\to O$ exponentially for both $t\to\pm\infty$, we get $\int_{-\infty}^{-T}\{H_0,H_1\}(\gamma(\nu+t), s+t)\,dt=O(\eps^{d_2})$.
We also replace $H_0(\phi_\eps^{-T}z^u_\eps)$ by $H_0(\gamma(-T'+\nu))$ creating an error estimated as $DH_0(\gamma(-T'+\nu))|\phi_\eps^{-T}z^u_\eps-\gamma(-T'+\nu)|=O(\eps^{1+d_2})$. So we get
$$H_0(z^u_\eps)-H_0(\gamma(-T'+\nu))=\eps \int_{-\infty}^0\{H_0,H_1\}(\gamma(\nu+t), s+t)\,dt+o(\eps). $$
Similarly, we get
$$H_0(\gamma(T''+\nu))-H_0(z^s_\eps)=\eps \int^{\infty}_0\{H_0,H_1\}(\gamma(\nu+t), s+t)\,dt+o(\eps). $$
Adding the two equations using the fact that $H_0$ is a constant along the orbit $\gamma$, we get
$$H_0(z^u_\eps)-H_0(z^s_\eps)=\eps \int_\R\{H_0,H_1\}(\gamma(\nu+t), s+t)\,dt+o(\eps). $$%$$H_0(O_\eps)-H_0(z^s_\eps)=\eps \int^{\infty}_0\{H_0,H_1\}(\gamma(\nu+t),s+t)\,dt+O(\eps^2)$$
In the above integral, we can fix $s=0$ section, then the resulting integral is the derivative of
$$\eps J(\nu)=\eps \int_{-\infty}^\infty H_1(\gamma(t+\nu),t)-H_1(O,t)\,dt.$$
Since we have $\gamma(t)\to O$ exponentially as $t\to\pm\infty$, we get that the integrals $J(\nu)$ as well as $J'(\nu)$ converge absolutely. The theorem is then proved by setting $ \mathcal M(\nu)=J'(\nu)$.
\end{proof}
\subsection{Smale Horseshoe}\label{SHorseShoe}
%In this section, we discuss briefly Smale's horseshoe.% and refer readers to \cite{KH} Chapter 6.5 for more details.

It was discovered by Smale that a separatrix splitting gives rise to a horseshoe. We refer readers to Chapter 5.8 of \cite{BS} and Chapter 6.5 of \cite{KH} for the definition of a horseshoe and the following theorem.
\begin{Thm}Let $p$ be a hyperbolic fixed point of a diffeomorphism $f: \ U(\subset \R^n)\to \R^n$ and $q$ be a transverse homoclinic point of $p$. Then for any $\eps>0$, there exists a horseshoe in the union of the $\eps$-neighborhoods of $p$ and $q$.
\end{Thm}
The existence of a horseshoe implies the existence of invariant set $\Lambda$ such that $f \Lambda=\Lambda$ and $f$ is conjugate to a two-sided shift in the space of bi-infinite sequences in the alphabet $\{0,1,\ldots,n-1\}$ for some $n\geq 2$.

\subsection{The scattering map, proof of Theorem \ref{ThmScatteringS1}}\label{SSScattering}
 \begin{proof}[Proof of Theorem \ref{ThmScatteringS1}]
 The proof is a generalization of the above proof of Poincar\'e's theorem.

 Let $H_0=H$ be the Schwarzchild Hamiltonian in \eqref{EqHamS} and denote $H_\eps=H_0+\eps H_1$ be the perturbed Hamiltonian. We introduce the auxiliary function $\mathcal H_\eps=\frac{r^2}{L_z}(H_0+\eps H_1+\frac12)$. If we use $\mathcal H_\eps$ as the new Hamiltonian, then its Hamiltonian flow on the level set $\mathcal H_\eps=0$ is a time reparametrization of that of $H_\eps$ on the energy level set $H_\eps=-\frac12$. We denote by $\lambda$ the time for the Hamiltonian $\mathcal H_\eps$, and by $t$ that of $H_\eps$, then we have $\frac{r^2}{L_z}d\lambda= dt$ for orbits with $\mathcal H_\eps=0$ and $H_\eps=-1/2$.

 We next denote $h_0(r,p_r)=\frac12\frac{r^2}{L_z}(-\al E^2+\al p_r^2+r^{-2}L^2+1)=\frac{r^2}{L_z}(H_0+\frac12)$ the Hamiltonian governing the radial motion and $\mathcal H_\eps=h_0+\eps \mathcal H_1$ with $\mathcal H_1=\frac{r^2}{L_z}H_1$.

Denote by $\Lambda$ the 5 dimensional homoclinic manifold  $\{(\Gamma(\nu),\Xi_0),\ \nu\in \R,\ \Xi_0\in T^*\mathbb S^2\}$ to the unperturbed NHIM $N$. Then we pick a line $\ell_{\nu,\Xi_0}$ normal to $\Lambda$ at the point $(\Gamma(\nu),\Xi_0)$ and parallel to $Dh_0(\Gamma(\nu),\Xi_0)$.  Let $z^{u/s}_\eps$ be the intersection of the stable and unstable manifolds $W^{u/s}(N_\eps)$ with $\ell_{\nu,\Xi_0}$. %Then there are points $x^{u/s}_\eps\in N_\eps$ such that we have $$\mathrm{dist}(\Phi^\lambda_\eps(z^u_\eps),\Phi^\lambda_\eps(x^u_\eps))\to 0\ \mathrm{as}\ \lambda\to -\infty,\quad \mathrm{dist}(\Phi^\lambda_\eps(z^s_\eps),\Phi^\lambda_\eps(x^s_\eps))\to 0\ \mathrm{as}\ \lambda\to +\infty. $$ where

We use $\Phi^\lambda_\eps$ the flow for the perturbed system $\mathcal H_\eps=h_0+\eps \mathcal H_1$ with time reparametrized by the Mino time. We use the difference $h_0(z^u_\eps)-h_0(z^s_\eps)$ as a measurement of the splitting of the stable and unstable manifolds. We thus have the following by Newton-Leibniz
%$$h_0(z^u_\eps)-h_0(\Phi^\lambda_\eps(z^u_\eps))=\eps \int_{\lambda}^0\{h_0,H_1\}(\Gamma(\nu+\mu), \Xi(\Xi_0,\mu))\,d\mu$$
$$h_0(z^u_\eps)-h_0(\Phi^{\mu}_\eps(z^u_\eps))=\eps \int_{\mu}^0\{h_0,\mathcal H_1\}(\Phi^\lambda_\eps(z^u))\,d\lambda.$$
%where the extra term $\frac{dt}{d\lambda}=\frac{r^2}{L_z}$ appears because we use the Mino time reparametrization.

By the smooth dependence on the base point of the stable and unstable fiber, we get that $|z^{u/s}_\eps-z_0|=O(\eps)$ and that for all $\lambda<0$, the orbit $\Phi^\lambda(z^{u})$ on the perturbed unstable fiber lies within $O(\eps)$ distance of the unperturbed unstable manifold $\{(\Gamma(\nu+\lambda),\Xi(\Xi_0,\lambda)),\ \lambda<0\}$. In particular, for each $\mu<0$, there exists $\mu'$ such that $|(\Gamma(\nu+\mu'),\Xi(\Xi_0,\mu'))-\Phi^\mu(z^{u})|= O(\eps)$, and we have $\mu'\to -\infty$ when $\mu\to-\infty$.  Next, by $\Gamma(\lambda)$'s exponential convergence to $\Gamma_c$, we get that $|\Gamma(\mu)-\Gamma_c|=O(\eps^{d_1})$ for $\mu=c\log\eps$, for some $0<d_1<1$ by choosing $c>0$ small and $|\Gamma(\mu')-\Gamma_c|=o_\eps(1)$. Also, by Gronwall estimate, we get that over the time interval $\lambda\in [c\log\eps,0]$, we have $|(\Gamma(\nu+\lambda),\Xi(\Xi_0,\lambda))-\Phi^\lambda(z^{u})|= O(\eps^{1-d_2})$ for some $0<d_2<1$ provided $c$ is small.

We next choose $\mu=c\log \eps$ and replace  the integral $\int_{\mu}^0\{h_0,\mathcal H_1\}(\Phi^\lambda_\eps(z^u))\,d\lambda$ by $\int_{\mu}^0\{h_0,\mathcal H_1\}(\Gamma(\nu+\lambda),\Xi(\Xi_0,\lambda))\,d\lambda$ creating an error $O(\eps^{2-d_2}\log\eps)=o(\eps)$. Next, we get $\int^{\mu}_{-\infty}\{h_0,\mathcal H_1\}(\Gamma(\nu+\lambda),\Xi(\Xi_0,\lambda))\,d\lambda=O(\eps^{1+d_1})$ by the exponential convergence of $\Gamma(\lambda)$ to $\Gamma_c$.
%Moreover, $\gamma$ enters
Moreover, we replace $h_0(\Phi^{\mu}_\eps(z^u_\eps))$ by $h_0(\Gamma(\nu+\mu'))$ creating an error estimated as
%&|(h_0(\Phi^{\mu}_\eps(z^u_\eps))-h_0(\gamma(\mu+\nu))|%+(H_0(\gamma(\mu+\nu)-H_0(\gamma_c(\mu+\nu)))\\
$|Dh_0(\gamma(\mu'+\nu))|(|\Phi^{\mu}_\eps(z^u_\eps))-\gamma(\mu'+\nu)|)=o(\eps). $

We thus obtain
$$h_0(z^u_\eps)-h_0(\gamma(\mu'+\nu))=\eps \int_{-\infty}^0\{h_0,\mathcal H_1\}(\Gamma(\nu+\lambda),\Xi(\Xi_0,\lambda))\,d\lambda+o(\eps).$$
Similarly, we get
$$h_0(\gamma(-\mu''+\nu))-h_0(z^s_\eps)=\eps \int^{\infty}_0\{h_0,\mathcal H_1\}(\Gamma(\nu+\lambda),\Xi(\Xi_0,\lambda))\,d\lambda+o(\eps).$$
Adding the expressions, we get
$$h_0(z^u_\eps)-h_0(z^s_\eps)=\eps \int_\R\{h_0,\mathcal H_1\}(\Gamma(\nu+\lambda),\Xi(\Xi_0,\lambda))\,d\lambda+o(\eps).$$
When the integral on the RHS has a nondegenerate zero, then the stable and unstable manifolds $W^{u/s}(N_\eps)$ intersect transversally. %Since Graph$(q,\partial_q S^\pm)$ is the stable and unstable manifolds, we get that to the leading order, $$\partial_\nu J(\nu,\Xi_0)=(\partial_rS^+(r(\nu),\Xi_0)-\partial_rS^-(r(\nu),\Xi_0))\dot r(\nu)\neq 0$$ implies that the stable and unstable manifolds $W^s(N_\eps)$ and $W^u(N_\eps)$ do not coincide seen from the $(r,p_r)$-plane.

We next prove the formula for the scattering map. Let $\tilde w$ be an intersection point of $W^u(N_\eps)$ and $W^s(N_\eps)$ and $\Phi_\eps^\lambda$ the perturbed flow, then by the theorem of NHIM, we find points $w^\pm\in N_\eps$ such that
\begin{equation*}\label{Eqzpm}
|\Phi_\eps^\lambda (w^\pm)-\Phi_\eps^\lambda (\tilde w)|\to 0,\quad  t\to\pm\infty.
\end{equation*} The scattering map thus sends $w^-$ to $w^+$ by definition.

We have the following equality by Newton-Leibniz
$$L_z(\tilde w)-L_z(\Phi^\mu_\eps\tilde w)=\eps\int_{\mu}^0\{L_z,H_1\}\frac{dt}{d\lambda}(\Phi^\lambda_\eps \tilde w)\,d\lambda,$$
$$L_z(w^-)-L_z(\Phi^\mu_\eps w^-)=\eps\int_{\mu}^0\{L_z,H_1\}\frac{dt}{d\lambda}(\Phi^\lambda_\eps w^-)\,d\lambda.$$
 Using the same argument as above, we replace the integrals on the RHS by integrals along the unperturbed flows $(\Gamma(\nu+\lambda), \Xi(\Xi_0,\lambda))$ and $(\Gamma_c(\nu+\lambda), \Xi(\Xi_0,\lambda))$ respectively.  To approximate $\Phi^\lambda_\eps w^-$ by $(\Gamma_c(\nu+\lambda), \Xi(\Xi_0,\lambda))$, we apply Gronwall over the time interval $[\mu,0]$ where $\mu=c\log\eps$ to get the estimate of the error $\eps^{1-d_2}$ for some $0<d_2<1$ if $c>0$ is chosen sufficiently small.

 Since we have $|\Phi^\mu_\eps w^--\Phi^\mu_\eps\tilde w|\to 0$ exponentially, we get by taking difference
$$L_z(\tilde w)-L_z(w^-)=\eps\int_{-\infty}^0\left[\partial_{\varphi}H_1 \frac{dt}{d\lambda}(\Gamma(\nu+\lambda), \Xi(\Xi_0,\lambda))-\partial_{\varphi}H_1 \frac{dt}{d\lambda}(\Gamma_c, \Xi(\Xi_0,\lambda))\right]\,d\lambda+o(\eps),$$
Similarly, we get
$$L_z(w^+)-L_z(\tilde w)=\eps\int^{\infty}_0\left[\partial_{\varphi}H_1 \frac{dt}{d\lambda}(\Gamma(\nu+\lambda), \Xi(\Xi_0,\lambda))-\partial_{\varphi}H_1 \frac{dt}{d\lambda}(\Gamma_c, \Xi(\Xi_0,\lambda))\right]\,d\lambda+o(\eps).$$
Adding up the last two equations, we get
$$L_z(w^+)-L_z( w^-)=\eps\int_\R\left[\partial_{\varphi}H_1 \frac{dt}{d\lambda}(\Gamma(\nu+\lambda), \Xi(\Xi_0,\lambda))-\partial_{\varphi}H_1 \frac{dt}{d\lambda}(\Gamma_c, \Xi(\Xi_0,\lambda))\right]\,d\lambda+o(\eps).$$
%Letting $\lambda\to\pm\infty$ and taking difference we get
%\begin{equation*}
%p_\star^+(x^+)-p_\star^-(x^-)=\int_\R\partial_\star D(S_0+\eps S^\pm)(\tilde x(\lambda))\frac{d}{d\lambda} \tilde x(\lambda)d\lambda. \end{equation*}
%We denote by $x$ the projection of $w$ to the $(r,\theta,\varphi)$-coordinates, and similarly for $x^\pm$.
This completes the proof.
\end{proof}
\begin{Rk}\label{RkK}
\begin{enumerate}
\item We do not have an expression for the scattering map in the $(\theta,p_\theta)$-components, since $p_\theta$ is not a conserved quantity for the unperturbed system.  If we introduce action-angle coordinates $(\al_\theta, J_\theta)$ in the $(\theta,p_\theta)$ components, then we have similar expression for the scattering map.
\item The proof of Theorem \ref{ThmScatteringK} is completely analogous. If the Kerr Hamiltonian is perturbed to $H=H_0+\eps H_1$. We introduce (see \eqref{EqHamK}) \begin{equation}\label{EqHamKR}\mathcal H_\eps=H_r+H_\theta+\Sigma+2\eps \Sigma H_1=\Delta p_r^2-\frac{R}{\Delta}+Q+2\eps \Sigma H_1.\end{equation} Then we choose $h_0=\Delta p_r^2-\frac{R}{\Delta}+Q$ and $\mathcal H_1=2  \Sigma H_1$ in the above proof, then the proof goes through for the Kerr case.
\end{enumerate}
%We treat $\mathcal H$ as a Hamiltonian with time variable $\lambda$. Then the Hamiltonian flow of $\mathcal H$ on energy level $\mathcal H=0$ coincide with that of $H$ on the energy level $-\frac12$. Note that $\mathcal H= \Delta p_r^2-\frac{R}{\Delta}+Q$, hence the radial and latitudinal motions separate.
\end{Rk}
\section{The theorem of normally hyperbolic invariant manifolds}\label{SNHIM}
In this section we give the version of normally hyperbolic invariant manifold theorem that we used in the proof of our main theorem. The standard references are \cite{HPS, F}. %Readers are also referred to \cite{DLS08}.
\begin{Def}[NHIM]\label{DefNHIM} Let $N\subset M$ be a submanifold $($maybe noncompact$)$ invariant under $f$, $f(N) = N$. We say that $N$ is a normally hyperbolic invariant manifold if there exist a constant $C > 0$, rates $0 <\lb <\mu^{-1} < 1$ and an invariant $($under $Df)$ splitting for every $x \in N$
\[T_x M = E^s(x) \oplus E^u(x) \oplus T_xN
\]
in such a way that
\begin{equation}\nonumber
\begin{aligned}
v\in E^s(x)\quad &\Leftrightarrow\quad  |Df^n(x)v| \leq C\lb^n |v|, \quad n \geq 0,\\
v\in E^u(x)\quad &\Leftrightarrow \quad |Df^n(x)v| \leq C\lb^{|n|} |v|, \ \ n \leq 0,\\
v\in T_xN\quad &\Leftrightarrow \quad |Df^n(x)v| \leq C\mu^n |v|, \quad n \in \Z.
\end{aligned}
\end{equation}
Here the Riemannian metric $|\cdot|$ can be any prescribed one, which may change the constant $C$ but not $\lambda,\mu$.
\end{Def}

\begin{Thm}\label{ThmNHIM}Suppose $N$ is a NHIM under the $C^r$, $r>1$, diffeomorphism $f:\ M\to M$. Denote $\ell=\min\{r,\frac{|\ln\lambda|}{|\ln\mu|}\}$.
Then for any $C^r$ $f_\epsilon$ that is sufficiently close to $f$ in the $C^1$ norm,
\begin{enumerate}
\item  there exists a NHIM $N_\epsilon$ that is a $C^\ell$ graph over $N$,
\item$(${\it Invariant splitting}$)$ There exists  a splitting for $x\in N_\epsilon$
\begin{equation*}\label{EqSplitting}T_x M=E^u_{\epsilon}(x)\oplus E^{s}_{\epsilon}(x)\oplus T_xN_\epsilon
\end{equation*}
invariant under the map $f_\epsilon$.  The bundle $E^{u,s}_{\epsilon}(x)$ is $C^{\ell-1}$ in $x$.
\item There exist $C^\ell$ stable and unstable manifolds $W^s(N_\epsilon)$ and $W^u(N_\epsilon)$ that are invariant under $f$ and are tangent to $E_\epsilon^s\oplus TN_\epsilon$ and $E_\epsilon^u\oplus TN_\epsilon$ respectively.
\item  The stable and unstable manifolds $W^{u,s}(N_\epsilon)$ are fibered by the corresponding stable and unstable leaves $W^{u,s}_{x,\epsilon}:=\{y\ | d(f_\eps^n(y),f_\eps^n(x))\to 0, n\to +\infty\ \mathrm{for}\ s, n\to -\infty\ \mathrm{for}\ s\}$: $$W^{u}(N_\epsilon)=\cup_{x\in N_\epsilon} W_{x,\epsilon}^u,\quad W^{s}(N_\epsilon)=\cup_{x\in N_\epsilon} W_{x,\epsilon}^s.$$ %where the stable manifold $W_{x,\epsilon}^s:=\{y \ |\ d(f_\epsilon^n(x),f_\epsilon^n(y))\leq C_{\dt,x,y}(\lb+\dt)^n,\quad n\geq 0\}$
%and unstable manifold $W_{x,\epsilon}^u:=\{y \ |\ d(f_\epsilon^n(x),f_\epsilon^n(y))\leq C_{\dt,x,y}(\lb+\dt)^n,\quad n\leq 0\}$ are $C^r$ and $T_x W_{x,\epsilon}^{s,u}=E^{s,u}_{x,\epsilon}$ for $x\in N_\epsilon$ for any small $\dt>0$. %And the $W^{s,u}_\Pi:=\cup_{x\in \Pi} W^{s,u}_x$ is $C^{\ell-1}$ manifolds.
\item The maps $x\mapsto W_{x,\epsilon}^{u,s}$ are $C^{\ell-j}$ when $W_{x,\epsilon}^{u,s}$ is given $C^j$ topology.
\item \ $($\cite{DLS08}$)$If $f$ and $f_\epsilon$ are Hamiltonian, then $N_\epsilon$ is symplectic and the map $f_\epsilon$ restricted to $N_\epsilon$ is also Hamiltonian.
\end{enumerate}
\end{Thm}
\section{Perturbing the Schwarzschild metric}\label{SReggeWheeler}

In this section, we verify the nondegeneracy condition in the above Theorem \ref{ThmGLS} and \ref{ThmOsc} for the perturbed Schwarzschild metric which satisfies the linearized Einstein field equation in the vacuum.
\subsection{Regge-Wheeler equation}
We first summarize the result from \cite{RW}.
%{\bf Remember to add something from Chandrasekhar.}
%In the following we first outline the steps in the proof of the above result.
Suppose the background metric $g_{\mu\nu}$ is perturbed by $h_{\mu\nu}$. Then the perturbed Ricci curvature has the form $R_{\mu\nu}+\dt R_{\mu\nu}$. The vacuum Einstein equation requires that $\dt R_{\mu\nu}=-\dt \Gamma^\beta_{\mu\nu;\beta}+\dt \Gamma^\beta_{\mu\beta;\nu}=0,$ where ``$;$" means the covariant derivative and the linearized Christoffel symbol has the form
$$\dt \Gamma^\beta_{\mu\nu}=\dfrac12g^{\al\nu}(h_{\beta\nu;\gamma}+h_{\gamma\nu;\beta}-h_{\beta\gamma;\nu}).$$
So the equation $\dt R_{\mu\nu}=0$ is linear in $h$.

We then expand $h$ into Fourier series in terms of spherical harmonic functions $$h_{\mu\nu}=\sum \tilde h_{\mu\nu,\ell}^{m}(T, r)Y_{\ell}^m(\theta,\phi).$$
Here $Y_{\ell}^m$ is the spherical harmonic function defined as
$$Y_{\ell}^m(\theta,\phi)=\sqrt{\frac{2\ell+1}{4\pi} \frac{(\ell-m)!}{(\ell+m)!}} P_\ell^m(\cos\theta) e^{im\phi} $$
and $P_\ell^m$ is the associated Legendre polynomial
$$P_\ell^m(x)=(-1)^m(1-x^2)^{m/2}\frac{d^m}{dx^m}P_\ell(x)=\frac{(-1)^m}{2^\ell \ell!}(1-x^2)^{m/2}\frac{d^{\ell+m}}{dx^{\ell+m}}(x^2-1)^\ell$$
 where $P_\ell$ is the unassociated Legendre polynomial. We list a few to be used later
 \begin{equation}
 \begin{aligned}
 P_0&=1,\quad P_1=x,\quad P_2=\frac12(3x^2-1),\\
 P_n^0&=P_n,\quad
 P_1^1=(1-x^2)^{1/2},\\
 P_2^1&=-3x(1-x^2)^{1/2},\quad P_2^2=3(1-x^2).
 \end{aligned}
 \end{equation}
 The linearity of the equation allows us to consider only each individual Fourier mode $k,\ell$. Moreover, we also assume that the $t$-dependence is quantized, i.e. $h$ is periodic with respect to $t$ with frequency $\omega=k c$, $\omega t=kT$. Then we further define $$\tilde h_{\mu\nu,\ell}^{m}(T, r)=\sum_{k}\hat h_{\mu\nu,\ell}^{m}(r) e^{ikT}. $$
There are two cases: the odd parity and even parity cases. %The analysis of the odd parity case leads to the Regge-Wheeler equation whose exact solutions are given explicitly by Fiziev. However, these solutions are rather complicated.

 After modulo a gauge transformation, the  case of odd parity has the general form
$$h_{\mu\nu}=\left[\begin{array}{cccc}
0&0&-h_0\frac{1}{\sin\theta}\partial_\varphi Y_\ell^m&h_0\sin\theta\partial_\theta Y_\ell^m\\
0&0&-h_1\frac{1}{\sin\theta}\partial_\varphi Y_\ell^m&h_1\sin\theta\partial_\theta Y_\ell^m\\
*&*&0&0\\
*&*&0&0
\end{array}\right] ,$$
where $h_0=h_0(\tau,r)$ and $h_1=h_1(\tau,r)$ are periodic in $\tau$. In both the odd and even parity cases, the radial functions are independent of $m=-\ell,-(\ell-1),\ldots, \ell-1,\ell$, for each given $\omega$ and $\ell=0,1\ldots$. %Where $P_L$ is the $L$-th Legendre polynomial $L=0,1,2,\ldots.
Moreover, there are equations that should be satisfied by the radial functions. %Here we only provide the equation in the odd case and refer readers to \cite{RW} for the even case since the latter is slightly more complicated. Defining
$$Q(r)=\frac{1}{r}\al h_1(r), \quad \al^{-1} k h_0+\frac{d}{dr} (rQ)=0.$$
Then $Q$ satisfies the equation called Regge-Wheeler equation
$\frac{d^2Q}{(dr^*)^2}+k_{\mathrm{eff}}^2(r)Q=0, $
where $r^*=r+2M(\ln(\frac{r}{2M}-1))$ and
$$k_{\mathrm{eff}}^2=k^2-L(L+1)\frac{\al}{r^2}+ 6M \al\frac{1}{r^3}=k^2+\frac{1}{r^2}\al \left(-L(L+1)+ \frac{6M }{r}\right). $$

We next consider the case of even parity. The general even parity perturbation has the form after the gauge transformation
$$h_{\mu\nu}=\left[\begin{array}{cccc}
\al HY_\ell^m&H_1 Y_\ell^m&0 &0\\
H_1 Y_\ell^m&\al^{-1} HY_\ell^m& 0&0\\
0&0 & r^2KY_\ell^m& 0\\
0&0&0& r^2K\sin^2\theta Y_\ell^m
\end{array}\right]$$
where $H,H_1,  K $ are functions of $\tau$ and $r$ and the dependence on $\tau$ is $1$-periodic. Moreover, if we consider only the stationary case $k=0$, then $H_1=0$.
%There are also equations satisfied by the radial functions $ H,K$ as follows.

%Here we look for a solution in the stationary case $k=0$ where we have $H_1=0$ and t

The functions $K$ and $H$ solve  the equations   \begin{equation*}
 \begin{cases}
&(K+H)'+\frac{2M}{r^2}\al^{-1} H=0,\\
&H'+\frac{2M}{r^2}\al^{-1} H+\frac{2H}{r}+\frac{1}{2M}(L-1)(L+2)(K+H)=0.
 \end{cases}
 \end{equation*}

 The solution for $\ell=0$ is trivial. For $\ell=1$, the solution is given explicitly as
\begin{equation}\label{Eql=1}H=-\frac{2M}{r^2\al},\quad K=\frac{1}{M}\left(\frac{2M}{r}+\ln (1-\frac{2M}{r})\right). \end{equation}
 For $\ell>1$, we eliminate $K+H$ to obtain an equation for $A=r(r-2M) H$
 $$A''-\left(\frac1r+\frac{1}{r-2M}\right)A'-(\ell-1)(\ell+2)\frac{A}{r(r-2M)}=0. $$

The general solution can be written as
$$A(r)=a P_\ell^{(2)}(1-\frac{r}{M})+b Q_\ell^{(2)}(1-\frac{r}{M}),\quad a,b\in \R,$$
where $P_\ell^{(2)},\ Q_\ell^{(2)}$ are associated Legendre polynomials of the first and second kinds respectively.

 We take $a=M/3,b=0,$ and $\ell=2$ to yield the solution
 \begin{equation}\label{Eql>1}H(r)=\frac{1-(1-r/M)^2}{r(r-2M)}M=-\frac{1}{M},\quad K(r)=\frac1M+\frac{1}{2r}+\frac{1}{2(r-2M)}\end{equation}
 with  $Y_{2}^2=\frac34\sqrt{\frac{5}{6\pi}}\sin^2\theta e^{2i\varphi}$.

 \subsection{Checking the Melnikov condition for Arnold diffusion}\label{AppMelnikovAD}
  In this section, we verify the nonvanishing condition of the Melnikov function for the above constructed perturbations. %People usually work in the non-autonomous setting and assuming that the unperturbed part is written in action-angle coordinates. In this section, we formulate the Melnikov nondegeneracy condition for the autonomous Hamiltonian and in a coordinate system that is far from being action-angle coordinate system, in particular, the dynamics is not separable and it becomes separable after a time change which breaks the Hamiltonian structure.
 Suppose the perturbed metric is
 $$\tilde g_{\mu\nu}=g_{\mu\nu}+\eps h_{\mu\nu}=\al(-1+\eps \tilde H)d\tau^2+\al^{-1}(1+\eps \tilde H)dr^2+r^2(1+\eps \tilde K)(d\theta^2+\sin^2\theta d\varphi^2). $$

 We treat $\frac{1}{2}\tilde g_{\mu\nu}$ as the perturbed Lagrangian $L_\eps$. Performing a formal Legendre transform
We obtain the perturbed Hamiltonian
\begin{equation*}
\begin{aligned}
2H_\eps&=\al^{-1}(-1+\eps \tilde H)^{-1}p_\tau^2+\al(1+\eps \tilde H)^{-1}p_r^2+r^{-2}(1+\eps \tilde K)^{-1}(p_\theta^2+\sin^{-2}\theta p_\varphi^2)\\
&=2H_0-\eps\left(\al^{-1} \tilde Hp_\tau^2+\al \tilde Hp_r^2+r^{-2}\tilde K(p_\theta^2+\sin^{-2}\theta p_\varphi^2)\right)+O(\eps^2)\\
&:=2H_0+2\eps H_1+O(\eps^2) .
\end{aligned}
\end{equation*}
%The dependence of $\tilde H$ and $\tilde K$ on $\theta$ and $\varphi$ makes the total angular momentum and $L_z$ no longer conserved.

Next, we show that an explicit perturbation satisfies the assumption of Theorem \ref{ThmGLS}.
\begin{proof}[Checking the assumption of Theorem \ref{ThmGLS}]
We choose $H,K$ as in \eqref{Eql>1} and denote $\tilde H=HY_2^2,\ \tilde K=K Y_2^2.$
 We have
 $$H_1=e^{2i\varphi }\sin^2\theta\left(-\al^{-1} \frac{E^2}{M}-\al \frac{p_r^2}{M}+r^{-2} KL^2\right).$$
 By the total energy conservation along the homoclinic orbit in the unperturbed system, we have
 $$-1=2H=-\al^{-1} E^2+\al p_r^2+r^{-2} L^2,\quad -1=-\al^{-1}(r_-) E^2+r_-^{-2} L^2.$$

Thus we simplify the expression of $H_1$ as follows by eliminating $p_r$
$$H_1=e^{2i\varphi }\sin^2\theta\left(\frac1M-2 \frac{E^2}{\al M}+r^{-2} (K+\frac1M)L^2\right).$$
 We get the Melnikov potential from \eqref{EqMelnikov}
 \begin{equation}\label{EqJS}
 \begin{aligned}
 J(\nu)&=\frac{1}{L} \int_\R \sin^2\theta(\lambda-\nu) e^{2i(\varphi(0)+\frac{\lambda-\nu}{L^2} )}\left(-\frac{2E^2}{M}\left( \frac{r(\lambda)^2}{\al(r(\lambda)) }- \frac{r_-^2}{\al(r_-) }\right)\right.\\
 &\left. +\frac1M\left(r(\lambda)^2-r_-^2\right)+L^2 \left(\frac{1}{r(\lambda)}+\frac{1}{r(\lambda)-2M}-\frac{1}{r_-}-\frac{1}{r_--2M}\right)\right)\,d\lambda.
 \end{aligned}
 \end{equation}
 where $\sin^2\theta$ can be expressed using \eqref{EqVertical} and $r=\frac{1}{u}$ using \eqref{EqRadial}. We rewrite the above $J(\nu)$ as
 $$J(\nu)=\frac{1}{L} \int_\R  \left(1-\left(\cos\theta(0)+\frac{\sqrt{b^2-1}}{b}\sin(b(\lambda-\nu))\right)^2\right)e^{2i(\varphi(0)+\frac{\lambda-\nu}{L^2} )}Q(r(\lambda))\,d\lambda,$$
 where we use $Q(r(\lambda))$ to denote the $r$-dependent term in the integrand of \eqref{EqJS}.  When we complexify $\lambda$, the only poles of the integrand are poles of $r(\lambda), r(\lambda)^{-1}$ and $(r(\lambda)-2M)^{-1}$ where $$r(\lambda)=\left(u_1-a^2\mathrm{sech}^2(\frac{a\lambda}{2})\right)^{-1},\ \lambda\in \C.$$

    We introduce our initial contour of integration as $C=C_-\cup C_+$ where $C_+=\{\mathrm{Im}z=\frac{i\pi}{a}\}$ oriented from left to right and $C_-=\{\mathrm{Im}z=-\frac{i\pi}{a}\}$ oriented from right to left.  To simplify the calculation, we consider first  the  extremal case $b^2\to 1$ so that $\sin^2\theta(\lambda-\nu)\to \sin^2\theta(0)$ and $\theta(0)\to \pi/2$. Then, we obtain the following contour integral
 $$J(\nu)=\frac{1}{L}e^{2i(\varphi(0)-\frac{\nu}{L^2}) } \frac{-2}{\sinh \frac{2\pi}{aL^2}}\oint_{C}e^{2i\frac{z}{L^2} }Q(r(z))\,dz.$$
We note that the general $b^2\neq 1$ case can be done similarly as follows with a more complicated factor in front of the contour integral.

We next evaluate the contour integral. To simplify the discussion, we normalize $M=1$ and also use an observation of \cite{Mo} by examining $a\to 0$ (we remark that the general $a\neq 0$ can be done directly with some more labors). This means the roots $u_2<u_1$ of the RHS of the $\frac{du}{d\lambda}$-equation in \eqref{EqHomo} nearly coincide (we have $a^2=u_1-u_2$). This means that $L^2$ is close to $L^2_{isco}=12$, $E^2$ is close to $E_{isco}^2=\frac{8}{9}$ and $u_1$ is close to $1/3$. Setting $z=\lambda-\frac{i\pi}{a}$, then we get
 $$u(z)=u_1+a^2\frac{1}{\sinh^2(a z/2)}\to \frac{1}{3}+\frac{4}{z^2}=\frac{z^2+12}{3z^2},\quad \mathrm{as}\ a\to 0.$$

 $Q(r)$ is a linear combination of terms $1,r^2, \frac{r^3}{r-2},r^{-1}$ and $\frac{1}{r-2}$. Note that $1, r^2, r^{-1}$ and $r^3$ do not contribute to the residue and we have in the limit $a\to 0$ that $\frac{1}{r-2}=\frac{z^2+12}{z^2-24}$ has residue $\pm \frac{9}{\sqrt 6}$ at $z=\pm \sqrt{24}$ and $\frac{r^3}{r-2}=\frac{27 z^6}{(z^2+12)^2}\frac{1}{z^2-24}$ has residue $\pm 9\sqrt 6$ at $z=\pm \sqrt{24}$.

  We then deform the contour $C$ to another bounded closed curve $\tilde C$ enclosing the poles. Since $u$ and the poles depend on $a$ continuously, for small $a\ne0$, the poles are still enclosed by $\tilde C$ and the residues also depends on $a$ continuously. Our goal is then to show that the contour integral $\frac{1}{2\pi i} \oint_{\tilde C}e^{2i\frac{z}{L^2} }Q(r(z))\,dz$ is nonvanishing for $a=0$. Indeed after  some calculations (substituting $Q(r)$ with $M=1,E^2=\frac{8}{9},L^2=12$), we have that sum of the residues is
  $$ e^{4i\frac{\sqrt{6}}{12} }(2 \sqrt6)- e^{-4i\frac{\sqrt{6}}{12} }(2 \sqrt6)\neq 0.$$
This completes the proof.
\end{proof}

\subsection{Checking the Melnikov condition for oscillatory motion}\label{AppOsc}

%For the purpose of constructing oscillatory orbits, the scattering map has the same expression as \eqref{EqScattering}, but we need to redefine the orbits $R(\al+\lambda)$.
In this section, we show that an explicit perturbation satisfies the assumption of Theorem \ref{ThmOsc}.
\begin{proof}[Checking the assumption of Theorem \ref{ThmOsc}]
We choose $H,K$ as in \eqref{Eql=1} and denote $\tilde H=HY_1^1,\ \tilde K=K Y_1^1$ where $Y^1_1=\sin\theta e^{i\varphi}$ up to a scalar multiple.  We next show that the perturbation chosen in this way satisfies the assumptions of Theorem \ref{ThmOsc}.

We first verify assumption (1). This follows immediately from the expressions of $\tilde H$ and $\tilde K$. We next verify assumption (2). From the equations $$\dot\theta=\frac{\partial H_\eps}{\partial p_\theta}=r^{-2}(1+\eps \tilde K)^{-1}p_\theta,\quad \dot p_\theta=-\frac{\partial H_\eps}{\partial \theta},$$
we see that $\theta=\pi/2,p_\theta=0$ is a solution. Indeed, we have $\dot p_\theta=0$ if $\theta=\pi/2$ since $H_\eps$ depends on $\theta$ through $\sin\theta$ only so the $\frac{\partial H_\eps}{\partial \theta}$ has a factor $\cos\theta$ which vanishes when $\theta=\pi/2$. On the other hand when $p_\theta=0$ we also have $\dot \theta=0$. We next work on assumption (3). For this purpose we write down the expression for the perturbing Hamiltonian function $H_1$ first (using energy conservation as along the unperturbed homoclinic orbit to infinity)
$$H_1=e^{i\varphi}(-H+2\al^{-1} HE^2+r^{-2}(K-H)L^2).$$
Since we have the expression of the homoclinic orbit parametrized by $\lambda$ and $d\lambda=\frac{L}{r^2}dt$. We also have $r(t)\sim |t|^{2/3}$, which shows that $\lambda$ is actually defined on a bounded time interval $(a,b)$ to be determined more precicely below.  Then recalling that $\lim_{r\to\infty} r^2(-H+2\al^{-1} HE^2)=-2M(1-2E^2)$, we get the Melnikov potential
$$J(\nu)=\frac{1}{L} \int_a^b e^{i(\varphi_0+\frac{\lambda-\nu}{L^2} )}\left(r^2(-H(r)+2\al^{-1} H(r)E^2)+2M(1-2E^2)+(K(r)-H(r))L^2)\right)\,d\lambda$$
where $r$ is a function of $\lambda$.

To evaluate the function, we next work out $r(\lambda)$ in the range of large $L^2$. Then using the definitions in Section \ref{SSSOscHomo}, we find that $u_1u_2=\frac{2}{\ell}(\frac{1}{2M}-\frac{2}{\ell})=\frac{1}{L^2}$ is small. Then $\mu=\frac{M}{\ell}$ and $k=\frac{4\mu}{1-4\mu}$ are also small. For sufficiently small $k$, the elliptic integrals can be approximated as follows
\begin{equation*}
\begin{aligned}
F(\varphi,k)&=\int_0^\varphi\frac{d\theta}{\sqrt{1-k^2\sin^2\theta}}=\int_0^\varphi(1+\frac{k^2}{2}\sin^2\theta+O(k^4))\,d\theta\\
&=(1+\frac{k^2}{2})\varphi-\frac{k^2}{4}\sin 2\varphi+O(k^4),\\
K(k)&=(1+\frac{k^2}{2})\frac\pi2+O(k^4).\\
\lambda(\chi)&=\frac{2}{1-4\mu}\left[\frac12(1+\frac{k^2}{2})\chi-\frac{k^2}{4}\sin \chi+O(k^4)\right].
\end{aligned}
\end{equation*}
We can then find $\chi$ as a function of $\lambda$ as
$$\chi=c\lambda+\frac{k^2}{2+k^2}\sin (c\lambda)+O(k^4),$$
where $c=\frac{1-4\mu}{1+\frac{k^2}{2}}=1-4\mu+O(\mu^2)$ and the domain for $\lambda$ is $[a,b]=c^{-1}[-\pi+O(k^4),\pi+O(k^4)]$ whose endpoints correspond to $\chi=\pm\pi$ where $u=0$.
Then we get $$r(\lambda)=\frac{1}{u}=\frac{\ell}{1+\cos\chi}=\frac{\ell}{1+\cos\chi}=\frac{\ell}{1+\cos(c\lambda)-\frac{k^2}{2+k^2}\sin^2 (c\lambda)+O(k^4)}.$$

In \eqref{Eql=1}, we have the expansion for large $\ell$ and large $r$
\begin{equation*}
\begin{aligned}
K&=-\frac{8M}{r^2}+O(r^{-3})\\
 r^2(-H(r)+2\al^{-1} H(r)E^2)&=-2M(1-2E^2)-4M^2(1-2E^2)^2/r+O(r^{-2}).
\end{aligned}
\end{equation*}

Now the integral becomes (using $\ell^{-1}=O(L^{-2})=O(k)$)
\begin{equation*}
\begin{aligned}
J(\nu)&= \frac{-4M^2(1-2E^2)^2}{L} e^{i(\varphi_0+\frac{-\nu}{L^2})}\int_a^b e^{i\frac{\lambda}{L^2}} \left(\frac{1}{r}+O(r^{-2})\right)\,d\lambda\\
&=\frac{-4M^2(1-2E^2)^2}{L} e^{i(\varphi_0+\frac{-\nu}{L^2})}\frac{1}{\ell}\int_a^b e^{i\frac{\lambda}{L^2}} \left(1+\cos(c\lambda)\right)\,d\lambda+O(L^{-5}).\\
\end{aligned}
\end{equation*}
Letting $L\to \infty$, the integral becomes $\int_{-\pi}^{\pi}(1+\cos\lambda)\,d\lambda=2\pi$, so for large $L$, it is clear that the integral does not vanish.
\end{proof}

\section*{Acknowledgment}
The author would like to express his deep gratitudes to Professor S.-T. Yau for suggesting the topic. He also would like to thank Mr. Yifan Guo for drawing all the figures. The author is supported by NSFC (Significant project No.11790273) in China and Beijing Natural Science Foundation (Z180003).

\end{document}